\documentclass{lmcs} 
\pdfoutput=1

\usepackage[utf8]{inputenc}

\usepackage{lastpage}
\lmcsdoi{19}{1}{6}
\lmcsheading{}{\pageref{LastPage}}{}{}%
{Dec.~17,~2020}{Jan.~18,~2023}{}

\usepackage[ruled]{algorithm2e}

\usepackage{amsthm}
\theoremstyle{definition}
\newtheorem{example}[thm]{Example}

\usepackage{tabularx}
\usepackage{booktabs}
\usepackage{amsmath}
\usepackage{amssymb}
\usepackage{graphicx}
\usepackage{epsfig}
\usepackage{verbatim}
\usepackage{float}
\usepackage{datetime}
\usepackage{listings}
\usepackage{subfig}
\usepackage{stmaryrd}
\usepackage{latexsym}
\usepackage{amsfonts,amstext,wasysym}
\usepackage{url}
\usepackage{mathpartir}
\usepackage{enumitem}
\usepackage[]{todonotes} 
\usepackage{fancyvrb}
\usepackage{multirow}
\usepackage{tikz}
\usetikzlibrary{arrows,shapes.geometric,positioning,automata,patterns}
\usetikzlibrary{trees,decorations.pathmorphing,calc}
\usetikzlibrary{intersections}
\usetikzlibrary{decorations.pathreplacing,external,backgrounds}

\usepackage{xifthen}
\usepackage{xargs}

\newenvironment{repthm}[2][]{\par\noindent\textbf{Restatement of Theorem \ref{#2}} (#1).\itshape}{}
\newenvironment{replem}[2][]{\par\noindent\textbf{Restatement of Lemma \ref{#2}} (#1).\itshape}{}

\definecolor{dark-gray}{gray}{0}

\usepackage{float}
\usepackage{afterpage}

\newfloat{listing}{tbp}{lol}
\floatname{listing}{Listing}
\floatstyle{boxed}
\restylefloat{listing}

\newif\ifpgf\pgftrue  

\lstdefinestyle{mystyle}{flexiblecolumns=true,showstringspaces=false,keepspaces=true,basewidth={0em,0em},basicstyle=\sffamily,commentstyle=\itshape,stringstyle=\sffamily}

\definecolor{diffhilight}{rgb}{0.6, 0.81, 0.93}
\definecolor{darkgreen}{rgb}{0,0.5,0}

\usepackage{xspace}

\newcommand{\scafi}{{ScaFi}\xspace}

\newcommand{\Protelis}{{Protelis}\xspace}

\lstdefinelanguage{scafi}{
  keywords={abstract,case,catch,class,def,%
    do,else,extends,final,finally,%
    for,if,implicit,import,match,mixin,%
    new,null,object,override,package,%
    private,protected,requires,return,sealed,%
    super,this,throw,trait,try,lazy,%
    type,val,var,while,with,yield,forSome},
  otherkeywords={=>,<-,<\%,<:,>:,\#},
  keywordstyle=\color{blue}\textbf,
  keywordstyle=[2]\color{red},
  keywords=[2]{rep,nbr,foldhood,branch,@@,aggregate},
  keywordstyle=[3]\color{violet},
  keywords=[3]{false,true,nbrvar,sense,mid,mux,nbrRange,sumHood,pickHood,countHood,minHood,maxHood,foldhoodPlus,G,C,S,T,broadcast,distanceTo,summarize,minHoodPlus},
  keywordstyle=[4]\color{Emerald},
  keywordstyle=[5]\color{Brown},
  sensitive=true,
  morecomment=[l]{//},
  morecomment=[n]{/*}{*/},
  commentstyle=\color{darkgreen},
  morestring=[b]",
  morestring=[b]',
  morestring=[b]""",
  basicstyle=\lst@ifdisplaystyle\small\fi\ttfamily
}
\lstdefinelanguage{fscafi}{
  keywords={abstract,case,catch,class,def,%
    do,else,extends,final,finally,%
    for,if,implicit,import,match,mixin,%
    new,null,object,override,package,%
    private,protected,requires,return,sealed,%
    super,this,throw,trait,try,lazy,%
    type,val,var,while,with,yield,forSome},
  otherkeywords={=>,<-,<\%,<:,>:,\#},
  keywordstyle=\color{black}\textbf,
  keywordstyle=[2]\color{black}\textbf,
  keywords=[2]{rep,nbr,foldhood,branch,@@},
  keywordstyle=[3]\color{black}\textbf,
  keywords=[3]{false,true,nbrvar,sense,mid,mux,nbrRange,sumHood,pickHood,countHood,minHood,maxHood,foldhoodPlus,G,C,S,T,broadcast,distanceTo,summarize,minHoodPlus},
  keywordstyle=[4]\color{black},
  keywordstyle=[5]\color{black},
  sensitive=true,
  morecomment=[l]{//},
  morecomment=[n]{/*}{*/},
  commentstyle=\color{darkgreen},
  morestring=[b]",
  morestring=[b]',
  morestring=[b]""",
  basicstyle=\lst@ifdisplaystyle\small\fi\ttfamily
}

\lstdefinelanguage{hfc}{
  keywords={def},
  otherkeywords={=>,<-,<\%,<:,>:,\#,@},
  keywordstyle=\color{blue}\textbf,
  keywordstyle=[2]\color{red},
  keywords=[2]{rep,nbr,if,let,in},
  keywordstyle=[3]\color{violet},
  keywords=[3]{false,true,foldhood,mid,mux,min,max,fst,snd,nbrRange,sumHood,pickHood,countHood,minHood,maxHood,foldhoodPlus,G,C,S,T,broadcast,distanceTo,summarize,minHoodPlus,pair,fst,snd,dag_empty,dag_node,dag_join,dag_union},
  keywordstyle=[4]\color{Emerald},
  keywordstyle=[5]\color{Brown},
  sensitive=true,
  morecomment=[l]{//},
  morecomment=[n]{/*}{*/},
  commentstyle=\color{darkgreen},
  morestring=[b]",
  morestring=[b]',
  morestring=[b]"""
}

\newcommand{\setLanguage}[1]{
\lstset{frame=single,basewidth=0.5em,language={#1}}
}

\setLanguage{scafi}

\newcommand{\revisedMinor}[1]{#1}
\newcommand{\revised}[1]{#1}

\newcommand{\FORGET}[1]{}

\newcommand{\HFC}{HFC}
\newcommand{\HFCprime}{$\text{HFC}\,'$}
\newcommand{\FSCAFI}{{NC}\xspace}
\newcommand{\FSCAFIi}{{NC$\,'$}\xspace}
\newcommand{\FeatherweightSCAFI}{{neighbours calculus}\xspace}

\newcommand{\BNFcce}{{\bf ::=}}
\newcommand{\BNFmid}{\;\bigr\rvert\;}

\newcommand{\PROGRAM}{\mathtt{P}}
\newcommand{\FUNCTION}{\mathtt{F}}

\newcommand{\main}{\mathtt{main}}
\newcommand{\e}{\mathtt{e}}
\newcommand{\s}{\mathtt{s}}
\newcommand{\emain}{\e_{\main}}
\newcommand{\fname}{\mathtt{d}}
\newcommand{\xname}{\mathtt{x}}
\newcommand{\yname}{\mathtt{y}}
\newcommand{\zname}{\mathtt{z}}
\newcommand{\bname}{\mathtt{b}}
\newcommand{\oname}{\mathtt{o}}

\newcommand{\anyvalue}{\mathtt{v}}
\newcommand{\lvalue}{\ell}
\newcommand{\fvalue}{\phi}
\newcommand{\funvalue}{\mathtt{f}}
\newcommand{\svalue}{\mathtt{s}}
\newcommand{\snvalue}{\mathtt{r}}
\newcommand{\truevalue}{\mathtt{True}}
\newcommand{\falsevalue}{\mathtt{False}}

\newcommand{\dc}{\mathtt{c}}
\newcommand{\dcOf}[2]{#1(#2)}

\newcommand{\auxNAME}{\textit{aux}}
\newcommand{\aux}[1]{\auxNAME(#1)}
\newcommand{\bodyNAME}{\textit{body}}
\newcommand{\body}[1]{\bodyNAME(#1)}
\newcommand{\argsNAME}{\textit{args}}
\newcommand{\args}[1]{\argsNAME(#1)}
\newcommand{\nameOf}{\textit{name}}

\newcommand{\FVname}{\textbf{FV}}
\newcommand{\FV}[1]{\FVname(#1)}
\newcommand{\FTVname}{\textbf{FTV}}
\newcommand{\FTV}[1]{\FTVname(#1)}

\newcommand{\defK}{\mathtt{def}}
\newcommand{\nbrK}{\mathtt{nbr}}
\newcommand{\repK}{\mathtt{rep}}
\newcommand{\letK}{\mathtt{let}}
\newcommand{\inK}{\mathtt{in}}
\newcommand{\ifK}{\mathtt{if}}

\newcommand{\foldK}{\mathtt{foldhood}}

\newcommand{\toSym}[1]{\stackrel{#1}{\mathrm{\texttt{=>}}}}
\newcommand{\aggrK}{}
\newcommand{\eqSymK}[1]{\mathrm{ \texttt{\aggrK\!\!\{} #1 \texttt{\}} }}
\newcommand{\toSymK}[1]{\mathrm{ \texttt{=> \aggrK\!\!\{} #1 \texttt{\}} }}
\newcommand{\toSymKabs}[1]{\mathrm{ \texttt{=>} #1  }}
\newcommand{\toSymKtag}[2]{\toSym{#1}\mathrm{ \texttt{\aggrK\{} #2 \texttt{\}}}}
\newcommand{\name}{\tau}

\newcommand{\fstK}{\mathtt{fst}}
\newcommand{\sndK}{\mathtt{snd}}
\newcommand{\headK}{\mathtt{head}}
\newcommand{\tailK}{\mathtt{tail}}

\newcommand{\pairK}{\mathtt{pair}}
\newcommand{\consK}{\mathtt{cons}}
\newcommand{\PairK}{\mathtt{Pair}}
\newcommand{\NullK}{\mathtt{Null}}
\newcommand{\ConsK}{\mathtt{Cons}}
\newcommand{\listK}{\mathtt{list}}

\newcommand{\pairltypeOf}[2]{\pairK(#1,#2)}
\newcommand{\listltypeOf}[1]{\listK(#1)}

\newcommand{\nbrlt}{\mathtt{nbrlt}}

\newcommand{\muxK}{\mathtt{mux}}

\newcommand{\type}{\textit{T}}
\newcommand{\ltype}{\textit{L}}
\newcommand{\ftype}{\textit{F}}
\newcommand{\rtype}{\textit{R}}
\newcommand{\stype}{\textit{S}}
\newcommand{\builtintype}{\textit{B}}

\newcommand{\ftypeOf}[1]{\mathtt{field}(#1)}

\newcommand{\btype}{\mathtt{bool}}
\newcommand{\ntype}{\mathtt{num}}

\newcommand{\stvar}{s}
\newcommand{\rtvar}{r}
\newcommand{\ltvar}{l}
\newcommand{\tvar}{t}

\newcommand{\typescheme}{\textit{TS}}
\newcommand{\ltypescheme}{\textit{LS}}

\newcommand{\TtypEnv}{\mathcal{A}}
\newcommand{\OStypEnv}{\mathcal{B}}
\newcommand{\TStypEnv}{\mathcal{D}}
\newcommand{\LTStypEnv}{\mathcal{D}}
\newcommand{\typeofNAME}{\OStypEnv}
\newcommand{\typeof}[1]{\typeofNAME(#1)}

\newcommand{\expTypJud}[4]{#1 ; #2 \vdash #3 : #4}
\newcommand{\funTypJud}[3]{#1 \vdash #2 : #3}
\newcommand{\proTypJud}[2]{\vdash #1 : #2}

\newcommand{\surfaceTyping}[3]{
  \begin{array}{l@{\;}c}
    \stackrel{~}{{\tiny \textrm{[#1]}}} & #2 \\ \hline 
    \multicolumn{2}{c}{#3}
  \end{array}
}
\newcommand{\nullsurfaceTyping}[2]{
  \surfaceTyping{#1}{}{#2}
}

\newcommand{\denotf}[2]{\lambda #1.#2}

\newcommand{\deviceIdSet}{\textbf{D}}
\newcommand{\builtinop}[3]{\llparenthesis #1 \rrparenthesis_{#2}^{#3}}
\newcommand{\filter}{F}

\newcommand{\Trees}{\Theta}
\newcommand{\emptyseq}{\bullet}

\newcommand{\Topo}{\rightarrowtail}
\newcommand{\Sens}{\Sigma}
\newcommand{\Envi}{\textit{Env}}
\newcommand{\EnviS}[2]{\langle #1,#2 \rangle}
\newcommand{\SystS}[2]{\langle #1;#2 \rangle}
\newcommand{\Field}{\Psi}
\newcommand{\Cfg}{N}
\newcommand{\wfn}[1]{\textit{WFE}(#1)}
\newcommand{\senstate}{\sigma}
\newcommand{\Stat}{\textit{Stat}}
\newcommand{\Activation}{\alpha}
\newcommand{\actOFF}{\mathtt{false}}
\newcommand{\actON}{\mathtt{true}}

\newcommand{\nettran}[3]{#1\xrightarrow{#2} #3}
\newcommand{\act}{\textit{act}}
\newcommand{\envact}{\textit{env}}

\newcommand{\envmap}[2]{#1\mapsto #2}
\newcommand{\mapupdate}[2]{#1[#2]}
\newcommand{\globalupdate}[2]{#1\llbracket #2 \rrbracket}
\newcommand{\proj}[2]{{#1}|_{#2}}

\newcommand{\ruleNameSize}[1]{{\scriptsize #1}}

\newcommand{\domofNAME}{\textbf{dom}}
\newcommand{\domof}[1]{\domofNAME(#1)}
\newcommand{\erasureofNAME}{\textbf{erasure}}
\newcommand{\erasureof}[1]{\erasureofNAME(#1)}

\newcommand{\vtree}{\theta}
\newcommand{\mkvtree}[3]{#2 \langle #3 \rangle}
\newcommand{\mkvt}[2]{#1 \langle #2 \rangle}
\newcommand{\piB}[1]{\pi^{#1}}
\newcommand{\piBof}[2]{\piB{#1}(#2)}
\newcommand{\piI}[1]{\pi_{#1}}
\newcommand{\piIof}[2]{\piI{#1}(#2)}
\newcommand{\piIofOv}[1]{\overline{\pi}(#1)}

\newcommand{\lengthOf}[1]{\textit{length}(#1)}

\newcommand{\bsopsem}[5]{#1;#2;#3\vdash #4\Downarrow #5}
\newcommand{\deviceId}{\delta}
\newcommand{\vroot}{\mathbf{\rho}}
\newcommand{\vrootOf}[1]{\vroot(#1)}
\newcommand{\substitution}[2]{#1:=#2}
\newcommand{\applySubstitution}[2]{#1[#2]}
\newcommand{\bsopsemFAIL}[4]{#1;#2;#3\vdash #4\;\FAIL}
\newcommand{\FAIL}{\textup{\textsc{fail}}}

\newcommand{\skiptransition}{\\[10pt]}
\newcommand{\skiptransitionR}{\\[-3pt]}
\newcommand{\skiptransitionN}{\\[-2.22pt]}

\newcommand{\netopsemRule}[3]{\surfaceTyping{#1}{#2}{#3}}

\newcommand{\coherent}[3]{\textit{WFVT}(#3;#1)}
\newcommand{\coherentEnv}[3]{\textit{WFVTE}(#3;#1)}

\newcommand{\ap}[1]{\langle #1 \rangle}
\newcommand{\bp}[1]{\left\lbrace #1 \right\rbrace}

\newcommand{\eqhl}[1]{\colorbox{gray!30}{$\displaystyle#1$}}

\keywords{distributed functional programming, aggregate computing, field calculi, Scala DSL, collective intelligence, self-organisation, decentralised systems} 

\usepackage{hyperref}
\theoremstyle{plain} 

\usepackage{cleveref}

\sloppypar

\begin{document}

%
%
%

\title[Computation Against a Neighbour]{Computation Against a Neighbour:\texorpdfstring{\\}{}Addressing Large-Scale Distribution and Adaptivity with Functional Programming and Scala}

 \author[G.~Audrito]{Giorgio Audrito\lmcsorcid{0000-0002-2319-0375}}[a]
\author[R.~Casadei]{Roberto Casadei\lmcsorcid{0000-0001-9149-949X}}[b]
\author[F.~Damiani]{Ferruccio Damiani\lmcsorcid{0000-0001-8109-1706}}[a]
\author[M.~Viroli]{Mirko Viroli\lmcsorcid{0000-0003-2702-5702}}[b]

\address{Universit{\`a} degli Studi di Torino, Corso Svizzera, 185, Turin, Italy}	
\email{giorgio.audrito@unito.it, ferruccio.damiani@unito.it}  

\address{Alma Mater Studiorum--Universit{\`a} di Bologna, Via Cesare Pavese, 50, Cesena, Italy}	
\email{roby.casadei@unibo.it, mirko.viroli@unibo.it}  






\begin{abstract}

Recent works in contexts like the Internet of Things (IoT) and large-scale Cyber-Physical Systems (CPS) propose the idea of programming distributed systems by focussing on their global behaviour across space and time.
In this view, a potentially vast and heterogeneous set of devices is considered as an ``aggregate''
 to be programmed as a whole, while abstracting away the details of individual behaviour and exchange of messages, which are expressed declaratively.
One such a paradigm, known as aggregate programming,
 builds on computational models
 inspired by field-based coordination.
%
%
%
Existing models such as the field calculus capture interaction with neighbours by a so-called ``neighbouring field'' (a map from neighbours to values). This requires ad-hoc mechanisms to smoothly compose with standard values, thus complicating programming and introducing clutter in aggregate programs, libraries and domain-specific languages (DSLs).

To address this key issue we introduce the novel notion of ``computation against a neighbour'',
 whereby the evaluation of certain subexpressions of the aggregate program are affected by recent corresponding evaluations in neighbours.
We capture this notion in the \emph{\FeatherweightSCAFI} (\FSCAFI), a new field calculus variant
which is shown to smoothly support declarative specification of interaction with neighbours, and correspondingly facilitate the embedding of field computations as internal DSLs in common general-purpose programming languages---as exemplified by a Scala implementation, called \scafi{}.
This paper formalises \FSCAFI{},
 thoroughly compares it with respect to the classic field calculus,
 and shows its expressiveness by means of a case study in edge computing, developed in \scafi{}.

\end{abstract}


\maketitle


\section{Introduction} \label{sec-intro}

Pervasive computing, Internet of Things (IoT), Cyber-Physical Systems (CPS), Smart Cities and related initiatives,
 all point out a trend in informatics envisioning a future where computation is fully pervasive and ubiquitous,
 and is carried on by a potentially huge and dynamic set of heterogeneous devices deployed in physical space.
To address the intrinsic complexity of these settings, a new viewpoint is increasingly emerging:
 a large-scale network of devices, situated in some environment (e.g., the urban area of a smart city),
 can be seen as a computational overlay of the physical world, to be programmed as a single ``distributed machine''.
These kinds of systems are sometimes referred to as \emph{Collective Adaptive Systems (CAS)}~\cite{anderson2013adaptive-collective-systems-herding-black-sheep,DBLP:journals/sttt/NicolaJW20},
to emphasise that computational activities are collective (i.e., they involve multiple coordinated individuals),
  and that a main expected advantage is inherent adaptivity of behaviours to unforeseen changes---whether they are changes/faults
  in the computational environment or pertain unexpected interaction with humans or other systems.
%

\emph{Aggregate Computing} \cite{BPV-COMPUTER2015} is \revised{a macro-approach~\cite{DBLP:journals/corr/abs-2201-03473}} to CAS engineering\revised{~\cite{DBLP:journals/sttt/NicolaJW20}}
\revised{where a network of devices
operating at asynchronous sense-compute-interact rounds
is programmed through a single macro-level program, called an \emph{aggregate program}.}
%
%
%
\revised{An aggregate program expresses both computation and interaction, and its repeated execution by devices on their local contexts promotes the unfolding of collective adaptive behaviour.}
This approach is especially suitable to problems and application domains such as crowd engineering, complex situated coordination, robot/UAV swarms, smart ecosystems and the
 like~\cite{Viroli-et-al:JLAMP-2019}.
One fundamental enabling abstraction for specifying the dynamics of situated collectives
 is that of a \emph{computational field} (or simply, field) \cite{audrito2019tocl,proto06a,tota}: a distributed data structure that
 maps devices to computational objects across time.
Accordingly, \emph{Aggregate Programming} is about describing field computations, namely, how input fields (data coming from sensors) turn into output fields
 (actions feeding actuators)---computations that can be conveniently expressed using the functional paradigm.
%

\revised{Aggregate computing is part of \emph{field-based coordination}~\cite{Viroli-et-al:JLAMP-2019}.
The core formal language
 capturing the essence of aggregate computing using fields
 is the \emph{field calculus} \cite{audrito2019tocl}.
\revised{The field calculus can be seen as a particular kind of \emph{context-oriented programming}~\cite{DBLP:journals/jss/SalvaneschiGP12} approach with \emph{multiple local contexts}
 where behavioural variations
 are determined through branching conditions
 on field values~\cite{DBLP:conf/saso/CasadeiPSV19}.
\revised{
The field calculus
 and its computational model
 also provide an expressive formal framework
for supporting both the investigation of properties 
 about distributed computations
 (like self-stabilisation~\cite{viroli:selfstabilisation},
 density-independence~\cite{BVPD-TAAS2017},
 real-time guarantees~\cite{Audrito-et-al:RTSS-2018},
 universality~\cite{abdv:universality},
 time-fluidity~\cite{DBLP:journals/lmcs/PianiniCVMZ21},
 and deployment flexibility~\cite{ieee-iot-2022-deployment-ac})
 and novel constructs
 (like aggregate processes~\cite{casadei19processes} and the \emph{share} construct~\cite{abdpv:lmcs:share}).
}
The field calculus had been implemented by}  Domain Specific Languages (DSLs) such as Proto (a Scheme external DSL)~\cite{proto06a}, Protelis (a Java external DSL)~\cite{Protelis15},
FCPP (a C++ internal DSL)~\cite{FCPP-ACSOS-2020}.
\revised{These DSLs} rely on
the notion of ``neighbouring value'' (a map from neighbours to data values) to model device interaction.
A neighbouring value} is used to locally express the outcome of message reception from neighbours, and to manipulate such information to collect resulting values.
As a consequence, these languages need two classes of operators and types:
 one for dealing with local values,
 and one for dealing with collection-like, neighbouring values.
Managing and reconciling these two kinds of operations may tend to complicate aggregate programming, design of libraries, and language implementation.
Existing DSLs deal with this issue in a variety of ways: relying on dynamic typing (as in Proto), using macros and meta-programming techniques to alleviate such issues to the user (as in FCPP), or requiring duplication of operation across local and neighbouring values (as in Protelis)---none of which is completely satisfactory.

In order to resolve the duality of local/neighbouring values, with the goal of \emph{conceptual economy} and to promote \emph{smooth embedding} of field computations in mainstream languages and programming practice, in this paper we introduce a novel notion of ``computation against a neighbour'', used to express interaction in field-based coordination by entirely replacing the notion of neighbouring value.
The key idea is to allow the evaluation of a certain sub-expression of the aggregate program to depend on a recent outcome of the evaluation of the same subexpression in a neighbour.
Such a dependency is hence expressed fully declaratively, without escaping the functional paradigm of field-based computing, and is then internally implemented by asynchronous message exchange across neighbours.
To present this mechanism and study its implications, in this paper we:
\begin{enumerate}
\item
define syntax, typing, operational semantics of a foundation calculus, called \emph{\FeatherweightSCAFI} (\FSCAFI);
\item
investigate the properties of the \FSCAFI\ and its relationship with the field calculus;
\item show the advantages that \FSCAFI\ brings in term of smooth embedding into \revised{Scala---host language chosen for its hybrid object-oriented/functional nature, flexibility, and suitability for internal DSL development~\cite{DBLP:conf/icfem/ArthoHKY15,Calus:2017:FIM:3141858.3141861}}.
\end{enumerate}
%
%
%
In particular, the last contribution is based on the implementation of the \FSCAFI\ computational model in \emph{\scafi{}}\footnote{\url{https://github.com/scafi/scafi}} (\emph{Sca}la \emph{Fi}elds).
\scafi{} is a Scala~\cite{sloane:hal-00350269} aggregate programming toolkit comprising an internal DSL for aggregate programming, which has been used in a variety of
applications \cite{casadei2021eaai-procs,casadei2019scc,casadei2019fmec,casadei2018attackresistantaggregate}.
Hence, in addition to formalising the ``computation against a neighbour'' mechanism, \FSCAFI\ also serves as formalisation of the core of \scafi{}.
%
%

The remainder of this paper is structured as follows.
\Cref{sec-motivation} provides background information about Aggregate Computing and field computations\revised{, as well as motivation for this work}.
\Cref{sec-calculus-syntax-and-typing} describes syntax, typing, and operational semantics  of \FSCAFI. 
\Cref{sec-properties} provides a detailed account on the properties of \FSCAFI, by comparison with the field calculus.
\Cref{sec-NC-at-work} presents the \scafi{} Scala-based \FSCAFI{} implementation, and the key advantage of the approach over other implementations approaches for field calculi.
\Cref{sec-case} describes an edge computing case study built on \FSCAFI{}/\scafi{}.
Finally, \Cref{sec-related} summarises related works, and
\Cref{sec-conclusion} ends up the paper with a wrap-up and discussion of future works.

\setLanguage{scafi}

\section{Background \revised{and Motivation}} \label{sec-motivation}

In this section, we recap Aggregate Computing (\Cref{background:ac}) and its formal embodiment into a computational model based on a \emph{field} abstraction (\Cref{background:fc}), as a response to the need of engineering the collective adaptive behaviour of large-scale distributed and physically situated systems.
\revised{Then, we describe how Aggregate Computing
 can be used to express the paradigmatic example
 of self-healing channels (\Cref{sec:background:example}),
 to give an early idea of
 how the approach works
 and its main benefits (e.g., composability).
Finally,
 we motivate
 the work developed in this paper in terms of goals, decisions, and technical challenges (\Cref{sec:motivation-detail}).
}


\subsection{Aggregate Computing}
\label{background:ac}

Aggregate Computing is an approach to CAS development
that abstracts from the traditional, device-centric viewpoint (where the programs describe what a device should do)
in favour of a ``holistic stance'' where the target of design/development is the whole collection of situated and interacting devices that compose the system, seen as a single programmable, distributed, computational body \cite{BPV-COMPUTER2015}.
The core idea is to provide a single specification \revised{-- which we call an \emph{aggregate program} --
of how the system globally behaves. This is done} without even mentioning the existence of individuals, and hence independently of the shape and size of the set of devices: it is ``under-the-hood'' that the local behaviours to be executed by the individual devices are derived.
This approach may be labelled as a \emph{macro-to-micro} one, to differentiate it from the more classical (micro-to-macro) approach where the components are individually implemented in order to produce the intended system-level behaviour by more or less ``controlled'' emergence.
Hence, from the programming perspective, the key advantage is the ability of \emph{declaratively} expressing the logic of an ensemble, avoiding to directly solve the generally intractable local-to-global mapping problem and clearly \emph{separating the concerns} of overall aggregate behaviour from the device-specific ones.
%

The general idea of programming at the macro level dates back to works such as \cite{regiment,bealamorphous} in the context of wireless sensor networks, but it has recently received considerable renewed interest (see, e.g., the
 survey\revised{s~\cite{SpatialIGI2013,DBLP:journals/corr/abs-2201-03473}} and the many related approaches mentioned in \Cref{sec-related}) with the emergence of the IoT.
In this research context, the key contribution of Aggregate Computing lies in supporting \emph{abstract and resilient composition of collective adaptive behaviours}.
Since an aggregate program is expressed in a way that is independent of the actual number and dislocation of devices -- and the resulting computation is a repetitive, gossip-like process of data distribution and computation -- adaptation to changes and faults is inherent, and can be controlled by relying on specific programming patterns \cite{VBDP-SASO2015}, often biologically inspired.
Then, at this level, compositionality refers to the capability of combining behaviours in such a way that the result of the combination,
and its properties, are in several cases fully predictable
 \cite{BVPD-TAAS2017,viroli:selfstabilisation,vcp-ubicomp16}.
In particular, \emph{self-stabilisation} is retained by a large class of aggregate programs: it makes an aggregate system able to tolerate circumscribed failures and react to disruptive changes in the environment in order to re-establish, after some transient adjusting phase, the proper operation.

\subsection{Computing with Fields}
\label{background:fc}

\begin{figure}
\centering
\includegraphics[width=\textwidth]{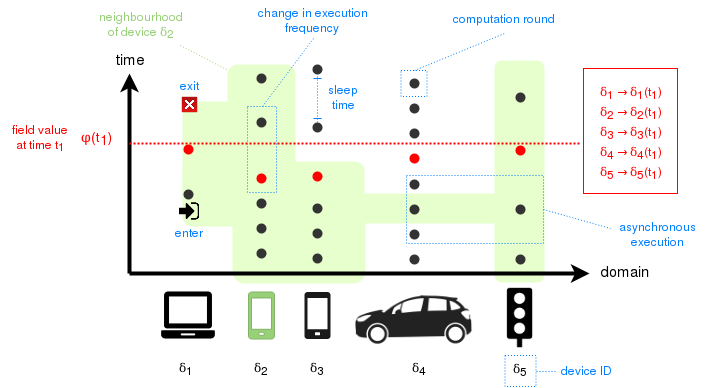}
\caption{Dynamic view of a computational field: a field maps devices to values over time; devices can enter and exit a given \emph{space-time domain} freely (openness); devices fire asynchronously (possibly at varying rates), execute a \emph{computational round} upon fires producing the local value of the field, and  then ``sleep'' in between fires. At a given time, a snapshot of the most-recent value at each device gives a field value (this view of an external, global observer is conceptually useful but global fields are not to be computationally manipulated in a direct fashion).
A device can only directly communicate with a subset of other devices which is known as its \emph{neighbourhood} (e.g., for $\delta_2$, represented by the green area), which usually (but not necessarily) reflects spatial proximity in situated systems.}
\label{fig:compfield}
\end{figure}

The computational framework incarnating the ideas and goals of Aggregate Computing
is based on the notion of \emph{computational field} (or simply field) \cite{Viroli-et-al:JLAMP-2019}.
Intuitively, as shown in \Cref{fig:compfield}, a computational field is an abstraction that represents a distributed data structure mapping points in space-time (the field \emph{domain}) to values produced (by some device) as computation result at those points.
%
%
Programming field computations hence encourages to reason in terms of evolving global structures that continuously express both the result of computation and the relationships between individuals.

Field computations generally comprise mechanisms for:
\emph{(i)} lifting standard values and computations, which are ``local'', to work as whole fields by flatly applying them to each space-time point of the domain;
\emph{(ii)} expressing the dynamics of fields, namely, how fields evolve over time;
\emph{(iii)} exchanging values across \emph{neighbours}, such that information can flow beyond localities and interaction can unfold to realise complex patterns such as, e.g., outward propagation of local events; and
\emph{(iv)} branching fields into spatio-temporally-isolated sub-fields, in order to organise a computation into multiple parts co-existing in different space-time regions.

Inspired by 
minimal core
calculi such as $\lambda$-calculus~\cite{LambdaCalculus} 
and FJ~\cite{FJ}, the above mechanisms
 have been formalised by the \emph{field calculus}  \cite{audrito2019tocl}.
The field calculus relies on the basic functional programming model abstractions (i.e., first-class functions and functional composition) to support composability of distributed behaviour.
In a nutshell,
 expressions denote whole fields; a specific construct, called \texttt{rep}, deals with field dynamics;
another construct, called \texttt{nbr}, declaratively expresses interaction with neighbours; finally, higher-order function call captures behaviour activation as well as \emph{branching} (as a ``field of functions'' can be invoked through the \emph{call} operator, space-time regions are naturally defined by where/when the same function is actually called).

As suggested in \Cref{fig:compfield}, field computations assume a set of networked devices that
  run at asynchronous or partially synchronous \emph{rounds of execution}.
 It is such an iterative execution that, combined with repetitive local sensing
  and device-to-device interaction,
  enables intrinsic adaptation to environmental perturbations
  and progressive steering of the system towards the desired states---much like in \emph{swarm} and \emph{computational collective intelligence}~\cite{Bonabeau:1999,szuba2001computational-collective-intelligence}.
 In each computation round, a device: \emph{(i) determines its local context},
  by retrieving any previous state and collecting sensor values as well as messages potentially received from its neighbours;
  \emph{(ii) locally executes a field computation expression}, in a contextual fashion and
  according to the local semantics, producing an output and
  an \emph{export} value which provides information for collective coordination;
  \emph{(iii) shares the export with its neighbours}, through a conceptual broadcast;
  and finally \emph{(iv) feeds actuators}, with the produced output.
In other words, this execution model implements a sort of distributed, continuous closed loop
 between field computation and the environment, where each device computes values
 that are contextually related with those computed by neighbours.

\revised{
\subsection{Example: self-healing channel}\label{sec:background:example}
A paradigmatic example of what can be done using the Aggregate Computing approach is the \emph{self-healing channel} (cf. \Cref{fig:example-channel}), i.e., the construction of a Boolean field that is true only in the set of devices
 that connect two endpoints of the system
 through a minimum path,
 and that automatically adapts
 to changes of the endpoints and the system connectivity network.
Such a collective adaptive behaviour can be expressed through an aggregate program as follows:
}
\begin{lstlisting}[language={hfc},basicstyle=\small\ttfamily]
// Definition
def channel(endpoint1, endpoint2) {
  gradient(endpoint1) + gradient(endpoint2) <= distanceBetween(endpoint1,endpoint2)
}
// Use
channel(isSource(), isDestination(), 10)
\end{lstlisting}
\revised{
where \lstinline|isSource| and \lstinline|isDestination| are assumed to be sensors available in every device.
The self-healing channel algorithm
 can be implemented through a function \lstinline|channel|
 that reuses existing building blocks of collective behaviour,
 such as \lstinline|gradient(s)| (computing the self-healing field of minimum distances from the source \lstinline|s|)
 and \lstinline|distanceBetween(a,b)|
 (spreading in the network the distance between \lstinline|a| and \lstinline|b|, which can be built over \lstinline|gradient|s).
\revisedMinor{It works by exploiting the triangle inequality principle: given a device at position $p$, the sum of the distance $d_1$ from \lstinline|endpoint1| to $p$ (computed by the first call to \lstinline|gradient|) with the distance $d_2$ from \lstinline|endpoint2| to $p$ (computed by the second call to \lstinline|gradient|) must be greater than (or equal to, if we consider degenerate triangles) the distance $d_{1\leftrightarrow{}2}$ from \lstinline|endpoint1| to \lstinline|endpoint2| (which is computed by the call to \lstinline|distanceBetween|). So, the devices belonging to the channel are those for which $d_1+d_2 \leq d_{1\leftrightarrow{}2}$, i.e., the devices that lie in the path from \lstinline|endpoint1| to \lstinline|endpoint2|.
}
\revisedMinor{Then, the execution is decentralised, progressive, and open-ended. Given the round-based execution model mentioned previously, the expression corresponding to the \lstinline|channel|
  function application is to be repeatedly evaluated
  by every device,
  so that up-to-date local information (perceived by sensors or obtained from recent messages from neighbours)
  is properly incorporated to progressively adjust the collective result.
So, local changes (e.g., a device un/becomes a source, as perceived by \lstinline|isSource|, or the distance to neighbours changes as a result of mobility, which is perceived internally in \lstinline|gradient|s)
 will affect the local evaluation of the involved gradients,
 and such effects will propagate to neighbours,
 eventually affecting the rest of the network
 and the global channel result.
Furthermore, it can be shown that this computation is self-stabilising~\cite{viroli:selfstabilisation},
 which means that it would eventually converge to the correct global result once the environmental inputs stop changing.
}
It turns out that \lstinline|gradient|
 can be implemented in many ways~\cite{ACDV-SASO2017},
 but a simple implementation leverages a combination of \lstinline|rep| and \lstinline|nbr|---this will be described in detail and incrementally in \Cref{sec-calculus-syntax-and-typing} (cf. \Cref{ex:time-nbr-gradient,exa:typing}).
\revisedMinor{
Notice that the use of \lstinline|nbr|
 means that an evaluation of \lstinline|gradient|
 will cause pertinent information to be exchanged with neighbours.
}
The point of this example is not to fully understand the technical details involved, which will be clarified in \Cref{sec-calculus-syntax-and-typing},
 but to give the intuition of the compositionality of Aggregate Computing, where functions expressing collective adaptive behaviour can be combined to express more complex collective adaptive behaviour.

\begin{figure}
\centering
\includegraphics[width=0.45\textwidth]{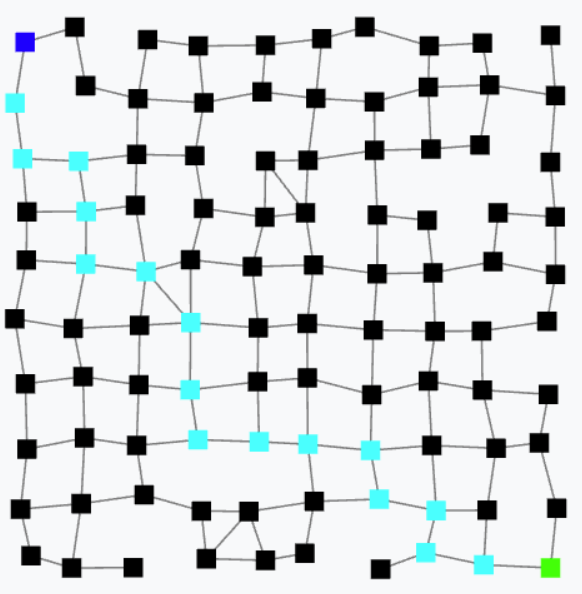}\hfill\includegraphics[width=0.45\textwidth]{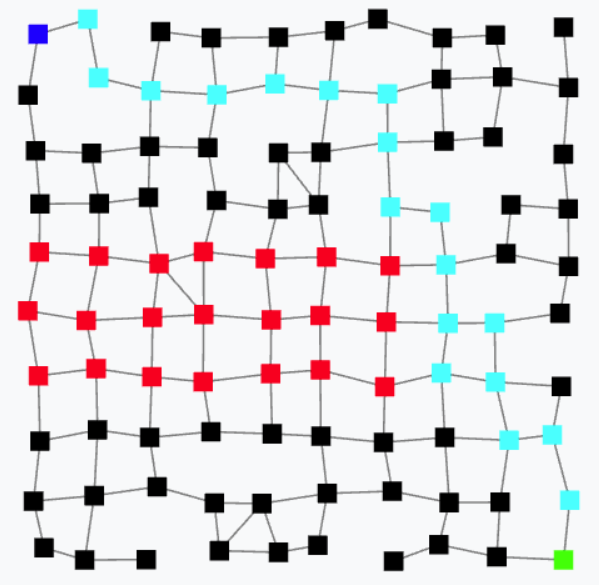}
\caption{\revised{A pictorial view of the self-healing channel. The devices where the \lstinline|channel| yields \lstinline|true| are shown as cyan dots, whereas the blue and green dots denote the two endpoints of the channel, and red dots denote non-working devices (e.g., as a result of a blackout). Through the two snapshots, we see that a channel eventually self-adjust itself reacting to the perturbation induced by the blackout.}}
\label{fig:example-channel}
\end{figure}

\subsection{Motivation}\label{sec:motivation-detail}

Given the background on Aggregate Computing and field computations,
 we are now able to provide the motivation for this work.
This is done in three parts:
 (i) we first motivate the usefulness of internal or embedded DSLs for field computations;
 then,
 (ii) we motivate the choice of Scala
 as our target host language;
 finally,
 (iii) we explain technical issues that have driven the design of NC
 in the attempt to
 embed field computations
 with a small set of requirements at the host side.

\subsubsection{Embedded DSLs for Field Computations}

A key issue for the practical applicability of the aggregate programming model enabled by field calculi is to find ways to smoothly integrate with the standard development practice---languages, processes and tools.
We observe that an aggregate program is generally a small part of an overall application or system, addressing issues like advanced coordination
and adaptivity, and interacting with other subsystems to gather inputs
or provide system services~\cite{cv2017actorcas}.

Most former aggregate programming languages such as  Proto~\cite{proto06a} and Protelis~\cite{Protelis15}
 are \emph{standalone} DSLs.
Also called \emph{external} DSLs~\cite{Riti2018external-dsls},
 these languages have their own syntax and semantics,
 and are either interpreted
 or compiled to an intermediate language of some target platform (e.g., Java and the JVM, as in Protelis).
Despite the flexibility and
 the support provided by modern language workbenches
 (cf.\ Xtext~\cite{XtextBook2016}),
 effort in development and maintenance of an external DSL is typically large,
 high expertise is required,
 and, additionally, integration with other development languages and tools
 tends to be cumbersome.

A prominent, modern approach to address this problem is to devise an \emph{internal DSL}~\cite{voelter2013dsl} (also called \emph{embedded DSL}~\cite{DBLP:journals/cacm/Ghosh11}) that provides mechanisms to support the new features on top of an adequate and well-known host programming language---through a very expressive API.
This provides easy reuse of functionality and tools of the host language, for a smooth adoption.
Given a source DSL and a target language, \emph{DSL embedding} is the problem of implementing the DSL internally to the target language.

Field calculi deal with fields, i.e., collections of values, which are to be treated as first-class entities.
However, common GPLs are subject to functional and typing constraints when dealing with ``collections'' so that point-wise manipulation typically requires widespread, explicit mapping operations which make the programs verbose.
So far, this approach has been followed by the aggregate programming language FCPP \cite{FCPP-ACSOS-2020}, a C++ internal DSL. The difficulties of implementing \HFC{} as an internal DSL are partially addressed in FCPP by operator overloading, template metaprogramming techniques and macros; however, several low-level details are still left to the responsibility of the programmer.

%
%

\subsubsection{Scala as Host Language}

%
Both the syntax and the semantics of an embedded DSL are limited by the constraints exerted by its host language.
Therefore, the choice of the host language is led by requirements and desiderata for the DSL. In our case these include:
\begin{itemize}
\item \emph{pragmatism} (supporting easy reuse of existing programming structures and mechanisms),
\item \emph{reliability} (intercepting errors concerning syntax and types at compile-time), and
\item \emph{expressivity} (offering a concise and eloquent syntax, minimising the accidental language complexity induced by the environment);
\end{itemize}
which
led us to select the Scala programming language.
%
Indeed, Scala is a premier language
 for implementing embedded DSLs~\cite{DBLP:conf/icfem/ArthoHKY15,dubochet2011embedded-dsls,DBLP:conf/oopsla/ChafiDMRSHOO10,Calus:2017:FIM:3141858.3141861}.
It runs on the Java Virtual Machine (JVM) and thus enables straightforward interaction and reuse of libraries in the Java ecosystem.
It has a flexible syntax, powerful type system, and convenient features (such as by-name parameters, type members, traits etc.) for building type-safe, expressive APIs.
Additionally, its smooth integration of the object-oriented and functional paradigms~\cite{odersky2014unifying}, as well as its ecosystem of libraries for distributed computing, make it a convenient tool for in-large software development of frameworks.
Functional programming support is especially important
 for implementing field calculi;
 however, compared with languages such as Haskell,
 it is more pragmatic and versatile.
On the other hand, compared with very practical languages such as Kotlin,
 it is also (like Haskell) oriented towards
 programming language research
 and a philosophy of scalable abstractions~\cite{odersky2005scalacomponents}.

\subsubsection{Technical Issues in Embedding Field Computations in a Host Language}

A first key problem in embedding field computations in a host language lies in the potential mismatch between the representation of local values (values of local, standard types) and the representation of field expressions.
For the former, one would seek for the natural representation of the host language,  e.g., literal \texttt{1} representing a local value of type \texttt{Int}.
For the latter, one has to combine the natural representation with the additional constructs manipulating fields; namely,  as in field calculus, literal \texttt{1} (or some variation) should represent a computational field equal to $1$ in all space-time points of the local device's neighbourhood: this should in principle be given a type \texttt{Field[Int]}, inheriting all the operations one would use over the local counterpart \texttt{Int} (\texttt{+}, \texttt{-}, and so on), since in field calculus \texttt{1+2} is the sum of two fields giving a field of $3$ in all space-time points.
As a consequence, standard types such as numbers, Booleans, characters, objects, and the corresponding operators, need to be coherently lifted in order to work under the field interpretation.
That is, given any type \texttt{T}, the type \texttt{Field[T]} should support all the operations provided by \texttt{T} but lifted to the field context:
for instance, expression $\texttt{e}_1+\texttt{e}_2$ where $\texttt{e}_1$ and $\texttt{e}_2$ have type \texttt{Field[Int]} should naturally give an element of type \texttt{Field[Int]}.
Ideally, the type system should continue to do its job with field types, and field expressions should be written in a notation which is analogously simple as the local counterpart.

A second key problem stems from the declarativeness and compositionality of field computations: the host language should provide the means for defining blocks of code
and for controlling \emph{when} and \emph{how} these are executed, mainly by deferring their evaluation until a later time and properly contextualising them to the local information available in each device.
These mechanisms should also be lightweight so as to keep the impact on the user-side of the DSL as little as possible.
For instance, using Java as host language would require many expression to be wrapped in a 0-ary lambda, which would clutter code.

The third, key technical difficulty is to properly deal with the interaction model of field computations, which is neighbour-driven.
Field computations are equipped with a (logical or physical) notion of \emph{neighbourhood} that basically expresses the boundary of direct communication, i.e., what devices can be reached by a given device through a communication act occurring in a certain position of computation.
The observation of values that a device reads from its neighbours is typically modelled by reification into a \emph{neighbouring value} (a map from neighbours to values), which can be manipulated functionally until being collapsed back to a local value by means of some \emph{folding} operator---e.g., computing the minimum value.
This requires one to explicitly differentiate, syntactically and semantically, the two classes of types, neighbouring types and local types, raising the issue of how to lift standard local operators to neighbour types.
Thanks to the features of the Scala programming language, and as detailed in this paper,
the problem could be handled as follows:
\begin{itemize}
 \item the same type, say \texttt{Int}, is used both for local types and neighbouring types;
 \item the notion of computation over neighbouring values is semantically turned into a notion of computation ``against'' a neighbour (namely, a computation whose outcome depends on recent result of computation in that neighbour), hence there is no longer need of two kinds of type;
 \item folding operations are the triggers for a universal quantification process, iterating computations against all pertinent neighbours.
\end{itemize}
Such changes may impact theory and practice of field computations.
On a practical level,
 they enable expressive and smooth integration with Scala programming.
On a theoretical side, however,
 we need to investigate
 how these changes affect the expressiveness
 with respect to the field calculus~\cite{audrito2019tocl}.
%
}

\setLanguage{hfc}

\section{The \FeatherweightSCAFI{}
}\label{sec-calculus-syntax-and-typing}

\setLanguage{fscafi}

%
This section presents  the \emph{\FeatherweightSCAFI  (\FSCAFI)}, a minimal core calculus that models Aggregate Computing via \emph{computations against a neighbour}, instead of the {neighbouring values} used in the \emph{higher-order field calculus (\HFC{})}~\cite{audrito2019tocl}---a thoughtful comparison between \FSCAFI\ and \HFC{} is presented in Section~\ref{sec-properties}.

Devices undergo computation in rounds. When a round starts, the device gathers information about messages received from neighbours (only the last message from
 each neighbour is actually considered), performs an evaluation of the program, and finally emits a message to all neighbours with information about the outcome of
 computation. The scheduling policy of such rounds is abstracted in this formalisation, so any (synchronous or asynchronous) scheduling policy is admissible, even though in many aggregate programming settings it is typically assumed fair\footnote{An infinite scheduling sequence is \emph{fair} if every device appears infinitely many times in the sequence.} and non-synchronous.

\Cref{sec-calculus-syntax} presents the syntax of \FSCAFI;
\Cref{sec-typing} presents its type system;
\Cref{sec-calculus-device-semantics} presents an operational semantics for the computation that takes place on individual devices;
and  \Cref{sec-calculus-network-semantics} presents an operational semantics for the evolution of whole networks.

\subsection{Syntax}\label{sec-calculus-syntax}

%
The syntax of \FSCAFI\ is given in Figure~\ref{fig:source:syntax}.
Following~\cite{FJ}, the overbar notation denotes metavariables over sequences and the empty sequence is denoted  by $\emptyseq$;
e.g., for expressions, we let $\overline\e$ range over sequences of expressions, written $\e_1,\,\e_2,\,\ldots\,\e_n$ $(n\ge 0)$.
\FSCAFI\ focusses on aggregate programming constructs: hence,
it is  parametric in the set of built-in data constructors and functions. In the examples, we consider the set of built-in data constructors and functions listed (with their types) in Section~\ref{sec-typing}.
\revised{In following explanations, we use symbol $\deviceId$ (and variants as $\deviceId', \deviceId_1, \ldots$) to identify devices.
For convenience, we also summarise meta-variables and symbols commonly used throughout the paper in \Cref{table:metasymbols}.
}

\begin{table}
\revised{
\begin{tabular}{|l|l|l|l|}
\hline
\textbf{} & \textbf{Description} &
\textbf{} & \textbf{Description}
\\\hline
$\PROGRAM$ & Program
&
$\FUNCTION$ & Function declaration
\\
$\e$ & Expression
&
$\xname$ & Variable
\\
$\anyvalue$ & Value
&
$\lvalue$ & Local value
\\
$\fvalue$ & Field value
&
$\funvalue$ & Function value
\\
$\fname$ & Defined function
&
$\bname$ & Built-in function
\\
$\dc$ & Data constructor
&
$\deviceId$ & Device identifier
\\
$\type$ & Type
&
$\builtintype$ & Built-in type
\\
$\tvar$ & Type variable
&
$\typescheme$ & Type scheme
\\
$\TStypEnv$ & Type-scheme environment
&
$\OStypEnv$ & Built-in type-scheme environment
\\
$\TtypEnv$ & Type assumptions
&
$\vtree$ & Value-tree
\\
$\Trees$ & Value-tree environment
&
$\senstate$ & (Local) Sensor state
\\
$\Field$ & Value-tree field
&
$\Activation$ & Activation predicate
\\
$\Stat$ & Status
&
$\Sens$ & (Global) Sensor state
\\
$\Envi$ & Environment
&
$\Cfg$ & Network configuration
\\
$\act$ & Action label
&&
\\
\hline
\end{tabular}
}
\caption{\revised{Summary of meta-variables and symbols used in this paper.}}
\label{table:metasymbols}
\end{table}

%

\begin{figure}[!t]
\centering
\centerline{\framebox[\textwidth]{$
\begin{array}{l@{\hspace{1mm}}c@{\hspace{1mm}}l@{\hspace{20mm}}r}
        \PROGRAM & \BNFcce & \overline{\FUNCTION}  \; \e
                                                                                                                                                                                                        &   {\footnotesize \mbox{program}}                                                                                                                                                                                                         \\[3pt]
        \FUNCTION & \BNFcce &  \defK \; \fname (\overline{\xname}) \; \eqSymK{\e}                                                                                                                                                                                                        &   {\footnotesize \mbox{function declaration}}
        \\[3pt]
        \e & \BNFcce & \xname \, \BNFmid \, \anyvalue \, \BNFmid \, (\overline{\xname}) \; \toSymK{\e} \, \BNFmid \, \e(\overline{\e}) \, \BNFmid \,
         \repK(\e)\{ \e \}
          \, \BNFmid \, \nbrK\{\e\}
 &   {\footnotesize \mbox{expression}} \\[1pt]
     & \BNFmid & \foldK(\e, \e, \e)
        \\[3pt]
        \anyvalue & \BNFcce &  {\dc(\overline{\anyvalue})} \; \BNFmid \; \funvalue                                                                                                                                                                                                        &   {\footnotesize\mbox{value}}
        \\[3pt]
        \funvalue & \BNFcce & \bname \; \BNFmid \; \fname \; \BNFmid \; (\overline{\xname}) \; \toSymK{\e}                                                                                                                                                                                                        &   {\footnotesize\mbox{function value}} \\
\end{array}
$}
}
\caption{Syntax of \FSCAFI{}.}
\label{fig:source:syntax}
\end{figure}

A program  $\PROGRAM$ consists of a sequence $\overline{\FUNCTION}$ of function declarations and a main expression $\e$.
A function declaration  $\FUNCTION$ defines a (possibly recursive) function; it consists of a name $\fname$, $n\ge 0$ variable names $\overline{\xname}$ representing the formal parameters, and an expression $\e$ representing the body of the function.

Expressions $\e$ are the main entities of the calculus, modelling a whole field computation. An expression can be:
a variable $\xname$, used as function formal parameter;
a value $\anyvalue$;
an anonymous function $(\overline{\xname}) \; \toSymK{\e}$ (where $\overline{\xname}$ are the formal parameters and $\e$ is the body),
a function call $\e(\overline{\e})$;
a $\repK$-expression $\repK(\e)\{\e\}$, modelling time evolution;
an $\nbrK$-expression $\nbrK\{\e\}$, modelling  neighbourhood interaction;
or a $\foldK$-expression $\foldK(\e,\e,\e)$ which combines values obtained from neighbours.

The set of the \emph{free variables} of an expression $\e$, denoted by $\FV{\e}$, is defined as usual (the only binding construct is
 $(\overline{\xname}) \; \toSymK{\e}$). An expression $\e$ is \emph{closed}  if $\FV{\e}=\emptyseq$. The main expression of a program must be closed, \revised{whereas anonymous functions $(\overline{\xname}) \; \toSymK{\e}$ may be \emph{open} (i.e., not closed).}

A value can be either a \emph{data value} $\dc(\overline{\anyvalue})$ or a \emph{functional value} $\funvalue$.
A data value consists of a \emph{data constructor} $\dc$ of some arity $m\ge 0$ applied to a sequence of $m$ data values
$\overline{\anyvalue}=\anyvalue_1,...,\anyvalue_m$. For readability, the parenthesis may be omitted for arity $m=0$, writing $\dc()$ as $\dc$.
According to the data constructors listed in Figure~\ref{fig:typeof-built-in}, examples of data values are:
the Booleans $\truevalue$ and $\falsevalue$, numbers, pairs (like $\PairK(\truevalue,\PairK(5,7))$) and lists (like $\ConsK(3,\ConsK(4,\NullK))$).

Functional values $\funvalue$ comprise:
\begin{itemize}
        \item
        declared function names $\fname$;
        \item
        \revised{\emph{closed}} anonymous function expressions $(\overline{\xname}) \; \toSymK{\e}$ (i.e., such that $\FV{\e} \subseteq \{\overline\xname\}$);
        \item
        built-in functions $\bname$, which can in turn be:
        \begin{itemize}
                \item {\emph{pure operators} $\oname$, such as functions for building and decomposing pairs (\texttt{pair}, \texttt{fst}, \texttt{snd})
      and lists (\texttt{cons}, \texttt{head}, \texttt{tail}),
 the equality function ($=$), mathematical and logical functions (\texttt{+}, \texttt{\&\&}, ...), and so on;
                }
                \item \emph{sensors} $\svalue$, which depend on the current environmental conditions of the computing device $\deviceId$, such as a \texttt{temperature} sensor; 
                \item \emph{relational sensors} $\snvalue$, 
                which in addition depend also on a specific neighbour device $\deviceId'$ (e.g., $\texttt{nbrRange}$, which measures the distance with a neighbour device).
        \end{itemize}
{In case $\e$ is a binary built-in function $\bname$, we  write
$\e_1 ~\bname~ \e_2$ for the function call $\bname(\e_1,\e_2)$ whenever convenient for readability of the whole expression in which it is contained.}
\end{itemize}

The key constructs of the calculus are:
\begin{itemize}
        \item
        {\em Function  call:} $\e(\e_1,\ldots,\e_n)$ is the main construct of the language.
        The function call evaluates to the result of applying the function value $\funvalue$ produced by the evaluation of $\e$  to the value of the parameters $\e_1, \ldots, \e_n$
      \emph{relatively to the aligned neighbours}, that is, relatively to the neighbours that  in their last execution round
    have evaluated $\e$ to a function value
    with the same name of  $\funvalue$. Hence, calling an expression $\e$ of function type acts as a branch, where each function $\funvalue$ obtainable from $\e$ is applied only on the subspace of devices which evaluated $\e$ to a (syntactically) identical function $\funvalue$.
        \item
        {\em Time evolution:} $\repK (\e_1)\{\e_2\}$ is a 
         construct for dynamically changing fields through the ``\texttt{rep}eated'' application of the functional expression $\e_2$.
          At the first computation round
          (or, more precisely, when no previous state is available---e.g., initially or at re-entrance after state was cleared out due to branching), $\e_2$ is applied to $\e_1$, then at each other step it is applied to the value obtained at the previous step.
        For instance, $\repK(0)\{(\xname)\; \toSymK{\xname + 1}\}$ counts how many rounds each device has computed (from the beginning, or more generally, since that piece of state was missing).
        \item
        {\em Neighbourhood interaction:} $\foldK(\e_1, \e_2, \e_3)$ and $\nbrK \{\e\}$ model device-to-device interaction, and are at the core of the ``computation against a neighbour'' mechanism.
        The $\foldK$ construct evaluates expression $\e_3$ against every aligned neighbour (excluding the device itself),
		then aggregates the values collected through binary operator $\e_2$ together with the initial value $\e_1$.
		\revised{The interaction pattern is thus a common ``map (expression $\e_3$ to each neighbours' data) and reduce (folding through $\e_2$)''. Following standard practice, data aggregation is performed through binary folding, which smoothly covers most application scenarios.\footnote{\revised{In case where a non-binary aggregation is needed, it is still possible to fold neighbour's data into a list through a binary join operator, and then process the resulting list with the needed aggregation.}}}
        Used inside $\e_3$, the $\nbrK$ construct tags a sub-expression $\e$ signalling that, when evaluated against a neighbour $\deviceId$, it should not be actually evaluated as usual, but should give as result the one obtained by evaluating $\e$ in $\deviceId$. Put in other words, when evaluated against $\deviceId$, $\nbrK \{\e\}$ means ``observing'' the recent resulting value of $\e$ in $\deviceId$---and also let later $\deviceId$ observe the local valued of $\e$, conversely.
		\revised{Such behaviour requires to effectively broadcast\footnote{\label{footnote:broadcast}\revised{Even though $\texttt{nbr}$ requires to share values with neighbours, we do not assume the underlying implementation to be necessarily based on network broadcasts. Indeed, messages from different $\texttt{nbr}$s are usually packed together, and may be dispatched to neighbours even through different point-to-point messages.}} the values evaluated for $\e$, in order to make them accessible to neighbour devices.} Subexpressions of $\e_3$ not containing $\nbrK$ are evaluated as usual instead, i.e., with \emph{no} gathering of information from neighbours.

        As an example, consider the expression \[\foldK(2, \texttt{+}, \texttt{min}(\nbrK\{\texttt{temperature}()\}, \texttt{temperature}()) )\]
evaluated in device $\deviceId_1$ (in which $\texttt{temperature}() = 10$) with neighbours $\deviceId_2$ and $\deviceId_3$
(in which $\texttt{temperature}()$ gave $15$ and $5$ in their last evaluation round, orderly).
The result of the expression is then computed adding $2$, $\texttt{min}(15,10)$ and $\texttt{min}(5,10)$ for a final value of $17$.
\end{itemize}

\subsubsection*{Additional syntactic sugar} To facilitate specification of examples, we  write $\ifK (\e_1) \{ \e_2 \} \{ \e_3 \}$ as a shorthand for $\muxK(\e_1, () ~\toSymK{\e_2}, () ~\toSymK{\e_3})()$, where the multiplexer function $\muxK$ selects between its second and third argument based on the value of the first, in formulas: $\muxK(b, \anyvalue_\top, \anyvalue_\bot) = \anyvalue_b$.
Additionally, in order to improve readability, we  also sometimes write $\letK \; \xname = \e_1 \; \inK \; \e_2$  as syntactic sugar for $((\xname) \toSymK{\e_2}) \, (\e_1)$.


%
%


\subsection{Typing}\label{sec-typing}

We now present a type system for \FSCAFI. Since the type system is a customisation of the Hindley-Milner type system~\cite{Damas-Milner:POPL-1982},
there is an algorithm (not presented here) that, given an expression $\e$ and type assumptions for its free variables,
either fails (if the expression cannot be typed under the given type assumptions) or returns its \emph{principal type}, i.e.,
a type such that all the types that can be assigned to $\e$ by the type inference rules can be obtained from the principal type by substituting
type variables with types. The syntax of type and type schemes is presented in Figure~\ref{fig:SurfaceTyping} (top),
where $\builtintype$ ranges over the built-in types provided by the host language (e.g., $\ntype$, $\btype$, $\pairltypeOf{\type_1}{\type_2}$ for any $\type_1$ and $\type_2$, $\listltypeOf{\type}$ for any $\type$). We do not expand on the structure of built-in types and type operators, as it is not needed to define \FSCAFI\ typing.
\revised{So,
 a type $\type$
 is either a built-in type $\builtintype$,
 a function type $(\overline{\type}) \to \type$,
 or a type variable $\tvar$.
A type scheme has syntax $\forall\overline{\tvar}.\type$,
 and denotes a type quantifying over type variables.
}
The set of type variables occurring in a type $\type$ is denoted by $\FTV{\type}$.
\revised{Type variables $\overline\tvar$ occurring in a type $\type$ can be substituted with concrete types $\overline\type$ through the common bracket-notation $\applySubstitution{\type}{\substitution{\overline\tvar}{\overline\type}}$.}

\begin{figure}[!t]{
 \framebox[1\textwidth]{
 $\begin{array}{l}
\textbf{Types:}\\
\begin{array}{rcl@{\hspace{7.5cm}}r}
\type & \BNFcce &  \tvar  \; \BNFmid \;  \builtintype \; \BNFmid \;  (\overline\type) \rightarrow \type        &   {\footnotesize \mbox{type}} \\
\typescheme & \BNFcce &  \forall\overline{\tvar}.\type      &   {\footnotesize \mbox{type scheme}} \\
%
\end{array}\\
\hline\\[-8pt]
\textbf{Expression typing:}
  \hfill
  \boxed{\expTypJud{\TStypEnv}{\TtypEnv}{\e}{\type}}
\vspace{0.1cm}
  \\
\begin{array}{c}
\nullsurfaceTyping{T-VAR}{
\expTypJud{\TStypEnv}{\TtypEnv,\xname:\type}{\xname}{\type}
}
\qquad
{
\surfaceTyping{T-DAT}{ \quad
\applySubstitution{\type'}{\substitution{\overline\tvar}{\overline\type''}} = (\overline{\type})\rightarrow\type
\qquad
\expTypJud{\TStypEnv}{\TtypEnv}{\overline{\anyvalue}}{\overline{\type}}
}{
\expTypJud{\TStypEnv,\dc: \forall \overline{\tvar}. \type'}{\TtypEnv}{\dc(\overline{\anyvalue})}{\type} }
}
\skiptransition
\surfaceTyping{T-A-FUN}{ \quad
\expTypJud{\TStypEnv}{\;\TtypEnv,\,\overline{\xname}:\overline{\type}}{\e}{\type}
}{ \expTypJud{\TStypEnv}{\TtypEnv}{ (\overline{\xname}) \toSymK{\e}}{(\overline{\type})\rightarrow\type} }
\qquad
{
\surfaceTyping{T-N-FUN}{ \quad
\mbox{$\funvalue$ is a (built-in or declared) function}}{
\expTypJud{\TStypEnv,\funvalue: \forall \overline{\tvar}. \type}{\TtypEnv}{\funvalue}{\applySubstitution{\type}{\substitution{\overline\tvar}{\overline\type}}} }
}
\skiptransition
\surfaceTyping{T-APP}{ \quad
\expTypJud{\TStypEnv}{\TtypEnv}{\e}{(\overline{\type})\rightarrow\type} \qquad
\expTypJud{\TStypEnv}{\TtypEnv}{\overline{\e}}{\overline{\type}} }{
\expTypJud{\TStypEnv}{\TtypEnv}{\e(\overline{\e})}{\type} }
\skiptransition
%
\surfaceTyping{T-REP}{ \qquad
\expTypJud{\TStypEnv}{\TtypEnv}{\e_1}{\type}
\qquad \expTypJud{\TStypEnv}{\TtypEnv}{\e_2}{(\type)\to \type} }{
\expTypJud{\TStypEnv}{\TtypEnv}{\repK(\e_1)\{\e_2\}}{\type} }
\qquad\qquad
\surfaceTyping{T-NBR}{ \qquad
\expTypJud{\TStypEnv}{\TtypEnv}{\e}{\type}
}{ \expTypJud{\TStypEnv}{\TtypEnv}{\nbrK\{\e\}}{\type} }
\skiptransition
\surfaceTyping{T-FOLD}{ \qquad
\expTypJud{\TStypEnv}{\TtypEnv}{\e_1}{\type}
\quad \expTypJud{\TStypEnv}{\TtypEnv}{\e_2}{(\type, \type) \to \type}
\quad \expTypJud{\TStypEnv}{\TtypEnv}{\e_3}{\type} }{
\expTypJud{\TStypEnv}{\TtypEnv}{\foldK(\e_1, \e_2, \e_3) }{\type} }
\skiptransition
\end{array}
\\
\textbf{Function typing:}
  \hfill
  \boxed{\funTypJud{\TStypEnv}{\FUNCTION}{\typescheme}}
  \\
\begin{array}{c}
\surfaceTyping{T-FUNCTION}{
\qquad
\expTypJud{\TStypEnv,\,\fname:\forall\emptyseq.(\overline{\type})\rightarrow\type}{\overline{\xname}:\overline{\type}}{\e}{\type}
\qquad
\overline{\tvar}=\FTV{(\overline{\type})\rightarrow\type}
}{ \funTypJud{\TStypEnv}{\defK \; \fname (\overline{\xname}) \; \eqSymK{\e}}{\forall\overline{\tvar}.(\overline{\type})\rightarrow\type}}
\skiptransition
\end{array}
\\
\textbf{Program typing:}
  \hfill
  \boxed{\proTypJud{\PROGRAM}{\type}}
  \\
\begin{array}{c}
\surfaceTyping{T-PROGRAM}{
\\
\TStypEnv_0=\OStypEnv
\\
\FUNCTION_i = \defK \; \fname_i (\_) \; \eqSymK{\_}
\qquad
\funTypJud{\TStypEnv_{i-1}}{\FUNCTION_i}{\typescheme_i}
\qquad
\TStypEnv_i=\TStypEnv_{i-1},\, \fname_i:\typescheme_i
\qquad
 (i \in 1..n)
\\
\expTypJud{\TStypEnv_n}{\emptyset}{\e}{\type}
}{ \proTypJud{\FUNCTION_1\cdots\FUNCTION_n  \;
\e}{\type}}
\end{array}
\end{array}$}
} \caption{Type rules for expressions, function declarations, and programs.} \label{fig:SurfaceTyping}
\end{figure}

\emph{Type environments}, ranged over by $\TtypEnv$ and written $\overline{\xname}:\overline{\type}$,
are used to collect type assumptions for program variables {(i.e., formal parameters of functions).}
\emph{Type-scheme environments}, ranged over by $\TStypEnv$ and written $\overline{\anyvalue}:\overline{\typescheme}$, are used to collect the type schemes for built-in constructors and built-in operators together with the type schemes inferred for user-defined functions. In particular,
the distinguished  \emph{built-in type-scheme environment}  $\OStypEnv$ associates a type scheme
{to each  built-in constructor $\dc$ and}
to each built-in function $\bname$---Figure~\ref{fig:typeof-built-in} shows the type schemes for the {built-in constructors and} built-in functions
 used throughout this paper.

\begin{figure}[t]{
\centerline{\framebox[\textwidth]{ $\begin{array}{l}
\textbf{Built-in data constructors}\\
\begin{array}{lcl@{\hspace{0.2cm}}r}
\typeof{\truevalue} & = & \btype
 \\
\typeof{\falsevalue} & = & \btype
 \\
\typeof{n} & = & \ntype, \quad \text{where } n \text{ is a number or } \texttt{PositiveInfinity}
 \\
\typeof{\PairK} & = & \forall\tvar_1\tvar_2.(\tvar_1,\tvar_2)\rightarrow \pairltypeOf{\tvar_1}{\tvar_2}
 \\
\typeof{\NullK} & = & \forall\tvar. \listltypeOf{\tvar}
 \\
\typeof{\ConsK} & = & \forall\tvar. (\tvar,\listltypeOf{\tvar})\to \listltypeOf{\tvar}
\\
 \end{array}
 \\
 \textbf{Built-in functions: pure operators}\\
\begin{array}{lcl@{\hspace{0.2cm}}r}
\typeof{\pairK} & = & \forall\tvar_1\tvar_2. \tvar_1 \to \tvar_2 \to \pairltypeOf{\tvar_1}{\tvar_2}
 \\
\typeof{\fstK} & = & \forall\tvar_1\tvar_2.(\pairltypeOf{\tvar_1}{\tvar_2}) \to \tvar_1
 \\
\typeof{\sndK} & = & \forall\tvar_1\tvar_2.(\pairltypeOf{\tvar_1}{\tvar_2}) \to \tvar_2
 \\
\typeof{\consK} & = & \forall\tvar. \tvar\to \listltypeOf{\tvar} \to \listltypeOf{\tvar}
 \\
\typeof{\headK} & = & \forall\tvar.(\listltypeOf{\tvar}) \rightarrow \tvar
 \\
\typeof{\tailK} & = & \forall\tvar.(\listltypeOf{\tvar}) \rightarrow \listltypeOf{\tvar}
 \\
\typeof{\texttt{=}} & = & \forall\tvar.(\tvar,\tvar) \to \btype
\\
\typeof{\texttt{mux}} & = & \forall\tvar.(\btype,\tvar,\tvar)\to\tvar
\\
\typeof{\texttt{+}} & = & (\ntype,\ntype)\to\ntype
\\
\typeof{\texttt{and}} & = & (\btype,\btype)\to\btype
 \\
\typeof{\texttt{min}} & = & (\ntype, \ntype) \to \ntype
 \\
\typeof{\texttt{<}} & = & (\ntype,\ntype)\to \btype
\\
 \end{array}
   \\
 \textbf{Built-in functions: sensors}\\
\begin{array}{lcl@{\hspace{0.2cm}}r}
\typeof{\texttt{temperature}} & = & () \to \ntype
\\
 \end{array}
    \\
 \textbf{Built-in functions: relational sensors}\\
\begin{array}{lcl@{\hspace{0.2cm}}r}
\typeof{\texttt{nbrRange}} & = & () \to \ntype
\\
 \end{array}
\end{array}
 $}}}
\caption{Type schemes for the built-in value constructors and functions used in the examples.}
\label{fig:typeof-built-in}
\end{figure}

The typing judgement for expressions is of the form ``$\expTypJud{\TStypEnv}{\TtypEnv}{\e}{\type}$'', to be read:
``$\e$ has type $\type$ under the type-scheme assumptions $\TStypEnv$ ({for built-in constructors and} for  built-in and user-defined functions)
and the type assumptions $\TtypEnv$ (for the program variables occurring in $\e$), respectively''.
As a standard syntax in type systems~\cite{FJ}, given
$\overline{\type}=\type_1,\ldots,\type_n$ and
$\overline{\e}=\e_1,\ldots,\e_n$ ($n\ge 0$), we write
$\expTypJud{\TStypEnv}{\TtypEnv}{\overline{\e}}{\overline{\type}}$
as short for $\expTypJud{\TStypEnv}{\TtypEnv}{\e_1}{\type_1}$
$\cdots$ $\expTypJud{\TStypEnv}{\TtypEnv}{\e_n}{\type_n}$.

{The typing rules for expressions are presented in Figure~\ref{fig:SurfaceTyping} (bottom).
The rules for variables (\ruleNameSize{[T-VAR]}), data values (\ruleNameSize{[T-DAT]}), anonymous function expressions (\ruleNameSize{[T-A-FUN]}),
built-in or defined function names (\ruleNameSize{[T-N-FUN]}), and function application (\ruleNameSize{[T-APP]}),  are almost standard.
Rule \ruleNameSize{[T-REP]} (for $\repK$-expressions) ensures that both the initial value $\e_1$  and the domain and range of function $\e_2$
have the same type, and then assigns it to $\repK(\e_1)\{\e_2\}$;
rule \ruleNameSize{[T-NBR]} (for $\nbrK$-expressions) assigns to $\nbrK\{\e\}$ the same type as $\e$; and rule
\ruleNameSize{[T-FOLD]} (for $\foldK$-expressions) ensures that
  $\e_1$ and  $\e_3$ have the same type $\type$ and that $\e_2$ has type $(T, T) \to T$, and then assigns
 type $\type$ to $\foldK(\e_1,\e_2,\e_3)$.}
The typing rules for declared functions (\ruleNameSize{[T-FUNCTION]})
and programs (\ruleNameSize{[T-PROGRAM]}) are almost standard
 \revised{and leverage judgements of the form $\funTypJud{\TStypEnv}{\FUNCTION}{\typescheme}$ and $\funTypJud{\LTStypEnv_0}{\PROGRAM}{\type}$, respectively for a function $\FUNCTION$ and program $\PROGRAM$}.

\begin{example}[Typing]\label{exa:typing}
Consider the following simple implementation of a self-healing gradient~\cite{ACDV-SASO2017} (i.e., the field of minimum distances from source devices) that also circumvents obstacles. 
Its semantics will be presented later (see \Cref{ex:time-nbr-gradient,exa:NIandB-formal}); now, we merely consider its typing.
\begin{lstlisting}[language={hfc}]
def gradient(source, metric) { // : (bool, ()->num) -> num
  rep(PositiveInfinity){ (distance) => {
    mux(source, 0.0,
      foldhood(PositiveInfinity, min, nbr{distance} + metric())
    )
  }}
}
if (isObstacle) { PositiveInfinity } { gradient(isSource) }  // : num
\end{lstlisting}
The types of the gradient function and of the main expression inferred by the type system are inserted above as comments. By rule \ruleNameSize{[T-APP]} and assumptions on built-in \texttt{mux}, the type system infers that the third argument of the \texttt{mux} expression must be $\ntype$, since the second argument is also $\ntype$. It follows that \texttt{distance} must be of type $\ntype$ (rule \ruleNameSize{[T-NBR]}) as well and \texttt{metric} must be of type $() \to \ntype$ (rule \ruleNameSize{[T-APP]}), from which function \texttt{gradient} can be inferred to have type $(\btype, () \to \ntype) \to \ntype$ (rule \ruleNameSize{[T-FUNCTION]}). The overall program has then type $\ntype$ (rule \ruleNameSize{[T-PROGRAM]}).
\end{example}

\subsection{Operational Semantics: Device Semantics}\label{sec-calculus-device-semantics}

This section presents a formal semantics of device computation as happens in \FSCAFI{}. 
Starting from \FSCAFI{} syntax as previously described, we  assume a fixed program $\PROGRAM$. We say that ``device $\deviceId$ \emph{fires}'', to mean that the main expression of $\PROGRAM$ is evaluated on $\deviceId$.

\begin{rem}[On termination of device computation] \label{rem:termination}
As \FSCAFI\ allows recursive functions, termination of a device firing is not decidable. In the rest of the paper we assume
that only terminating programs are considered.
\end{rem}

\subsubsection{Device semantics: overall picture and preliminary definitions}\label{sec-calculus-device-semantics-overall}

We model device computation by a big-step operational semantics
where the result of evaluation is a \emph{value-tree} $\vtree$  (see \Cref{fig:deviceSemantics}, first frame),
which is an ordered tree of values, tracking the result of any evaluated subexpression. Intuitively, the evaluation of an expression at a given time in a device $\deviceId$ is performed against the recently-received value-trees of neighbours, namely, its outcome depends on those value-trees.
The result is a new value-tree that is conversely made available to $\deviceId$'s neighbours (through a broadcast) for their firing; this includes $\deviceId$ itself, so as to support a form of state across computation rounds (note that an implementation may massively compress the value-trees, storing only enough information for expressions to be aligned).
\revised{
In other words, a value-tree is a data structure
 that embeds all of the following:
 (i) the final and partial results of a local computation;
 (ii) state resulting from a previous computation;
 (iii) incoming data from neighbours; and
 (iv) outcoming data to be sent to neighbours.
Trees are in general a suitable representation
 for denoting programs (statically)
 and computations (dynamically); notice that
 value-trees may differ in structure
 if different sub-computations are chosen
 e.g. due to branching (cf. $\ifK$).
In particular, state and communication data
 are located at specific nodes -- corresponding to $\repK$ and $\nbrK$ construct applications --
 which are the only ones that may be actually stored and transmitted, but
 whose position
 is important for the evaluation semantics.
}

A \emph{value-tree environment} $\Trees$ is a map from device identifiers to value-trees, \revisedMinor{a local view of} the outcomes of the last evaluations on neighbours, \revisedMinor{as it is locally reconstructed through received messages (no global knowledge is ever assumed).}
This is written $\envmap{\overline\deviceId}{\overline\vtree}$ as short for $\envmap{\deviceId_1}{\vtree_1},\ldots,\envmap{\deviceId_n}{\vtree_n}$.
The syntax of value-trees and value-tree environments is given in Figure~\ref{fig:deviceSemantics} (first frame).

\begin{example}[Value-trees]\label{exa:value-trees}
The graphical representation of the value-trees
\[
  \mbox{$\vtree_1=\mkvt{\truevalue}{\mkvt{<}{},\mkvt{-2}{},\mkvt{5}{},\mkvt{\truevalue}{}}$ $\qquad$
and $\qquad$ $\vtree_2=\mkvt{4}{\mkvt{f}{},\mkvt{3}{},\mkvt{4}{\mkvt{+}{},\mkvt{3}{},\mkvt{1}{},\mkvt{4}{}}}$}
\]
 is as follows:

\noindent\begin{minipage}{\textwidth}
\centering
\begin{tikzpicture}[-,node distance=1.5cm and 1.5cm, node/.style={rectangle,draw}, node2/.style={draw}]
  \node[node] (a1) [] {\texttt{True}};
  \node[node] (a2) [below left=0.5cm and 1.5cm of a1] {\texttt{<}};
  \node[node] (a3) [below left=0.5cm and 0cm of a1] {\texttt{-2}};
  \node[node] (a4) [below right=0.5cm and 0cm of a1] {\texttt{5}};
  \node[node] (a5) [below right=0.5cm and 1.5cm of a1] {\texttt{True}};

 \node[node] (b1) [right=6cm of a1] {\texttt{4}};
 \node[node] (b2) [below left=0.5cm and 1.5cm of b1] {\texttt{f}};
 \node[node] (b3) [below=0.5cm of b1] {\texttt{3}};
 \node[node] (b4) [below right=0.5cm and 1.5cm of b1] {\texttt{4}};
 \node[node] (b41) [below left=0.5cm and 1.5cm of b4] {\texttt{+}};
 \node[node] (b42) [below left=0.5cm and 0cm of b4] {\texttt{3}};
 \node[node] (b43) [below right=0.5cm and 0cm of b4] {\texttt{1}};
 \node[node] (b44) [below right=0.5cm and 1.5cm of b4] {\texttt{4}};


  \path 
    (a1) edge [] node [] {} (a2)
    (a1) edge [] node [] {} (a3)
    (a1) edge [] node [] {} (a4)
    (a1) edge [] node [] {} (a5)
    (b1) edge [] node [] {} (b2)
    (b1) edge [] node [] {} (b3)
    (b1) edge [] node [] {} (b4)
    (b4) edge [] node [] {} (b41)
    (b4) edge [] node [] {} (b42)
    (b4) edge [] node [] {} (b43)
    (b4) edge [] node [] {} (b44)
    ;
\end{tikzpicture}
\end{minipage}
\end{example}
In the following, for sake of readability, we sometimes write the value $\anyvalue$ as shorthand for the value-tree $\mkvt{\anyvalue}{}$. Following this convention, the value-tree $\vtree_1$ is shortened to $\mkvt{\truevalue}{<,-2,5,\truevalue}$,  and the value-tree $\vtree_2$ is shortened to $\mkvt{4}{\texttt{f},3,\mkvt{4}{+,3,1,4}}$.


We assume that prior to execution each anonymous function sub-expression $(\overline\xname) \toSymK{\e}$ of a program $\emain$ is automatically annotated as $(\overline\xname) \toSymKtag{\name}{\e}$ with a tag $\name$ which uniquely identifies the expression.\footnote{For example, the tag could be generated as $\name = (\emain, n)$ where $n$ is the index of the occurrence of the $\toSym{}$ keyword in $\emain$.} The tag serves as a \emph{name} for anonymous functions, \revised{allowing to consider} function values \revised{to be} equal if they share the same name. \revised{This notion of syntactic equality will later be exploited for the definition of the semantics of function calls.}

\revised{
\begin{rem}[Function Equality]
	Notice that \emph{behavioural equivalence} of function expressions is not decidable in a Turing-complete language, and thus cannot be used in a formal semantics of a (practically executable) programming language. Pure syntactic equality could instead be used, but it would fail to correctly relate multiple occurrences of an anonymous function capturing a variable, whenever that gets substituted for different values. The proposed naming system prevents those issues, while also avoiding unintended interferences between functions that happen to be syntactically identical, but were written in different units of a program and were not meant to be integrated with each other.
\end{rem}
}

Figure~\ref{fig:deviceSemantics} (second frame) defines: the auxiliary functions $\rho$ and $\pi$ for extracting the root value and a subtree of a value-tree, respectively (further explanations about function $\pi$ will be given later); the extension of functions $\rho$ and $\pi$ to value-tree environments; and the auxiliary functions $\nameOf$, $\argsNAME$ and $\bodyNAME$ for extracting the name, formal parameters and body of a (user-defined or anonymous) function, respectively.

\revised{
\begin{example}[Auxiliary functions on value-trees]\label{exa:value-trees-aux}
Consider the following value-trees.
\[
  \mbox{$\vtree_1=\mkvt{\truevalue}{\mkvt{<}{},\mkvt{-2}{},\mkvt{5}{},\mkvt{\truevalue}{}}$ $\qquad$
and $\qquad$ $\vtree_2=\mkvt{4}{\mkvt{f}{},\mkvt{3}{},\mkvt{4}{\mkvt{+}{},\mkvt{3}{},\mkvt{1}{},\mkvt{4}{}}}$}
\]
The following picture graphically shows
 the nodes
 denoted by
 auxiliary functions working on individual value trees.

\noindent\begin{minipage}{\textwidth}
\centering
\begin{tikzpicture}[-,node distance=1.5cm and 1.5cm, node/.style={rectangle,draw}, node2/.style={draw}]
  \node[node] (a1) [] {\texttt{True}};
  \node[node] (a2) [below left=0.5cm and 1.5cm of a1] {\texttt{<}};
  \node[node] (a3) [below left=0.5cm and 0cm of a1] {\texttt{-2}};
  \node[node] (a4) [below right=0.5cm and 0cm of a1] {\texttt{5}};
  \node[node] (a5) [below right=0.5cm and 1.5cm of a1] {\texttt{True}};

 \node[node] (b1) [right=6cm of a1] {\texttt{4}};
 \node[node] (b2) [below left=0.5cm and 1.5cm of b1] {\texttt{f}};
 \node[node] (b3) [below=0.5cm of b1] {\texttt{3}};
 \node[node] (b4) [below right=0.5cm and 1.5cm of b1] {\texttt{4}};
 \node[node] (b41) [below left=0.5cm and 1.5cm of b4] {\texttt{+}};
 \node[node] (b42) [below left=0.5cm and 0cm of b4] {\texttt{3}};
 \node[node] (b43) [below right=0.5cm and 0cm of b4] {\texttt{1}};
 \node[node] (b44) [below right=0.5cm and 1.5cm of b4] {\texttt{4}};

\node [above=0cm of a1] {$\revised{\vrootOf{\vtree_1}}$};
\node (a5b) [below=0.4cm of a2] {$\revised{\piIof{1}{\vtree_1}}$};
%
\node [below=0.2cm of a5b, xshift=3.5cm] {$\revised{\piBof{+}{\vtree_1} = \bullet}$};
\node [above=0cm of b1] {$\revised{\vrootOf{\vtree_2}}$};
\node [above right=0cm of b4,xshift=0.4cm] {$\revised{\piIof{3}{\vtree_2} = \vtree_3}$};
\node [above right=0.2cm of b44] {$\revised{\piBof{+}{\vtree_3}}$};

\begin{scope}[on background layer]
    \filldraw[gray	, line join=round, line width=1cm, fill opacity=0.1,text opacity=1.0, opacity=0.2] plot coordinates{(a2.north west) (a2.north east) (a2.south east) (a2.south west)}--cycle;
    \filldraw[gray	, line join=round, line width=1cm, fill opacity=0.1,text opacity=1.0, opacity=0.2] plot coordinates{(b4.north east) (b44.south east) (b41.south west) (b4.north west)}--cycle;
    \filldraw[purple	, line join=round, line width=1cm, fill opacity=0.1,text opacity=1.0, opacity=0.2] plot coordinates{(b44.north west) (b44.north east) (b44.south east) (b44.south west)}--cycle;
\end{scope}


  \path 
    (a1) edge [] node [] {} (a2)
    (a1) edge [] node [] {} (a3)
    (a1) edge [] node [] {} (a4)
    (a1) edge [] node [] {} (a5)
    (b1) edge [] node [] {} (b2)
    (b1) edge [] node [] {} (b3)
    (b1) edge [] node [] {} (b4)
    (b4) edge [] node [] {} (b41)
    (b4) edge [] node [] {} (b42)
    (b4) edge [] node [] {} (b43)
    (b4) edge [] node [] {} (b44)
    ;
\end{tikzpicture}
\end{minipage}
\end{example}
}

The computation that takes place on a single device is formalised  by the big-step operational semantics rules given in Figure~\ref{fig:deviceSemantics} (fourth frame).
 The derived judgements are of the form
\[
\bsopsem{\deviceId,\deviceId'}{\Trees}{\senstate}{\e}{\vtree}
\]
to be read ``expression $\e$ evaluates to  value-tree $\vtree$ on device $\deviceId$ with respect to the neighbour $\deviceId'$, value-tree environment $\Trees$ and sensor state $\senstate$'', where:
        \emph{(i)} $\deviceId$ is the identifier of the current device and $\deviceId'$ is either equal to $\deviceId$ or is one of its neighbours;
        \emph{(ii)} $\Trees$ is the field of the value-trees produced by the most recent evaluation of (an expression corresponding to) $\e$ on $\deviceId$ and its neighbours, \revisedMinor{which were received by $\deviceId$ through messages};
        \emph{(iii)} $\e$ is an expression;
        \emph{(iv)} the value-tree $\vtree$  represents the values computed for all the expressions encountered during the evaluation of $\e$---in particular $\vrootOf{\vtree}$ is the result value of $\e$.

The operational semantics  rules  are based on rather standard rules for functional languages,
extended so as to be able to evaluate a subexpression $\e'$ of $\e$ with respect to\ the value-tree environment $\Trees'$ obtained from $\Trees$
by extracting the corresponding subtree (when present) in the value-trees in the range of $\Trees$. This process, called \emph{alignment},
 is modelled by the auxiliary function $\pi$, defined in Figure~\ref{fig:deviceSemantics} (second frame).
Function $\pi$ has two different behaviours (specified by its subscript or superscript): $\piIof{i}{\vtree}$ extracts the $i$-th subtree of $\vtree$, if it is present;
and $\piBof{\funvalue}{\vtree}$ extracts the last subtree of $\vtree$, if it is present and the root of first subtree of $\vtree$ is equal to $\funvalue$.

When a device $\deviceId$ fires, its main expression $\e$ is evaluated with respect to $\deviceId$ itself. That is, by means of a judgement where
$\deviceId'=\deviceId$:
\[
\bsopsem{\deviceId,\deviceId}{\Trees}{\senstate}{\e}{\vtree}.
\]
A key aspect of the semantics is that, if $\e$ is a $\foldK$-expression $\foldK(\e_1,\e_2,\e_3)$ then its body $\e_3$ is evaluated with respect to each of the devices $\deviceId'$  (if any)  in $\domof{\Trees}\setminus\{\deviceId\}$.
Because of alignment (see above), it might happen that a sub-expression $\e'$ of $\e_3$ is evaluated by a judgement
\[
\bsopsem{\deviceId,\deviceId'}{\Trees}{\senstate}{\e'}{\vtree}  \quad \mbox{ where $\deviceId\neq\deviceId'\not\in\domof{\Trees}$ }
\]
and, if the evaluation of $\e'$ exploits the device $\deviceId'$, then the evaluation of $\e_3$ with respect to $\deviceId'$ \emph{fails} and
the evaluation of the $\foldK$-expression $\foldK(\e_1,\e_2,\e_3)$ does not consider the neighbour $\deviceId'$. The
evaluation rule for $\foldK$-expressions, \ruleNameSize{[E-FOLD]},  formalises failure of evaluation with respect to a neighbour $\deviceId'$
by means of the auxiliary predicate
\[
\bsopsemFAIL{\deviceId,\deviceId'}{\Trees}{\senstate}{\e}
\]
  to be read ``expression $\e$ fails to evaluate on device $\deviceId$  against neighbour $\deviceId'$ with respect to value-tree environment $\Trees$ and
sensor state $\senstate$'', which is formalised by the big-step operational semantics rules given in Figure~\ref{fig:deviceSemanticsFAIL}.

\begin{figure}[!t]{
 \framebox[1\textwidth]{
 $\begin{array}{l}
\textbf{Value-trees
and value-tree environments:}\\
\begin{array}{rcl@{\hspace{8.1cm}}r}
	\vtree & \BNFcce &  \mkvtree{E-Rule}{\anyvalue}{\overline{\vtree}}
	&   {\footnotesize \mbox{value-tree}} \\
	\Trees & \BNFcce & \envmap{\overline{\deviceId}}{\overline{\vtree}}
	 &   {\footnotesize \mbox{value-tree environment}} \\
\end{array}\\
\hline\\[-8pt]
\textbf{Auxiliary functions:}\\
\begin{array}{l}
\begin{array}{l@{\hspace{8pt}}l@{\hspace{8pt}}l}
	\nameOf((\overline{\xname}) \; \toSymKtag{\name}{\e}) = \name
	&
	\args{ (\overline{\xname}) \; \toSymKtag{\name}{\e}} = \overline{\xname}
	&
	\body{(\overline{\xname}) \; \toSymKtag{\name}{\e}} = \e
	\\
	\nameOf(\fname) = \fname
	&
	\args{\fname} = \overline{\xname}
	&
	\body{\fname} = \e  \; (\text{if } \, \defK \; \fname (\overline{\xname}) \; \eqSymK{\e})
	\\
	\nameOf(\bname) = \bname
	&
	\vrootOf{\mkvtree{}{\anyvalue}{\overline{\vtree}}}  =   \anyvalue
	\\
	\piIof{i}{\mkvtree{E-Rule}{\anyvalue}{\vtree_1,\ldots,\vtree_n}}  =  \vtree_i
	&
	\text{if~} 1\le i \le n
	&
	\text{~else~} \emptyseq
	\\
	\piBof{\funvalue}{\mkvtree{E-Rule}{\anyvalue}{\vtree_1,\ldots,\vtree_{n+2}}}  =   \vtree_{n+2}
	&
	\text{if~} \nameOf(\vrootOf{\vtree_{1}}) = \nameOf(\funvalue)
	&
	\text{~else~} \emptyseq
\end{array} \!\!\!\!\!\!\!\!\!\!\!\!\!\!\!
\\[30pt]
\mbox{For } \auxNAME\in\rho,\piI{i},\piB{\funvalue}:
\quad
\left\{{\begin{array}{lcll}
	 \aux{\emptyseq}  & =  & \emptyseq
	\\
	\aux{\envmap{\deviceId}{\vtree}, \Trees}  & =   & \aux{\Trees}  & \quad \mbox{if} \; \aux{\vtree}=\emptyseq
	\\
	 \aux{\envmap{\deviceId}{\vtree}, \Trees}  & =  & \envmap{\deviceId}{\aux{\vtree}}, \aux{\Trees} & \quad \mbox{if} \; \aux{\vtree} \not=\emptyseq
\end{array}}\right.   \!\!\!\!\!\!
\end{array}\\
\hline\\[-10pt]
\textbf{Syntactic shorthands:}\\
\!\!\begin{array}{l@{\hspace{4pt}}l@{\hspace{4pt}}l}
\bsopsem{\deviceId,\deviceId'}{\piIofOv{\Trees}}{\senstate}{\overline{\e}}{\overline{\vtree}}
&
  \textrm{where~~} |\overline{\e}|=n
&
  \textrm{for~~}
  \bsopsem{\deviceId,\deviceId'}{\piIof{1}{\Trees}}{\senstate}{\e_1}{\vtree_1}
    \ldots
    \bsopsem{\deviceId,\deviceId'}{\piIof{n}{\Trees}}{\senstate}{\e_n}{\vtree_n} \!\!\!\!\!\!\!\!\!\!\!\! \\
\vrootOf{\overline{\vtree}}
&
  \textrm{where~~} |\overline{\vtree}|=n
  & \textrm{for~~}
\vrootOf{\vtree_1} ~ \ldots ~ \vrootOf{\vtree_n}\\
\substitution{\overline{\xname}}{\vrootOf{\overline{\vtree}}}
&   \textrm{where~~} |\overline{\xname}|=n
  &
  \textrm{for~~}
\substitution{\xname_1}{\vrootOf{\vtree_1}}~\ldots~\substitution{\xname_n}{\vrootOf{\vtree_n}}
\end{array}\\
\hline\\[-10pt]
\textbf{Rules for expression evaluation:} \hspace{55mm}
 \boxed{\bsopsem{\deviceId,\deviceId'}{\Trees}{\senstate}{\e}{\vtree}}
\skiptransitionR
\!\!\!\!\!\begin{array}{c}
\nullsurfaceTyping{E-VAL}{
\bsopsem{\deviceId,\deviceId'}{\Trees}{\senstate}{\anyvalue}{\mkvtree{}{\anyvalue}{}}
}
\skiptransitionN\\
\surfaceTyping{E-B-APP}{
\begin{array}{ll}
  \bsopsem{\deviceId,\deviceId}{\piIof{1}{\Trees}}{\senstate}{\e}{\vtree}
                    &
          \bsopsem{\deviceId,\deviceId'}{\piIof{i+1}{\Trees}}{\senstate}{\e_i}{\vtree_i}  \quad \text{for all}\; i \in 1, \ldots, n
      \\
         \anyvalue=\builtinop{\bname}{\deviceId,\deviceId'}{\piBof{\bname}{\Trees},\senstate}(\vrootOf{\overline{\vtree}})
         &
         (\bname = \vrootOf{\vtree} \text{ is not relational }) \vee (\deviceId' \in \domof{\piBof{\bname}{\Trees}} \cup \{\deviceId\})
\end{array} \!\!\!\!
 }{
\bsopsem{\deviceId,\deviceId'}{\Trees}{\senstate}{\e(\overline{\e})}{\mkvtree{}{\anyvalue}{\vtree,\overline{\vtree},\anyvalue}}
}
\skiptransitionN\\
\surfaceTyping{E-D-APP}{
\begin{array}{ll}
  \bsopsem{\deviceId,\deviceId}{\piIof{1}{\Trees}}{\senstate}{\e}{\vtree}
&
\bsopsem{\deviceId,\deviceId'}{\piIof{i+1}{\Trees}}{\senstate}{\e_i}{\vtree_i} \quad \text{for all}\; i \in 1, \ldots, n

\\
  \funvalue = \vrootOf{\vtree} \mbox{ is not a built-in}
 &
  \bsopsem{\deviceId,\deviceId'}{\piBof{\funvalue}{\Trees}}{\senstate}{\applySubstitution{\body{\funvalue}}{\substitution{\args{\funvalue}}{\vrootOf{\overline{\vtree}}}}}{\vtree'}
\end{array} \!\!\!\!
 }{
\bsopsem{\deviceId,\deviceId'}{\Trees}{\senstate}{\e(\overline{\e})}{\mkvtree{}{\vrootOf{\vtree'}}{\vtree,\overline{\vtree},\vtree'}}
}
\skiptransitionN\\
\surfaceTyping{E-REP}{
        \begin{array}{ll}
     \bsopsem{\deviceId,\deviceId}{\piIof{1}{\Trees}}{\senstate}{\e_1}{\vtree_1} & \anyvalue_1=\vrootOf{\vtree_{1}}\\
     \bsopsem{\deviceId,\deviceId}{\piIof{2}{\Trees}}{\senstate}{\e_2(\anyvalue_0)}{\vtree_2}~~& \anyvalue_2=\vrootOf{\vtree_{2}}
        \end{array}
        \quad
        \anyvalue_0 = \left\{\begin{array}{ll}
                             \vrootOf{\piIof{2}{\Trees}}(\deviceId) & \mbox{if} \;  \deviceId \in \domof{\Trees} \\
                             \anyvalue_1 & \mbox{otherwise}
                           \end{array}\right.
 }{
\bsopsem{\deviceId,\deviceId'}{\Trees}{\senstate}{\repK(\e_1)\{\e_2\}}{\mkvtree{}{\anyvalue_2}{\vtree_1,\vtree_2}}
}
\skiptransitionN\\
\surfaceTyping{E-NBR}{
         \quad
        \deviceId \neq \deviceId' \in \domof{\Trees}
\qquad
        \vtree = \Trees(\deviceId')
 }{
\bsopsem{\deviceId,\deviceId'}{\Trees}{\senstate}{\nbrK\{\e\}}{\vtree}
}
\qquad
\surfaceTyping{E-NBR-LOC}{
         \quad
     \bsopsem{\deviceId,\deviceId}{\piIof{1}{\Trees}}{\senstate}{\e}{\vtree}
 }{
\bsopsem{\deviceId,\deviceId}{\Trees}{\senstate}{\nbrK\{\e\}}{\mkvtree{}{\vrootOf{\vtree}}{\vtree}}
}
\skiptransitionN\\
\surfaceTyping{E-FOLD}{
        \begin{array}{ll}
                \bsopsem{\deviceId,\deviceId}{\piIof{1}{\Trees}}{\senstate}{\e_1}{\vtree^1}
                &
                 \quad
               \bsopsem{\deviceId,\deviceId}{\piIof{2}{\Trees}}{\senstate}{\e_2}{\vtree_f}
                  \qquad
               \funvalue = \vrootOf{\vtree_f}
                 \\
                 \deviceId_1, \ldots, \deviceId_n = \domof{\Trees} \setminus \{\deviceId\}
               &
               \quad
               n\ge m\ge 0, \deviceId_1, \ldots, \deviceId_m \text{ increasing}, \deviceId_0 = \deviceId
                \\
                 \bsopsem{\deviceId,\deviceId_i}{\piIof{3}{\Trees}}{\senstate}{\e_3}{\vtree_i}
                &
                \quad
               \text{for all } i \in 0,...,m
                \\
                \bsopsemFAIL{\deviceId,\deviceId_j}{\piIof{3}{\Trees}}{\senstate}{\e_3}{}
                &
                 \quad
               \text{for all } j \in m+1,...,n
                \\
                  \bsopsem{\deviceId,\deviceId}{\emptyset}{\senstate}{\funvalue(\vrootOf{\vtree^i}, \vrootOf{\vtree_i})}{\vtree^{i+1}}
                &
                \quad
                \text{for all } i \in 1,...,m
        \end{array}
 }{
\bsopsem{\deviceId,\deviceId'}{\Trees}{\senstate}{\foldK(\e_1,\e_2,\e_3)}{\mkvtree{}{\vrootOf{\vtree^{m+1}}}{\vtree^1, \vtree_f, \vtree_0}}
}
\end{array}\!\!\!\!\!
\end{array}$}
}
\vspace{-0.1cm}
 \caption{Big-step operational semantics for expression evaluation (see also \Cref{fig:deviceSemanticsFAIL}).} \label{fig:deviceSemantics}
\end{figure}

\begin{figure}[!t]{
 \framebox[1\textwidth]{
 $\begin{array}{l}
~\textbf{Auxiliary rules for expression evaluation failure:} \hspace{23mm}
    \boxed{\bsopsemFAIL{\deviceId,\deviceId'}{\Trees}{\senstate}{\e}}
\skiptransition
\begin{array}{c}
\surfaceTyping{E-NBR-FAIL}{
    ~~ \quad
    \deviceId \neq \deviceId' \not\in \domof{\Trees}
 }{
	\bsopsemFAIL{\deviceId,\deviceId'}{\Trees}{\senstate}{\nbrK\{\e\}}
}
\skiptransitionN
\\
\surfaceTyping{E-R-APP-FAIL}{
	\begin{array}{ll}
	\bsopsem{\deviceId,\deviceId}{\piIof{1}{\Trees}}{\senstate}{\e}{\vtree}
      &
       \bsopsem{\deviceId,\deviceId'}{\piIof{i+1}{\Trees}}{\senstate}{\e_i}{\vtree_i} ~~ \text{for all}\; i \in 1, \ldots, n
      \\
      \snvalue = \vrootOf{\vtree} \mbox{ is relational-built-in}
      &
    \deviceId \neq \deviceId' \not\in \domof{\piBof{\snvalue}{\Trees}}
	\end{array}\!\!\!\!
 }{
	\bsopsemFAIL{\deviceId,\deviceId'}{\Trees}{\senstate}{\e(\overline\e)}
}
\skiptransitionN
\\
\surfaceTyping{E-APP-ARG-FAIL}{
	\begin{array}{ll}
	 \bsopsem{\deviceId,\deviceId}{\piIof{1}{\Trees}}{\senstate}{\e}{\vtree}
       &
      \quad
       \overline{\e} = \e_1,...,\e_n
        \qquad
         n\ge 0
	   \\
	  \bsopsem{\deviceId,\deviceId'}{\piIof{i+1}{\Trees}}{\senstate}{\e_i}{\vtree_i}
	&
	  \quad
	               \text{for all } i \in 1,...,m < n
	\\
	  \bsopsemFAIL{\deviceId,\deviceId'}{\piIof{m+2}{\Trees}}{\senstate}{\e_{m+1}}
	\end{array}
 }{
	\bsopsemFAIL{\deviceId,\deviceId'}{\Trees}{\senstate}{\e(\overline{\e})}
}
\skiptransitionN
\\
\surfaceTyping{E-D-APP-FAIL}{
	\begin{array}{ll}
		\bsopsem{\deviceId,\deviceId}{\piIof{1}{\Trees}}{\senstate}{\e}{\vtree}
       &
       \quad
       \bsopsem{\deviceId,\deviceId'}{\piIof{i+1}{\Trees}}{\senstate}{\e_i}{\vtree_i} ~~ \text{for all}\; i \in 1, \ldots, n
        \\
        \funvalue = \vrootOf{\vtree} \mbox{ is not a built-in}
        &
        \quad
	  \bsopsemFAIL{\deviceId,\deviceId'}{\piBof{\funvalue}{\Trees}}{\senstate}{\applySubstitution{\body{\funvalue}}{\substitution{\args{\funvalue}}{\vrootOf{\overline{\vtree}}}}} 
	\end{array}\!\!\!\!
 }{
	\bsopsemFAIL{\deviceId,\deviceId'}{\Trees}{\senstate}{\e(\overline{\e})}
}
\end{array}
\end{array}$}
}
\vspace{-0.1cm}
\caption{Big-step operational semantics for expression evaluation (auxiliary rules for expression evaluation failure).} \label{fig:deviceSemanticsFAIL}
\end{figure}

\subsubsection{Device semantics: rules for expression evaluation}\label{sec-calculus-device-semantics-expressions}

We start by explaining the rules in \Cref{fig:deviceSemantics} (fourth frame), then we will
explain the rules in \Cref{fig:deviceSemanticsFAIL}.

Rule \ruleNameSize{[E-VAL]} implements the evaluation of an expression that is already a value.
For instance, evaluating the expression $\mathtt{1}$ produces (by Rule  \ruleNameSize{[E-VAL]}) the value-tree $\mkvtree{}{1}{}$,
while evaluating the expression $\mathtt{+}$ produces the value-tree $\mkvtree{}{+}{}$.

Rules \ruleNameSize{[E-B-APP]} and \ruleNameSize{[E-D-APP]} model function application $\e(\e_1 \cdots \e_n)$.
In case $\e$ evaluates to a built-in function $\bname$, rule \ruleNameSize{[E-B-APP]} is used, whose behaviour is driven by the special auxiliary
function $\builtinop{\bname}{\deviceId,\deviceId'}{\Trees, \senstate}$ (operational interpretation of $\bname$), whose actual definition is abstracted away.

\begin{example}[Built-in function application]\label{exa:rule-E-B-APP}
Evaluating the expression $\mathtt{<(-2, 5)}$ produces the value-tree $\vtree_1=\mkvt{\truevalue}{<,-2,5,\truevalue}$ introduced in Example~\ref{exa:value-trees}.
The operational interpretation $\builtinop{<}{\deviceId,\deviceId'}{\Trees, \senstate}$ of $<$ is the following (notice that this interpretation does not depend on $\Trees, \senstate, \deviceId, \deviceId'$, since $<$ is a pure mathematical operator):
\[
\builtinop{<}{\deviceId,\deviceId'}{\Trees, \senstate} = \denotf{x}{} \denotf{y}{} \begin{cases}
\truevalue & x < y \\
\falsevalue & \text{otherwise}
\end{cases}
\]
The value of the whole expression, $\truevalue$ (the root of the last subtree of the value-tree), has been computed by using rule \ruleNameSize{[E-B-APP]} to evaluate the application $\builtinop{<}{\deviceId,\deviceId'}{\Trees, \senstate}(-2,5)$ of the less-then operator $<$ (the root of the first subtree of the value-tree) to the values $-2$  (the root of the second subtree of the value-tree) and $5$ (the root of the third subtree of the value-tree).
\end{example}
In case $\e$ evaluates to a user-defined or anonymous function $\funvalue$, rule \ruleNameSize{[E-D-APP]} is used:
it performs
domain restriction $\piBof{\funvalue}{\Trees}$ (thus discarding devices that did not apply the same function $\funvalue$, for which no consistent information on the application of $\funvalue$ is present),
then continues the evaluation by substituting the arguments into the body of $\funvalue$.
We remark that we do not assume that $\Trees$ is empty whenever it does not contain $\deviceId$.
In fact, in any round where $\e$ evaluates to a function $\funvalue$ for the first time on a device, $\funvalue(\overline\e)$
will be evaluated with respect to an environment not containing $\deviceId$ but possibly containing other devices (whose $\e$ evaluated to
$\funvalue$ in their previous round of computation).
\begin{example}[Defined or anonymous function application]\label{exa:rule-D-APP}
Evaluating the expression $\funvalue(3)$, where $\funvalue$ is the name of the declared function $\defK \; \funvalue(\xname) \; \eqSymK{\xname + 1}$, produces the value-tree $\vtree_2=\mkvt{4}{\texttt{f},3,\mkvt{4}{+,3,1,4}}$ introduced in Example~\ref{exa:value-trees}.
The value of the whole expression, $4$ (the root of  $\vtree_2$), which has been computed by using rule \ruleNameSize{[E-D-APP]}, is the root of the last subtree of $\vtree_2$, which is produced by the evaluation of the expression $3+1$ (obtained from the body of $\funvalue$ by replacing $\xname$ with $3$).
Evaluating the similar expression $\funvalue(3)$ where $\funvalue$ is the anonymous function $((\xname)\; \toSymKtag{\name}{\xname + 1})$, produces the same value-tree $\vtree_2$ by the same rule \ruleNameSize{[E-D-APP]}.
\end{example}

Rule \ruleNameSize{[E-REP]} implements internal state evolution through computational rounds: on the first firing of a device, $\repK(\e_1)\{\e_2\}$ evaluates to $\e_2(\e_1)$, then it evaluates to $\e_2(\anyvalue)$ where $\anyvalue$ is the value calculated in the previous round.

\begin{example}[Time evolution]\label{exa:rule-E-REP}
To illustrate rule \ruleNameSize{[E-REP]}, as well as computational rounds, we consider the program $\repK(3)\{\funvalue\}$
where $\funvalue$ is the anonymous function $(\xname)\; \toSymKtag{\name}{\xname + 1}$ introduced in Example~\ref{exa:rule-D-APP}.
The first firing of a device $\deviceId$  is performed against the empty tree environment.\revised{\footnote{\revised{In the first firing, no messages have been yet received, and therefore the environment is empty. This fact will be later formalised and discussed in the next section on network semantics.}}} Therefore, according to rule  \ruleNameSize{[E-REP]},
evaluating  $\repK(3)\{\funvalue\}$ produces the value-tree $\vtree=\mkvt{4}{3, \vtree_2}$ where $\vtree_2=\mkvt{4}{\texttt{f},3,\mkvt{4}{+,3,1,4}}$
is the value-tree (introduced in Example~\ref{exa:value-trees}) produced by evaluating the expression $\funvalue(3)$
as described in Example~\ref{exa:rule-D-APP}. The overall result of the firing is the root $4$ of $\vtree$.
Any subsequent firing of the device $\deviceId$ is performed with respect to a value-tree environment $\Trees$ that associates to $\deviceId$
 the outcome $\vtree$ of the most recent firing of $\deviceId$.
Therefore, evaluating    $\repK(3)\{\funvalue\}$ at the second firing produces the value-tree $\vtree'=\mkvt{5}{4, \vtree'_2}$ where
$\vtree'_2=\mkvt{5}{\texttt{f},4,\mkvt{5}{+,4,1,5}}$
is the value-tree  produced by evaluating the expression $\funvalue(4)$, where $4$ is the root of $\vtree$.
 Hence, the results of the firings are $4$, $5$, $6$, and so on.
\end{example}

Rules \ruleNameSize{[E-NBR]} and \ruleNameSize{[E-NBR-LOC]} model device interaction (together with \ruleNameSize{[E-FOLD]} which we shall consider later).  When an $\nbrK$-expression is not evaluated against a neighbour (that is, $\deviceId' = \deviceId$), by Rule \ruleNameSize{[E-NBR-LOC]} the $\nbrK$ operator is discarded and the evaluation continues. Whenever instead an $\nbrK$-expression is evaluated against a neighbour (that is, $\deviceId' \neq \deviceId$), by Rule \ruleNameSize{[E-NBR]} the expression directly evaluates to $\Trees(\deviceId')$ (which is the value-tree calculated by device $\deviceId'$ in its last computational round for the same expression). Notice that it could be possible that $\deviceId'$ is not in the domain of $\Trees$ due to alignment operations performed in subexpressions of the enclosing instance of $\foldK$. In this case, no rule is applicable and the $\nbrK$-expression \emph{fails}, causing $\deviceId'$ to be ignored by the enclosing $\foldK$ operator (see Rule \ruleNameSize{[E-FOLD]}).

Rule \ruleNameSize{[E-FOLD]} implements collection and aggregation of results from neighbours, proceeding in the following steps:
\begin{itemize}
        \item Evaluate the initial value $\e_1$ with respect to the current device obtaining the value-tree $\vtree^1$.
        \item Evaluate the aggregator $\e_2$ with respect to  the current device obtaining $\vtree_f$ with root $\funvalue$.
        \item Evaluate the body $\e_3$ with respect to the current device $\deviceId_0 = \deviceId$, obtaining $\vtree_0$ which constitutes the third branch of the overall resulting value-tree, then with respect to every \revised{one of the $n$ neighbours}, $\deviceId' \in \domof{\piIof{3}{\Trees}}\setminus\{\deviceId\}$, and \emph{consider} only the \revised{subset of $m$ neighbours ($0 \leq m \leq n$)}  $\deviceId_1,...,\deviceId_m$ \emph{for which the evaluation does not fail} \revised{(cf. discussion above and \Cref{fig:deviceSemanticsFAIL})}, obtaining the value-trees $\vtree_1,...,\vtree_m$, respectively.\footnote{\label{footnote:aggregator}Usually, the aggregator $\funvalue$ is associative and commutative, so that the result of the aggregation does not depend on the order in which the neighbours $\deviceId' \in \domof{\piIof{3}{\Trees}}\setminus\{\deviceId\}$ are considered. To ensure determinism even in the unlikely case of the aggregator $\funvalue$ being not associative and commutative, we assume that the neighbours are considered according to any given total order on device identifiers. \revised{Such an order may be, e.g., lexicographical order of the bit representations of the unique device identifiers $\deviceId$ (which may be, e.g., MAC addresses).}} \revisedMinor{We remark that the evaluation with respect to neighbours $\deviceId'$ is performed locally, with respect to the locally-available information in $\Trees$.}
        \item Aggregate the values $\vrootOf{\vtree_i}$ ($1\le i\le m$) computed above together with the initial value $\vrootOf{\vtree^1}$ via function $\funvalue$, obtaining the final outcome $\vrootOf{\vtree^{m+1}}$. Notice that when $m=0$, the final outcome is the result $\vrootOf{\vtree^1}$ of $\e_1$. The aggregation is performed with respect to the current device and the empty environment, since the value-trees of the aggregation process cannot be meaningfully related with one another (and thus are not stored in the final outcome of the computation). In other words, the aggregator $\funvalue$ is forced to be a ``pure'' function independent of the current device and environment (even though the expression $\e_2$  as a whole might depend on the environment).
\end{itemize}

The values aggregated by $\foldK$ exclude the value of $\e_3$ in the current device $\deviceId$. However, an inclusive folding operation $\mathtt{foldhoodPlusSelf}$ can  be encoded as \\

\lstinline[language={hfc}]|def foldhoodPlusSelf(f, v) { foldhood(v, f, v) }|.\\ 

\noindent We also remark that a sequence of nested $\foldK$-operators (not interleaved by $\nbrK$-operators) can lead to an evaluation time which is exponential in the evaluation tree depth.\footnote{In actual implementations, the outcome of $\foldK$ and $\repK$ subexpressions can be ``memoised'' in order to prevent subsequent re-evaluation (since such expressions are independent of the neighbour against which are evaluated). This addresses the performance issues of nested $\foldK$-operators.}

Failure of evaluation against a neighbour is formalised by means of the auxiliary judgement $\bsopsemFAIL{\deviceId,\deviceId'}{\Trees}{\senstate}{\e}$ defined by the rules in  Figure~\ref{fig:deviceSemanticsFAIL}. Rules \ruleNameSize{[E-NBR-FAIL]} and \ruleNameSize{[E-R-APP-FAIL]} model the failure sources, while the other rules model failure propagation.

\begin{example}[Neighbourhood interaction]\label{exa:NI-formal}
To illustrate rules \ruleNameSize{[E-FOLD]}, \ruleNameSize{[E-NBR]} and \ruleNameSize{[E-NBR-LOC]},  we consider program
\[
\foldK(2,\texttt{+}, \texttt{min}(\nbrK\{\texttt{temperature}()\}, \texttt{temperature}()) )
\]
evaluated in device $\deviceId_0$ (in which $\texttt{temperature}() = 10$) with neighbours $\deviceId_1$ (in which $\texttt{temperature}() = 15$)
and $\deviceId_2$ (in which $\texttt{temperature}() = 5$). By Rule \ruleNameSize{[E-FOLD]}, the three subexpressions of the $\foldK$-expression
are evaluated with respect to $\deviceId_0$ into the value-trees $\vtree^1, \vtree_f, \vtree_0$ which will constitute the branches of the final tree. The first two of them are
$\vtree^1 = \mkvt{2}{}$ and $\vtree_f = \mkvt{\texttt{+}}{}$, each obtained by Rule \ruleNameSize{[E-VAL]}.
Then, the third subexpression is evaluated against $\deviceId_0$, $\deviceId_1$  and $\deviceId_2$, obtaining:
\begin{align*}
\vtree_0 &= \mkvt{10}{\texttt{min}, \mkvt{10}{\mkvt{10}{\mathtt{temperature},10}}, \mkvt{10}{\mathtt{temperature},10}, 10}, \\
\vtree_1 &= \mkvt{10}{\texttt{min}, \mkvt{15}{\mkvt{15}{\mathtt{temperature},15}}, \mkvt{10}{\mathtt{temperature},10}, 10},  \\
\vtree_2 &= \mkvt{5}{\texttt{min}, \mkvt{5}{\mkvt{5}{\mathtt{temperature},5}}, \mkvt{10}{\mathtt{temperature},10}, 5}
\end{align*}
the first one ($\vtree_0$) obtained through three applications of Rule \ruleNameSize{[E-B-APP]} and one of Rule \ruleNameSize{[E-NBR-LOC]}, and the other two ($\vtree_1 $ and $\vtree_2$) obtained through three applications of Rule \ruleNameSize{[E-B-APP]} and one of Rule \ruleNameSize{[E-NBR]}.
The roots of value-trees $\vtree_1$ and $\vtree_2$ are then combined through operator $\texttt{+}$, together with the initial value $2$, for a total result of $2 + 10+ 5 = 17$  which is the root of the final value-tree $\mkvt{17}{2, \texttt{+}, \vtree_0}$.
\revised{
Consider the following graphical representation of the value-trees involved.
}
\noindent\begin{minipage}{\textwidth}
\centering
\begin{tikzpicture}[-,node distance=1.5cm and 1.5cm, every node/.style={rectangle,draw},
    level 1/.style={sibling distance=1cm},
    level 2/.style={sibling distance=1cm},
    level 3/.style={sibling distance=1cm}, ]
\scriptsize
\node (root) { \texttt{17} }
  child { node (init) { \texttt{2} } }
  child { node (fun) { \texttt{+} } }
  child {
    node (d0) { \texttt{10} }
    child { node { \texttt{min} } }
    child {
    	node { \texttt{10} }
    	child {
    	  node { \texttt{10} }
    	  child { node { \texttt{temperature} } }
    	  child { node [yshift=0.5cm] { \texttt{10} } }
    	}
    }
    child { node { \texttt{10} }
  	  child { node [xshift=0.8cm] { \texttt{temperature} } }
    	child { node [xshift=0.6cm, yshift=0.5cm] { \texttt{10} } }
    }
    child { node { \texttt{10} } }
  };

\node[draw=none] [above right=0cm of d0] {$\vtree_0$};
\node[draw=none] [above left=0cm of init] {$\vtree^1$};
\node[draw=none] [above right=-0.1cm of fun] {$\vtree_f$};

\node[draw=none] [right=0cm of root,align=left,minimum size=2cm] {$\vtree^3 = \vtree^2+\vtree_2 = 12+5=17$\\
where $\vtree^2 = \vtree^1+\vtree_1=2+10=12$};


\node (n2) [right=5cm of root] { \texttt{10} }
  child { node { \texttt{min} } }
  child {
  	node { \texttt{15} }
  	child {
  	  node { \texttt{15} }
  	  child { node { \texttt{temperature} } }
  	  child { node [yshift=0.5cm] { \texttt{15} } }
  	}
  }
  child { node { \texttt{10} }
	  child { node [xshift=0.8cm] { \texttt{temperature} } }
  	  child { node [xshift=0.6cm, yshift=0.5cm] { \texttt{10} } }
  }
  child { node { \texttt{10} } };

\node (n3) [right=4cm of n2]  { \texttt{5} }
  child { node { \texttt{min} } }
  child {
  	node { \texttt{5} }
  	child {
  	  node { \texttt{5} }
  	  child { node { \texttt{temperature} } }
  	  child { node [yshift=0.5cm] { \texttt{5} } }
  	}
  }
  child { node { \texttt{10} }
	  child { node [xshift=0.8cm] { \texttt{temperature} } }
  	  child { node [xshift=0.6cm,yshift=0.5cm] { \texttt{10} } }
  }
  child { node { \texttt{5} } };

\node[draw=none] [above right=0cm of n2] {$\vtree_1$};
\node[draw=none] [above right=0cm of n3] {$\vtree_2$};

\end{tikzpicture}
\end{minipage}

\end{example}

Rules \ruleNameSize{[E-VAL]}, \ruleNameSize{[E-REP]}, \ruleNameSize{[E-FOLD]} are independent of the neighbour $\deviceId'$
against which the expression is computed (since $\deviceId'$ does not occur in the premises of those rules). Rules \ruleNameSize{[E-B-APP]}
and \ruleNameSize{[E-D-APP]} simply pass $\deviceId'$ through, allowing subexpressions to make use of it
(including evaluation of built-in relational sensors $\snvalue$). The neighbour device $\deviceId'$ is then non-trivially
exploited only in rules \ruleNameSize{[E-NBR]}, \ruleNameSize{[E-NBR-LOC]}, \ruleNameSize{[E-NBR-FAIL]} and \ruleNameSize{[E-R-APP-FAIL]}.

We say that a neighbour is \emph{considered} by the evaluation of a  $\foldK$-expression to mean that it contributes to the result of the expression. Because of the interplay between neighbourhood interaction and  branching (i.e., function call)  only a subset of the neighbourhood of a device might be considered by a $\foldK$-expression.

\begin{example}[Combining time evolution with neighbourhood interaction: the gradient]
\label{ex:time-nbr-gradient}
Consider the \emph{gradient} function from Example~\ref{exa:typing}.
\begin{lstlisting}[language={hfc}]
def gradient(source, metric) { // : (bool, ()->num) -> num
  rep(PositiveInfinity){ (distance) => {
    mux(source, 0.0,
      foldhood(PositiveInfinity, min, nbr{distance} + metric())
    )
  }}
}
\end{lstlisting}
The \texttt{gradient} function computes the field of minimum distances (according to \texttt{metric}) from devices where \texttt{source} is $\truevalue$. %
It uses $\repK$ to keep track of the local gradient value, which is computed by looking at the corresponding value in the neighbourhood.
If \texttt{source} is locally $\truevalue$, then the value is merely \texttt{0.0}, since it means that the device is a source;
otherwise, the gradient value is obtained by folding over neighbours (via \texttt{foldhood}) to collect the minimum value of \texttt{nbr\{distance\}+metric()}.
The repeated application of such a function, together with actual communication (formally covered in \Cref{sec-calculus-network-semantics}) consisting of the \texttt{nbr} evaluations of neighbours, makes the output field eventually converge to the correct value (minimum distances from sources).
The gradient is a fundamental building block
 for collective adaptive behaviour.
It is amenable to various implementations~\cite{ACDV-SASO2017}
 and plays a crucial role in higher-level patterns~\cite{casadei19scr}, as also shown in the case study of \Cref{sec-case}.
\end{example}

\begin{example}[Neighbourhood interaction and branching]\label{exa:NIandB-formal}
In order to illustrate the alignment process, guiding neighbour interaction through branching statements, consider the expression for a \emph{gradient avoiding obstacles}, introduced in Example \ref{exa:typing}.
\begin{lstlisting}[language={hfc}]
if (isObstacle) { PositiveInfinity } { gradient(isSource) }  // : num
\end{lstlisting}
Expanding the syntactic sugar, the $\ifK$ statement corresponds to the execution of a different anonymous function depending on the value of \texttt{isObstacle}:
\begin{lstlisting}[language={hfc}]
mux( isObstacle, () => {PositiveInfinity}, () => {gradient(isSource)} )()
\end{lstlisting}
Assume that device $\deviceId_0$ evaluates this program with respect to $\Trees = \{\envmap{\deviceId_0}{\vtree_0}, \envmap{\deviceId_1}{\vtree_1}, \envmap{\deviceId_2}{\vtree_2}\}$, where \texttt{isObstacle} is true in $\deviceId_2$ and false on the other devices.
Thus, the execution of the $\texttt{mux}$ statement produces $\funvalue_\bot = ()\, \toSymKtag{}{\texttt{gradient}(\texttt{isSource})}$ on $\deviceId_0$ and $\deviceId_1$, while it produces
$\funvalue_\top= ()\, \toSymKtag{}{\texttt{PositiveInfinity}}$ on $\deviceId_2$.

The evaluation of the main expression is performed through rule \ruleNameSize{[E-D-APP]}. First, the function to be applied is computed as the result of the \texttt{mux} expression. Then, the body $\texttt{gradient}(\texttt{isSource})$ is computed with respect to the environment $\piBof{\funvalue_\bot}{\Trees} = \{\envmap{\deviceId_0}{\piIof{2}{\vtree_0}}, \envmap{\deviceId_1}{\piIof{2}{\vtree_1}}\}$: the value-tree of device $\deviceId_2$ is removed since it corresponded to the evaluation of $\funvalue_\top$. The evaluation of $\texttt{gradient}(\texttt{isSource})$ will then require the evaluation of the \texttt{foldhood} expression, in which only devices $\deviceId_0$ and $\deviceId_1$ will be considered (since $\deviceId_2$ has already been discarded).
\end{example}

\subsection{Operational Semantics: Network Semantics}\label{sec-calculus-network-semantics}

\begin{figure}[!t]
\framebox[1\textwidth]{
$\begin{array}{l}
\textbf{System configurations and action labels:}\\[5pt]
\begin{array}{lcl@{\hspace{44mm}}r}
\Field & \BNFcce &  \envmap{\overline\deviceId}{\overline\Trees}    &   {\footnotesize \mbox{value-tree field}} \\
\Activation & \BNFcce &  \envmap{\overline\deviceId}{\overline a} \text{ with } a \in \{\actOFF,\actON\}   &   \footnotesize \mbox{activation predicate} \\
\Stat & \BNFcce &  \EnviS{\Field}{\Activation}    &   \footnotesize \mbox{status} \\
\Topo & \BNFcce &  \ap{\overline\deviceId, \overline\deviceId'}    &   {\footnotesize \mbox{topology}} \\
\Sens & \BNFcce &  \envmap{\overline\deviceId}{\overline\senstate}    &   {\footnotesize \mbox{sensor state}} \\
\Envi & \BNFcce &  \EnviS{\Topo}{\Sens}    &   {\footnotesize \mbox{environment}} \\
\Cfg & \BNFcce &  \SystS{\Envi}{\Stat}    &   {\footnotesize \mbox{network configuration}} \\
\act & \BNFcce &  \deviceId+ \;\BNFmid\; \deviceId- \;\BNFmid\; \envact    &   {\footnotesize \mbox{action label}} \\
\end{array}\\
\hline\\[-8pt]
\textbf{Environment well-formedness:}\\[5pt]
\begin{array}{l}
	\wfn{\EnviS{\Topo}{\Sens}} \textrm{~~holds iff~~} \bp{\ap{\deviceId,\deviceId} \mid \deviceId \in \deviceIdSet} \subseteq \, \Topo \, \subseteq \deviceIdSet \times \deviceIdSet \text{~~where~~} \deviceIdSet = \domof{\Sens}
\end{array}\\
\hline\\[-8pt]
\textbf{Transition rules for network evolution:} \hspace{5.5cm}
  \boxed{\nettran{\Cfg}{\act}{\Cfg}}
  \\[5pt]
\begin{array}{c}
\netopsemRule{N-COMP}{
	\;\; \Activation(\deviceId)\!=\!\actOFF
	\quad \Trees' = \filter_\deviceId(\Field(\deviceId))
	\quad \bsopsem{\deviceId,\deviceId}{\Trees'}{\Sens(\deviceId)}{\emain}{\vtree}
	\quad \Trees\!=\!\mapupdate{\Trees'}{\envmap{\deviceId}{\vtree}}
}{
	\nettran{\SystS{\EnviS{\Topo}{\Sens}}{\EnviS{\Field}{\Activation}}}{\deviceId+}{\SystS{\EnviS{\Topo}{\Sens}}{\EnviS{\mapupdate{\Field}{\envmap{\deviceId}{\Trees}}}{\mapupdate{\Activation}{\envmap{\deviceId}{\actON}}}}}
}
\\[15pt]
\netopsemRule{N-SEND}{
	\quad \Activation(\deviceId)\!=\!\actON
	\quad \overline\deviceId  = \bp{\deviceId' ~\mid~ \deviceId \Topo \deviceId'}
	\quad \vtree = \Field(\deviceId)(\deviceId)
	\quad \Trees = \envmap{\deviceId}{\vtree}
}{
	\nettran{\SystS{\EnviS{\Topo}{\Sens}}{\EnviS{\Field}{\Activation}}}{\deviceId-}{\SystS{\EnviS{\Topo}{\Sens}}{\EnviS{\globalupdate{\Field}{\envmap{\overline\deviceId}{\Trees}}}{\mapupdate{\Activation}{\envmap{\deviceId}{\actOFF}}}}}
}
\\[15pt]
\netopsemRule{N-ENV}{
	\quad \wfn{\Envi'}
	\quad
	\Envi'=\EnviS{\Topo}{\envmap{\overline\deviceId}{\overline\senstate}}
	 \quad
	\Field_0=\envmap{\overline\deviceId}{\emptyset}
	\quad
	\Activation_0=\envmap{\overline\deviceId}{\actOFF}
}{
	\nettran{\SystS{\Envi}{\Field,\Activation}}{\envact}{\SystS{\Envi'}{\mapupdate{\Field_0}{\Field}, \mapupdate{\Activation_0}{\Activation}}}
}\\
\end{array}\\
\end{array}$}
\caption{Small-step operational semantics for network evolution.} \label{fig:networkSemantics}
\end{figure}

We now provide an operational semantics for the evolution of whole networks, namely, for modelling the distributed evolution of computational fields over time. The semantics is given as a nondeterministic, small-step transition system on network configurations $\Cfg$. This semantics has already been given for \HFC{} in \cite{abdpv:lmcs:share}, with the only difference of referring to the \HFC{} device semantics instead of the \FSCAFI{} device semantics. Figure \ref{fig:networkSemantics} (top) defines key syntactic elements to this end:
\begin{itemize}
	\item
	$\Field$ is a computational field (called \emph{value-tree field}) that models the overall state of the computation as a map from device identifiers to the value-tree environments that are locally stored in the corresponding devices.
	\item
	$\Activation$ is an \emph{activation predicate} specifying whether each device is currently activated (i.e., is performing a computation round).
	\item
	$\Stat$ (a pair of value-tree field and activation predicate) models the overall computation \emph{status}.
	\item
	$\Topo$ models network \emph{topology} as a directed neighbouring graph, i.e.~a reflexive neighbouring relation $\Topo \subseteq \deviceIdSet \times \deviceIdSet$ so that $\deviceId \Topo \deviceId$ for each $\deviceId \in \deviceIdSet$.
	\item
	$\Sens$ models (distributed) \emph{sensor state}, as a map from device identifiers to (local) sensors representations (i.e., sensor name/value maps denoted as $\senstate$).
	\item
	$\Envi$ (a pair of topology and sensor state) models the network \emph{environment}.
	\item
	$\Cfg$ (a pair of status and environment) models a whole \emph{network configuration}.
\end{itemize}

We use the following notation for maps. Let $\envmap{\overline x}{y}$ denote a map sending each element in the sequence $\overline x$ to the same element $y$. Let $\mapupdate{m_0}{m_1}$ denote the map with domain $\domof{m_0} \cup \domof{m_1}$ coinciding with $m_1$ in the domain of $m_1$ and with $m_0$ otherwise. Let $\globalupdate{m_0}{m_1}$ (where $m_i$ are maps to maps) denote the map with the \emph{same domain} as $m_0$ made of $\envmap{x}{\mapupdate{m_0(x)}{m_1(x)}}$ for all $x$ in the domain of $m_1$, $\envmap{x}{m_0(x)}$ otherwise.
The notation $\filter_\deviceId(\cdot)$ used in rule \ruleNameSize{[N-COMP]}, Figure \ref{fig:networkSemantics} (bottom), models a filtering operation that clears out old stored value-trees from $\Field(\deviceId)$, implicitly based on space/time tags.\footnote{For example, the filter may remove value-trees that were stored before $t - \Delta t$, where $t$ is the time of the current firing and $\Delta t$ is a decay parameter of the filter.} Notice that this mechanism allows messages to persist across rounds.

We define network operational semantics in terms of small-steps transitions  $\nettran{\Cfg}{\act}{\Cfg'}$ of three kinds: firing starts on a given device (for which $\act$ is $\deviceId+$ where $\deviceId$ is the corresponding device identifier), firing ends and messages are sent on a given device (for which $\act$ is $\deviceId-$), and environment changes, where $\act$ is the special label $\envact$. This is formalised in Figure \ref{fig:networkSemantics} (bottom).

Rule \ruleNameSize{[N-COMP]} (available for sleeping devices, i.e., with $\Activation(\deviceId) = \actOFF$, and setting them to executing, i.e., $\Activation(\deviceId) = \actON$) models a computation round at device $\deviceId$: it takes the local value-tree environment filtered out of old values $\Trees' = \filter_\deviceId(\Field(\deviceId))$; then by the single device semantics it obtains the device's value-tree $\vtree$, which is used to update the system configuration of  $\deviceId$ to $\Trees = \mapupdate{\Trees'}{\envmap{\deviceId}{\vtree}}$.
Notice that expression $\emain$ is always evaluated against the device $\deviceId$ itself (that is, against no neighbour), and that local sensors $\Sens(\deviceId)$ are used by the auxiliary function $\builtinop{\bname}{\deviceId,\deviceId'}{\Trees,\Sens(\deviceId)}$ that gives the semantics to the built-in functions.
Furthermore, although this rule updates a device's system configuration instantaneously, it models computations taking an arbitrarily long time, since the update is not visible until the following rule \ruleNameSize{[N-SEND]}. Notice also that all values used to compute $\vtree$ are locally available (at the beginning of the computation), thus allowing for a fully-distributed implementation without global knowledge.

Rule \ruleNameSize{[N-SEND]} (available for running devices with $\Activation(\deviceId) = \actON$, and setting them to non-running) models the message sending happening at the end of a computation round at a device $\deviceId$. It takes the local value-tree $\vtree = \Field(\deviceId)(\deviceId)$ computed by last rule \ruleNameSize{[N-COMP]}, and uses it to update neighbours' $\overline\deviceId$ values of $\Field(\overline\deviceId)$. Notice that the usage of $\Activation$ ensures that occurrences of rules \ruleNameSize{[N-COMP]} and \ruleNameSize{[N-SEND]} for a device are alternated.

Rule \ruleNameSize{[N-ENV]} takes into account the change of the environment to a new \emph{well-formed} environment $\Envi'$---environment well-formedness is specified by the predicate $\wfn{\Envi}$ in Figure~\ref{fig:networkSemantics} (middle)---thus modelling node mobility as well as changes in environmental parameters. Let $\overline\deviceId$ be the domain of $\Envi'$. We first construct a value-tree field $\Field_0$ and an activation predicate $\Activation_0$ associating to all the devices of $\Envi'$ the empty context $\emptyset$ and the $\actOFF$ activation. Then, we adapt the existing value-tree field $\Field$ and activation predicate $\Activation$ to the new set of devices: $\mapupdate{\Field_0}{\Field}$, $\mapupdate{\Activation_0}{\Activation}$ automatically handles removal of devices, mapping of new devices to the empty context and $\actOFF$ activation, and retention of existing contexts and activation in the other devices. We remark that this rule is also used to model communication failure as topology changes.

\revisedMinor{
\begin{rem}[Locality of interactions]
	We remark that although the network semantics is given from a \emph{global} perspective, it only allows for interactions that are local in nature. Rule \ruleNameSize{[N-COMP]} can be understood as happening on device $\deviceId$, with respect to locally-available knowledge and without interactions with other devices. Rule \ruleNameSize{[N-SEND]} represents the broadcast of a value locally available on $\deviceId$, which is then received by neighbour devices and stored in their local storage. Rule \ruleNameSize{[N-ENV]} does not represent an interaction at all, but only a change in the set of possible interactions.
\end{rem}
}

\begin{example}[Network evolution]\label{exa:NetEvo}
Consider the program in Example~\ref{exa:NI-formal}:
\[
\foldK(2,\texttt{+}, \texttt{min}(\nbrK\{\texttt{temperature}()\}, \texttt{temperature}()) )
\]
and let $\vtree^n = \mkvt{n}{\texttt{min}, \mkvt{n}{\mkvt{n}{\mathtt{temperature},n}}, \mkvt{n}{\mathtt{temperature},n}, n}$ be the result of evaluation of $\texttt{min}(\nbrK\{\texttt{temperature}()\}, \texttt{temperature}())$ in a device where $\texttt{temperature}() = n$ (with respect to the device itself as neighbour).

We start from a configuration $\Cfg_0 = \SystS{\EnviS{\Topo}{\Sens}}{\EnviS{\Field_0}{\Activation_0}}$ with three devices $\overline\deviceId$, so that $\Topo = \bp{(\deviceId_i, \deviceId_j) ~ \mid ~ i,j \le 3}$ (all devices are connected), $\Field_0 = \envmap{\overline\deviceId}{\emptyset}$ (devices do not hold any information), $\Activation_0 = \envmap{\overline\deviceId}{\actOFF}$ (devices are not computing) and
\[
\Sens = \envmap{\deviceId_1}{\{t=10\}}, \envmap{\deviceId_2}{\{t=15\}}, \envmap{\deviceId_3}{\{t=5\}}
\]
(temperatures are as in Example~\ref{exa:NI-formal}).

After transitions $\nettran{\Cfg_0}{\deviceId_2^+}{\Cfg} \nettran{}{\deviceId_2^-}{\Cfg_1}$, the computational field $\Field_0$ is updated by sending the result $\vtree_0 = \mkvt{2}{2, \texttt{+}, \vtree^{15}}$ of the computation of $\deviceId_2$ (with respect to its empty environment) to every device, obtaining $\Field_1 = \envmap{\overline\deviceId}{\{\envmap{\deviceId_2}{\vtree_0}\}}$.
Then, other transitions take place: $\nettran{\Cfg_1}{\deviceId_3^+}{\Cfg} \nettran{}{\deviceId_3^-}{\Cfg_2}$, where $\Field_1$ is further updated with the result $\vtree_1 = \mkvt{7}{2, \texttt{+}, \vtree^{5}}$ of the computation of $\deviceId_3$ (with respect to the information received from $\deviceId_2$), obtaining $\Field_2 = \envmap{\overline\deviceId}{\{\envmap{\deviceId_2}{\vtree_0}, \envmap{\deviceId_3}{\vtree_1}\}}$.
Finally, transitions $\nettran{\Cfg_2}{\deviceId_1^+}{\Cfg} \nettran{}{\deviceId_1^-}{\Cfg_3}$ happen as described in Example~\ref{exa:NI-formal}, producing $\Field_3 = \envmap{\overline\deviceId}{\{\envmap{\deviceId_1}{\vtree_2}, \envmap{\deviceId_2}{\vtree_0}, \envmap{\deviceId_3}{\vtree_1}\}}$ where $\vtree_2 = \mkvt{17}{2, \texttt{+}, \vtree^{10}}$.

Lastly, a transition $\nettran{\Cfg_3}{\envact}{\Cfg_4}$ may happen, lowering temperatures, deleting device $\deviceId_2$, inserting device $\deviceId_4$, and disconnecting device $\deviceId_1$ from device $\deviceId_3$. The result is configuration $\Cfg_4 = \SystS{\EnviS{\Topo'}{\Sens'}}{\EnviS{\Field_4}{\Activation_4}}$ where:
\begin{align*}
\Topo' = ~& \ap{\deviceId_1, \deviceId_4}, \ap{\deviceId_3, \deviceId_4}, \ap{\deviceId_4, \deviceId_1}, \ap{\deviceId_4, \deviceId_3} \\
\Sens' = ~&\envmap{\deviceId_1}{\{t=9\}}, \envmap{\deviceId_1}{\{t=4\}}, \envmap{\deviceId_4}{\{t=1\}} \\
\Field_4 = ~&\envmap{\deviceId_1}{\{\envmap{\deviceId_1}{\vtree_2}, \envmap{\deviceId_2}{\vtree_0}, \envmap{\deviceId_3}{\vtree_1}\}}, \envmap{\deviceId_3}{\{\envmap{\deviceId_1}{\vtree_2}, \envmap{\deviceId_2}{\vtree_0}, \envmap{\deviceId_3}{\vtree_1}\}}, \envmap{\deviceId_4}{\emptyset}, \\
\Activation_4 = ~&\envmap{\deviceId_1}{\actOFF}, \envmap{\deviceId_3}{\actOFF}, \envmap{\deviceId_4}{\actOFF}.
\end{align*}
Notice that devices $\deviceId_1$, $\deviceId_3$ are not aware yet of the disappearance of $\deviceId_2$, nor of their disconnection. When one of them will fire, the filter $\filter(\cdot)$ may be able to remove the obsolete values from the corresponding value-tree environments.
\end{example}

\subsection{Type Preservation in \FSCAFI{}} \label{ssec:soundness}

In this section we show that the evaluation rules for \FSCAFI{} are deterministic and preserve types, provided that the value-tree environment used for the evaluation is coherent with the expression being evaluated according to the following definition.

\begin{defi}[Well Formed Value Tree] \label{def:Tcoherence}
	Given a closed expression $\e$, a local-type-scheme environment $\LTStypEnv$, a type environment $\TtypEnv=\overline{\xname}:\overline{\type}$, and a type $\type$ such that $\expTypJud{\LTStypEnv}{\TtypEnv}{\e}{\type}$ holds, the set $\coherent{\e}{\type}{\LTStypEnv;\TtypEnv}$ of the \emph{well-formed value-trees} for $\e$ is inductively defined as follows. $\vtree \in \coherent{\e}{\type}{\LTStypEnv;\TtypEnv}$ if and only if $\anyvalue = \vrootOf{\vtree}$ has type $T$ (i.e.~$\expTypJud{\LTStypEnv}{\emptyset}{\anyvalue}{\type}$) and
	\begin{itemize}
		\item
		if $\e$ is a value, $\vtree$ is of the form $\mkvtree{}{\anyvalue}{}$;
		\item
		if $\e = \nbrK\{\e_1\}$, $\vtree$ is of the form $\mkvtree{}{\anyvalue}{\vtree_1}$ where $\vtree_1 \in \coherent{\e_1}{\type}{\LTStypEnv;\TtypEnv}$;
		\item
		if $\e = \repK(\e_1)\{\e_2\}$, $\vtree$ is of the form $\mkvtree{}{\anyvalue}{\vtree_1,\vtree_2}$ where $\vtree_1 \in \coherent{\e_1}{\type}{\LTStypEnv;\TtypEnv}$ and $\vtree_2 \in \coherent{\e_2(\xname)}{\type}{\LTStypEnv;\TtypEnv, \xname:\type}$;
		\item
		If $\e = \foldK(\e_1, \e_2, \e_3)$, $\vtree$ is of the form $\mkvtree{}{\anyvalue}{\vtree_1,\vtree_2,\vtree_3}$ where $\overline\vtree \in \coherent{\overline\e}{\type, (\type,\type) \to \type, \type}{\LTStypEnv;\TtypEnv}$;
		\item
		if $\e = \e'(\overline\e)$ and $\expTypJud{\LTStypEnv}{\TtypEnv}{\e'}{\type'}$, $\expTypJud{\LTStypEnv}{\TtypEnv}{\overline\e}{\overline\type}$, then $\vtree$ is of the form $\mkvtree{}{\anyvalue}{\vtree',\overline\vtree,\vtree''}$ where $\vtree' \in \coherent{\e'}{\type'}{\LTStypEnv;\TtypEnv}$, $\overline\vtree \in \coherent{\overline\e}{\overline\type}{\LTStypEnv;\TtypEnv}$, and either:
		\begin{itemize}
			\item $\funvalue = \vrootOf{\vtree'}$ is a built-in function and $\vtree'' = \mkvtree{}{\anyvalue}{}$,
			\item $\funvalue$ is not a built-in function and $\vtree'' \in \coherent{\body{\funvalue}}{\type}{\LTStypEnv;\TtypEnv, \args{\funvalue} : \overline\type}$.
		\end{itemize}
	\end{itemize}
	Similarly, the set of \emph{well-formed value-tree environments} $\coherentEnv{\e}{\type}{\LTStypEnv;\TtypEnv}$ is the set of $\Trees = \envmap{\overline\deviceId}{\overline\vtree}$ such that $\overline\vtree \in \coherent{\e}{\type}{\LTStypEnv;\TtypEnv}$.
\end{defi}

In other words, the above definition demands value-trees to be plausible outcomes of the evaluation of $\e$.

\begin{lem}[Computation Determinism] \label{lem:completeness}
	Let $\e$ be a well-typed closed expression and $\Trees \in \coherentEnv{\e}{\type}{\LTStypEnv;\TtypEnv}$. Then for all device identifiers $\deviceId$, $\deviceId'$ and sensor state $\senstate$:
	\begin{enumerate}
		\item
		$\bsopsemFAIL{\deviceId,\deviceId}{\Trees}{\senstate}{\e}$ cannot hold.
		\item
		There is at most one derivation of the kind $\bsopsem{\deviceId,\deviceId'}{\Trees}{\senstate}{\e}{\vtree}$ or $\bsopsemFAIL{\deviceId,\deviceId'}{\Trees}{\senstate}{\e}$.
	\end{enumerate}
\end{lem}
\begin{proof}
	\revised{The proof follows from standard induction on rules (see Appendix \ref{apx:preservation}).}
\end{proof}

By this lemma, evaluation does not result on \FAIL{} when it is performed relative to the current device (as it is the case for main expressions), and rules are deterministic. Furthermore, the evaluation rules respect the types given in \Cref{fig:SurfaceTyping}, provided that the \emph{built-in interpretations respect the given types}. Formally, given $\bname$ such that $\expTypJud{\OStypEnv}{\emptyset}{\bname}{\overline\type \to \type}$ and any $\expTypJud{\OStypEnv}{\emptyset}{\overline\anyvalue}{\overline\type}$, $\Trees = \envmap{\overline\deviceId}{\mkvtree{}{\overline\anyvalue'}{}}$ with $\expTypJud{\OStypEnv}{\emptyset}{\overline\anyvalue'}{\type}$, $\deviceId' \in \{\deviceId,\overline\deviceId\}$, then we require $\builtinop{\bname}{\Trees,\senstate}{\deviceId,\deviceId'}$ to be a value of type $\type$.

\begin{thm}[Type Preservation] \label{thm:soundness}
	Assume that the interpretation of built-in operators respects the given types.
	Let $\TtypEnv=\overline{\xname}:\overline{\type}$ and $\expTypJud{\LTStypEnv}{\emptyset}{\overline{\anyvalue}}{\overline{\type}}$, so that $\lengthOf{\overline{\anyvalue}}=\lengthOf{\overline{\xname}}$.
	If $\expTypJud{\LTStypEnv}{\TtypEnv}{\e}{\type}$, $\Trees\in \coherentEnv{\e}{\type}{\LTStypEnv;\TtypEnv}$ and $\bsopsem{\deviceId,\deviceId'}{\Trees}{\senstate}{\applySubstitution{\e}{\substitution{\overline{\xname}}{\overline{\anyvalue}}}}{\vtree}$, then $\vtree \in \coherent{\e}{\type}{\LTStypEnv;\TtypEnv}$.
\end{thm}

Notice that, since the evaluation of $\e$ produces a value-tree which is coherent with $\e$, the value-tree environment $\Trees$ can be proved to be coherent with the main expression by induction on the network evolution.
Furthermore, observe that the typing rules (in Figure~\ref{fig:SurfaceTyping}) and the evaluation rules (in Figure~\ref{fig:deviceSemantics} and \ref{fig:deviceSemanticsFAIL}) are syntax directed. Then the proof can be carried out by induction on the derivation length for $\bsopsem{\deviceId,\deviceId'}{\Trees}{\senstate}{\applySubstitution{\e}{\substitution{\overline{\xname}}{\overline{\anyvalue}}}}{\vtree}$, while using the following standard lemmas.

\begin{lem}[Substitution]\label{lem-substitution}
	Let $\TtypEnv=\overline{\xname}:\overline{\type}$, $\expTypJud{\OStypEnv}{\emptyset}{\overline{\anyvalue}}{\overline{\type}}$. If $\expTypJud{\LTStypEnv}{\TtypEnv}{\e}{\type}$, then $\expTypJud{\LTStypEnv}{\emptyset}{\applySubstitution{\e}{\substitution{\overline{\xname}}{\overline{\anyvalue}}}}{\type}$.
\end{lem}
\begin{proof}
	Straightforward by induction on application of the typing rules for expressions in Figure~\ref{fig:SurfaceTyping}.
\end{proof}

\begin{lem}[Weakening]\label{lem-weakening}
	Let $\LTStypEnv'\supseteq\LTStypEnv$, $\TtypEnv'\supseteq\TtypEnv$ be such that $\domof{\LTStypEnv'}\cap\domof{\TtypEnv'}=\emptyset$. If $\expTypJud{\LTStypEnv}{\TtypEnv}{\e}{\type}$, then $\expTypJud{\LTStypEnv'}{\TtypEnv'}{\e}{\type}$.
\end{lem}
\begin{proof}
	Straightforward by induction on application of the typing rules for expressions in Figure~\ref{fig:SurfaceTyping}.
\end{proof}

\begin{proof}[Proof of Theorem \ref{thm:soundness}]
	We proceed by induction on the derivation length.

	The fact that $\vrootOf{\vtree}$ has type $\type$ can be verified by matching step-by-step every rule in \Cref{fig:deviceSemantics} with the corresponding rule in \Cref{fig:SurfaceTyping}, while using the inductive hypothesis and two further assumptions: for rule \ruleNameSize{[E-B-APP]}, that built-in functions $\bname$ respect the given types; for rules \ruleNameSize{[E-NBR]} and \ruleNameSize{[E-REP]}, that $\Trees$ is coherent with $\e$.

	Finally, the fact that $\vtree$ has the required sub-trees follows by inductive hypothesis since every rule in \Cref{fig:deviceSemantics} respects the corresponding row of Definition \ref{def:Tcoherence}, together with the fact that $\Trees$ is coherent with $\e$ (for rules \ruleNameSize{[E-NBR]} and \ruleNameSize{[E-REP]} only).
\end{proof}


\section{\FSCAFI{} vs \HFC{}}\label{sec-properties}

In this section, we provide a formal account of  the relationship between \FSCAFI{} and the \HFC{} minimal core calculus for Aggregate Computing~\cite{audrito2019tocl}.

\revised{In \Cref{ssec:HFC-NC-overview}, we give an informal overview of the relationship between the two calculi.}
In \Cref{ssec:HFC} we recollect \HFC{}.
In \Cref{ssec:HFCprime}, we define a fragment  of \HFC{}, which we call \HFCprime{}, aimed at ensuring that
each \HFCprime{} program is an \FSCAFI{} program that behaves in the same way.
In \Cref{ssec:alignedhfc}, we define the fragment  of  \FSCAFI{}  which
corresponds to \HFCprime{}, which we call call \FSCAFIi{}.
Then, in \Cref{ssec:NCi-equiv-HFCprime} we prove the equivalence between \FSCAFIi{} and \HFCprime{}.
Finally, in \Cref{ssec:expressiveness}, we  point out that \FSCAFI\ provides a different ``flavour'' of field computation with respect to \HFC{}, though without losing practical expressiveness.

\revised{
\subsection{Overview} \label{ssec:HFC-NC-overview}

\begin{figure}[t]
\begin{center}
	\includegraphics[scale=0.8]{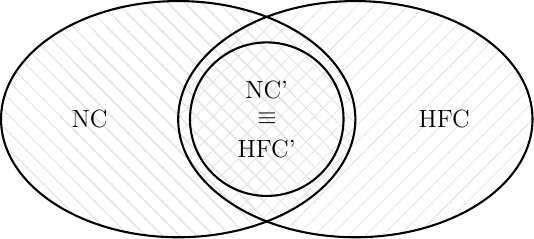}
\end{center}
\caption{Pictorial representation of relationship between \FSCAFI{}, \HFC{}  and their fragments \FSCAFIi{} and \HFCprime{}.} \label{fig:framents}
\end{figure}

\FSCAFI{} and \HFC{} are languages with a similar purpose and syntax, but differently structured semantics. Both languages are Turing-complete (in fact, \emph{space-time universal} \cite{abdv:universality}), thus they can both express all distributed computations, although possibly with structurally different programs.

As the syntax of \FSCAFI{} can be mapped to/from that of \HFC{}, a question naturally arises: what are the distributed computations that are expressed in both languages \emph{by the same program}?
The answer to this question, detailed in the remainder of this section, is summarised in Figure~\ref{fig:framents}. Not all programs are interpreted to the same behaviour in \FSCAFI{} and \HFC{}, however, the (subset of) programs with identical behaviour in \FSCAFI{} and \HFC{} can be identified. Through Lemma \ref{lem:hfcp} in Section \ref{ssec:HFCprime}, we define a restricted language \HFCprime{} and show that it is indeed a subset of \HFC{}. Similarly, in Theorem \ref{thm:correspondence} we define \FSCAFIi{} while ensuring that it is a subset of \FSCAFI{}. Finally, Theorem \ref{thm:equivalence} shows that \HFCprime{} coincides with \FSCAFIi{}, concluding the characterisation of this common fragment of programs with identical behaviour.

We remark that \HFCprime{}/\FSCAFIi{} does \emph{not} coincide with the ``intersection'' between \HFC{} and \FSCAFI{}: there are programs outside \HFCprime{}/\FSCAFIi{} which nonetheless result in the same behaviour in both \HFC{} and \FSCAFI{}. Still, we will show in \Cref{ssec:expressiveness} that most interesting programs do fit inside \HFCprime{}/\FSCAFIi{}, or can be refactored to fit into it through a small set of simple \emph{rewriting strategies}. In fact, we are not aware of any practically relevant program that cannot be refactored into \HFCprime{}/\FSCAFIi{} through the proposed strategies. We also remark that as \HFCprime{}/\FSCAFIi{} is Turing-complete, some equivalent form in \HFCprime{}/\FSCAFIi{} must exist for any given distributed program. Overall, these results show that \FSCAFI\ provides a different ``flavour'' of field computation with respect to \HFC{}, without losses in expressiveness.
}

\subsection{A Quick Recollection of \HFC{}} \label{ssec:HFC}

The syntax of \HFC{} programs (according to its original presentation \cite{audrito2019tocl}) is the same of \FSCAFI{} programs (in Figure \ref{fig:source:syntax}), with a richer syntax of values and two other minor differences: the update function $\e_2$ in a $\repK$ construct $\repK (\e_1) \{\e_2\}$ is required to be an anonymous function $\e_2 = (\xname) \toSymK{\e'}$,\footnote{We remark that this difference is historical in nature: \HFC{} could be straightforwardly extended to allow for general update function expressions.} and the language construct for $\foldK$ is replaced by a built-in with the same meaning, so it does not appear in the syntax.
In the richer syntax of \HFC{} values (given in Figure \ref{fig:source:hfc}), values are divided into \emph{local values}, which are  the values of \FSCAFI{}, and \emph{neighbouring (field) values}, which not allowed to appear in source code and arise at runtime.
Neighbouring values are maps $\envmap{\overline\deviceId}{\overline \lvalue}$ from device identifiers to local values. In \HFC{}, they are produced  by
evaluating  the $\nbrK$ construct and returned by some built-in functions.


\begin{figure}[!t]
\centering
\centerline{\framebox[\textwidth]{$
\begin{array}{lcl@{\hspace{55mm}}r}
        \anyvalue & \BNFcce & \lvalue   \; \BNFmid \;  \eqhl{\fvalue}   &   {\footnotesize\mbox{value}} \\[2pt]
        \fvalue & \BNFcce & \eqhl{\envmap{\overline{\deviceId}}{\overline{\lvalue}}}  &   {\footnotesize\mbox{neighbouring (field) value}} \\[2pt]
        \lvalue & \BNFcce & \dcOf{\dc}{\overline{\lvalue}} \; \BNFmid \; \funvalue &   {\footnotesize\mbox{local value}} \\[2pt]
        \funvalue & \BNFcce & \bname  \; \BNFmid \; \fname \; \BNFmid \; (\overline{\xname}) \; \toSymK{\e}  &   {\footnotesize \mbox{function value}} \\[2pt]
\end{array}
$}
}
\caption{Syntax of \HFC{} values -- extensions with respect to \FSCAFI{} values are highlighted in grey.} \label{fig:source:hfc}
\end{figure}

The Hindley-Milner type system for \HFC{}~\cite{audrito2019tocl} is given in \Cref{fig:restrictedTypes} (excluding the grey rule \ruleNameSize{[T'-FOLD]}), and distinguishes between types for \emph{local} values from those that are not (namely, neighbouring types $\ftype$ for neighbouring values), as well as between types that are allowed to be \emph{returned} by functions from those that are not.
This induces four different type categories: types $\type$, local types $\ltype$, return types $\rtype$, and local return types
$\stype$---they are illustrated at the bottom of \Cref{fig:source:hfc}.
The main restrictions enforced by this type system in order to ensure the \emph{domain alignment} property\footnote{Domain alignment holds iff the domain of neighbouring values $\fvalue$ obtained from expressions $\e$ is equal to the set of all neighbours which computed the same $\e$ in their previous evaluation round.} are:
\begin{itemize}
	\item
	anonymous functions cannot capture variables of neighbouring type;
	\item
	$\repK$ statements are demanded to have a local return type;
	\item
	neighbouring types can only be built from local return types $\stype$ (i.e., $\ftype = \ftypeOf{\stype}$), since neighbouring values need to be aggregated and this is possible only for return types, and avoiding ``neighbouring values of neighbouring values'' which may lead to unintentionally heavy computations;
	\item
	types of the form $(\overline\type) \to \ftype$ (functions returning neighbouring values) are not return types.
	Thus, functions of type $(\overline\type) \to \ftype$ are used almost as in a first-order language. In particular, there is no way to write a non-constant expression $\e$ evaluating to such a function.
\end{itemize}
The \HFC{} operational semantics \cite{audrito2019tocl} is given as a transition system analogous to that in Section \ref{sec-calculus-network-semantics}, but based on a different judgement for the device operational semantics $\bsopsem{\deviceId}{\Trees}{\senstate}{\emain}{\vtree}$.
For sake of completeness, we report the details of the \HFC{} operational device semantics in Appendix~\ref{apx:hfc:opsem}.
In the remainder of this paper, we assume that the built-in operators of \HFC{} always include:
\begin{itemize}
	\item $\mathtt{consthood}(\anyvalue)$, which returns a neighbouring field value constantly equal to its (local) input $\anyvalue$;
	\item $\mathtt{map}(\funvalue, \overline\fvalue)$, which applies a function $\funvalue$ with local inputs and outputs (of any ariety) pointwise to neighbouring field values $\overline\fvalue$;
	\item $\foldK(\lvalue, \funvalue, \fvalue)$, which collapses a field $\fvalue$ and starting value $\lvalue$ via an aggregator $\funvalue$ (exactly as in \FSCAFI{}).
\end{itemize}
Furthermore, we use $\ifK  (\e_1) \{ \e_2 \} \{ \e_3 \}$ as short for $\muxK(\e_1, () ~\toSymKabs{\e_2}, () ~\toSymKabs{\e_3})()$, as in \FSCAFI{}.

\begin{figure}[t!]
\framebox[1\textwidth]{$\begin{array}{l}
\begin{array}{lclrr}
\textbf{Types:} \\[1pt]
		\type	& \BNFcce & \tvar \; \BNFmid \; {\rtype} \; \BNFmid \; \ltype											   &   {\footnotesize \mbox{type}}  &
		\multirow{5}{*}{\raisebox{-1.3cm}{
			\begin{tikzpicture}[xscale=1.7,yscale=1.1]
				\draw[ultra thick] (-0.5,0) circle (1cm);
				\draw[ultra thick] (0.5,0) circle (1cm);
				\node[draw,rectangle] (TT) at (0, 1.25) {$\type$};
				\node[draw,rectangle] (LT) at (-1, 0.5) {$\ltype$};
				\node[draw,rectangle] (RT) at (1, 0.5) {$\rtype$};
				\node[draw,rectangle] (ST) at (0, -0.5) {$\stype$};
				\node (B) at (0.1, 0.5) {$\builtintype$};
				\node (TxS) at (0, 0.05) {$(\overline\type) \rightarrow \stype$};
				\node (F) at (0.8, -0.4) {$\ftype$};
				\node (TxF) at (-0.95, -0.3) {$(\overline\type) \rightarrow \rtype$};
			\end{tikzpicture}
			}}
		\\[2pt]
		\ltype	& \BNFcce & \ltvar \; \BNFmid \; {\stype} \; \BNFmid \; (\overline{\type})\to\rtype	 & {\footnotesize \mbox{local type}} \\[2pt]
		\rtype	& \BNFcce & \rtvar \; \BNFmid \; \stype \; \BNFmid \; \ftype										  &   {\footnotesize \mbox{return type}} \\[2pt]
		\stype	& \BNFcce & \stvar \; \BNFmid \; \builtintype \; \BNFmid \; (\overline{\type})\to\stype	&   {\footnotesize \mbox{local return type}} \\[2pt]
		\ftype	& \BNFcce & \ftypeOf{\stype}																							&   {\footnotesize \mbox{neighbouring type}} \\[2pt]
				\multicolumn{4}{l}{\textbf{Local type schemes:}} \\[1pt]
		\ltypescheme & \BNFcce & \forall\overline{\tvar}\overline{\ltvar}\overline{\rtvar}\overline{\stvar}.\ltype
																																							  &   {\footnotesize \mbox{local type scheme}} \\
\end{array} \\
\textbf{~~Expression typing:} \hfill \boxed{\expTypJud{\LTStypEnv}{\TtypEnv}{\e}{\type}} \\[5pt]
\begin{array}{c}
	\nullsurfaceTyping{T'-VAR}{
		\expTypJud{\LTStypEnv}{\TtypEnv, \xname:\type}{\xname}{\type}
	}
	\qquad
	\surfaceTyping{T'-DAT}{ \quad
		\applySubstitution{\stype'}{\substitution{\overline\stvar}{\overline\stype''}} = (\overline{\stype})\rightarrow\stype
		\qquad
		\expTypJud{\TStypEnv}{\TtypEnv}{\overline{\lvalue}}{\overline{\stype}}
	}{
		\expTypJud{\TStypEnv,\dc: \forall\overline\stvar. \stype'}{\TtypEnv}{\dc(\overline{\lvalue})}{\stype}
	}
	\skiptransition
	\surfaceTyping{T'-A-FUN}{  \quad
		\overline{\yname}=\FV{(\overline{\xname}) \; \toSymKtag{\name}{\e}} \quad
		{\expTypJud{\LTStypEnv}{\TtypEnv}{\overline\yname}{\overline{\ltype}}} \quad
		\expTypJud{\LTStypEnv}{\TtypEnv, \overline{\xname} : \overline\type}{\e}{\rtype}
	}{
		\expTypJud{\LTStypEnv}{\TtypEnv}{(\overline{\xname}) \; \toSymKtag{\name}{\e}}{(\overline\type) \rightarrow \rtype}
	}
	\skiptransition
	\surfaceTyping{T'-N-FUN}{ \quad
		\mbox{$\funvalue$ is a (built-in or declared) function}
	}{
		\expTypJud{\TStypEnv,\funvalue : \forall\overline\tvar\overline\ltvar\overline\rtvar\overline\stvar.\ltype}{\TtypEnv}{\funvalue}{\applySubstitution{\ltype}{\substitution{\overline\tvar}{\overline\type},~ \substitution{\overline{\ltvar}}{\overline{\ltype}},~ \substitution{\overline{\rtvar}}{\overline{\rtype}},~ \substitution{\overline{\stvar}}{\overline{\stype}}}}
	}
	\skiptransition
	\surfaceTyping{T'-APP}{ \qquad
		\expTypJud{\LTStypEnv}{\TtypEnv}{\e}{(\overline\type)\rightarrow\rtype} \qquad
		\expTypJud{\LTStypEnv}{\TtypEnv}{\overline{\e}}{\overline\type}
	}{
		\expTypJud{\LTStypEnv}{\TtypEnv}{\e(\overline{\e})}{\rtype}
	}
	\skiptransition
	\surfaceTyping{T'-REP}{ \quad
		\expTypJud{\LTStypEnv}{\TtypEnv}{\e_1}{\stype} \qquad
		\expTypJud{\LTStypEnv}{\TtypEnv}{(\xname) \toSymK{\e_2}}{\stype \rightarrow \stype}
	}{
		\expTypJud{\LTStypEnv}{\TtypEnv}{\repK(\e_1)\{(\xname) \toSymK{\e_2}\}}{\stype}
	}
	\quad
	\surfaceTyping{T'-NBR}{ \quad
		\expTypJud{\LTStypEnv}{\TtypEnv}{\e}{\stype}
	}{
		\expTypJud{\LTStypEnv}{\TtypEnv}{\nbrK\{\e\}}{\ftypeOf{\stype}}
	}
	\skiptransition
	\eqhl{\surfaceTyping{T'-FOLD}{ \qquad
		\begin{array}{ll}
			\expTypJud{\TStypEnv}{\TtypEnv}{\e_1}{\stype}
			& \quad
			\expTypJud{\TStypEnv}{\TtypEnv}{\e_3}{\ftypeOf{\stype} \text{ or } \stype}
			\\
			\expTypJud{\TStypEnv}{\TtypEnv}{\e_2}{(\stype, \stype) \rightarrow \stype}
			& \quad
	        \expTypJud{\TStypEnv}{\TtypEnv}{\overline\xname}{\overline\ltype}
	        \quad\text{where}\; \overline\xname = \FV{\e_3}
		\end{array}
	}{
		\expTypJud{\TStypEnv}{\TtypEnv}{\foldK(\e_1,\e_2,\e_3)}{\stype}
	}}
\end{array} \\~\\
\textbf{~~Function typing:} \hfill \boxed{\funTypJud{\LTStypEnv}{\FUNCTION}{\ltypescheme}} \\[5pt]
\qquad \qquad
\begin{array}{c}
	\surfaceTyping{T'-FUNCTION}{ \qquad
		\begin{array}{ll}
			\expTypJud{\LTStypEnv, \fname: (\overline\type)\to\rtype}{\overline{\xname}:\overline{\type}}{\e}{\rtype} \qquad
			\overline{\tvar}\overline{\ltvar}\overline{\rtvar}\overline{\stvar}=\FTV{(\overline{\type})\rightarrow\rtype} 
		\end{array}
	}{
		\funTypJud{\LTStypEnv}{\defK \; \fname (\overline{\xname}) \; \eqSymK{\e}}{
			\forall\overline{\tvar}\overline{\ltvar}\overline{\rtvar}\overline{\stvar}.(\overline{\type})\rightarrow\rtype
	}}
\end{array} \\~\\
\textbf{~~Program typing:} \hfill {\boxed{\funTypJud{\LTStypEnv_0}{\PROGRAM}{\type}}} \\[5pt]
\qquad \quad
\begin{array}{c}
	\surfaceTyping{T'-PROGRAM}{ \\
		\FUNCTION_i=(\defK \; \fname_i (\_) \; \_) \qquad
		\funTypJud{\LTStypEnv_{i-1}}{\FUNCTION_i}{\ltypescheme_i} \qquad
		\LTStypEnv_i=\LTStypEnv_{i-1},\, \fname_i:\ltypescheme_i \qquad
		(i \in 1..n) \\
		\expTypJud{\LTStypEnv_n}{\emptyset}{\e}{\type}
	}{
		\funTypJud{{\LTStypEnv_0}}{\FUNCTION_1\cdots\FUNCTION_n  \; \e}{{\type}}
	}
\end{array}
\end{array}$}
\caption{Hindley-Milner typing for \HFCprime{} and \FSCAFIi{} expressions, function declarations, and programs -- differences with \HFC{} typing are highlighted in grey.} \label{fig:restrictedTypes}
\end{figure}

\subsection{The \HFCprime{} Fragment of \HFC{}} \label{ssec:HFCprime}

\HFCprime{} is obtained by adding the following two custom restrictions, on how neighbouring field values can be processed, to the Hindley-Milner type system for \HFC{}~\cite{audrito2019tocl}:
\begin{description}
	\item[R1]
	Built-in functions need to have local arguments, except for the built-ins $\mathtt{map}$ and $\mathtt{foldhood}$.
	\item[R2]
	Expressions of neighbouring type can only be aggregated to local values with a $\foldK$ operator if they do not capture variables of neighbouring types; so that, e.g., aggregating arguments of neighbouring type is never allowed.
\end{description}

\begin{example}[About restriction R1]
	In order to show the rationale behind Restriction R1, consider a built-in function \texttt{sorthood} rearranging values $\fvalue(\deviceId)$ relative to neighbours in increasing order of neighbour identifier $\deviceId$, thus effectively mixing up values relative to different neighbours. Formally, applying this function to a neighbouring value $\fvalue = \envmap{\overline\deviceId}{\overline\lvalue}$ (assuming $\deviceId_1 \leq \ldots \leq \deviceId_n$), we obtain the neighbouring value $\fvalue' = \envmap{\deviceId_1}{\lvalue_{\pi_1}}, \ldots, \envmap{\deviceId_n}{\lvalue_{\pi_n}}$ where the permutation $\pi$ is such that $\lvalue_{\pi_1} \leq \ldots \lvalue_{\pi_n}$. This function is conceivable (although artificial) in \HFC{}, but it is not implementable in \FSCAFI{}, hence it is disallowed in \HFCprime{}. We remark however that all practically used built-in functions in \HFC{} are definable with respect to those allowed by restriction R1.
\end{example}

\begin{example}[About restriction R2]
	In order to show the rationale behind Restriction R2, consider the following  \HFC{} program
	\begin{lstlisting}[language={hfc}]
def hfc_avghood(x) { // : field(num) -> num
  foldhood(0, +, x) / foldhood(0, +, 1)
}
hfc_avghood(nbr{sns-temp()}) // : num
	\end{lstlisting}
	which on each device calculates the average temperature of neighbours.
	We may suppose that the same code, interpreted as an \FSCAFI{} program, would  calculate the same quantity.
 Instead, it is equivalent to the simpler \FSCAFI{} program
$\texttt{sns-temp}()$,
 which yields the temperature of the device where it is evaluated.
If we evaluate the expression \texttt{nbr\{sns-temp()\}} against a neighbour $\deviceId'$,
we obtain  the temperature $t'$ of that neighbour. Unfortunately, in the program the expression occurs outside of the scope of any $\foldK$ construct and so
   it is evaluated
against the device $\deviceId$ where it is evaluated.
When function \texttt{hfc\_avghood} is applied to $t$ (the temperature of $\deviceId$), the neighbour device $\deviceId'$ is ignored by both $\foldK$ constructs,
 which fail to interpret the captured neighbouring value as such. The value of the program on device $\deviceId$
  is then $n\cdot t / n = t$, where $n$ is the number of neighbours  of $\deviceId$ (including $\deviceId$ itself)
 and $t$ is the value of the temperature on device $\deviceId$.
\end{example}

We remark that an \HFC{} program computing the average temperature of neighbours also when interpreted as an \FSCAFI{} program can be conveniently written, by resorting to suitable programming patterns (illustrated in \Cref{sssec:patterns}). In particular, for the example above, it is sufficient to make \texttt{x} a ``by-name'' parameter, thus obtaining the following \HFC{} program:
	\begin{lstlisting}[language={hfc}]
def hfc_nc_avghood(y) { // : (() -> field(num)) -> num
  foldhood(0, +, y()) / foldhood(0, +, 1)
}
hfc_nc_avghood(() => {nbr{sns-temp()}}) // : num
	\end{lstlisting}

%

We remark that all \HFC{} programs considered in previous works \cite{audrito2019tocl,abdv:universality,viroli:selfstabilisation} actually belong to \HFCprime{} (or can straightforwardly  be reformulated in order to do so). The following lemma provides a characterisation of \HFCprime{} in terms of the type system in \Cref{fig:restrictedTypes}.

\begin{lem}[Characterisation of \HFCprime{}] \label{lem:hfcp}
The type system in \Cref{fig:restrictedTypes} is a restriction of the Hindley-Milner type system for \HFC{}~\cite{audrito2019tocl} enforcing restrictions R1 and R2.
 \end{lem}
\begin{proof}
	Restriction R1 is implicitly implemented in the definition of $\LTStypEnv_0$, in which the only built-in with field arguments is $\mathtt{map}$.
	Restriction R2 is implemented in Rule \ruleNameSize{[T'-FOLD]}, by requiring each free variable occurring in the third branch to be of local type -- in
	the type system for \HFC{} \cite{audrito2019tocl} $\foldK$ is a built-in function, so there is no special rule for $\foldK$-expressions.
	All the other rules are the same as for the Hindley-Milner type system for \HFC{}~\cite{audrito2019tocl}.
\end{proof}

\subsection{The \FSCAFIi{} Fragment of \FSCAFI{}} \label{ssec:alignedhfc}

\FSCAFIi{} is the fragment of \FSCAFI{} that can be typed by rules in \Cref{fig:restrictedTypes}. These rules are significantly more restrictive then the rules in \Cref{fig:SurfaceTyping}, in particular:

\begin{itemize}
	\item
	They acknowledge the existence of the four type categories, in particular of $\ftypeOf{\stype}$;
	\item
	They demand that (anonymous and defined) functions have an allowed return type $\rtype$;
	\item
	They require the sub-expressions of $\repK$, $\nbrK$ and $\foldK$ to have local return type $\stype$;
	\item
	Finally, they require captured variables to have local type, in the body of both anonymous functions and folding expressions.
\end{itemize}
For convenience in the comparison with \HFCprime{}, we assume that \FSCAFIi{} comprises the \texttt{consthood} and \texttt{map} built-ins, even though in \FSCAFIi{} those functions could also be defined as follows.
\begin{lstlisting}[language={hfc},mathescape=true]
def consthood(v) { v }
def map(f, v, $\ldots$) { f(v, $\ldots$) }
\end{lstlisting}
The embedding of \FSCAFIi{}  as a fragment of \FSCAFI{} can be formally characterised by means of the following definition and theorem.

\begin{defi}[Erasure]
	The \emph{erasure} of an \FSCAFIi{} type $\type$ is the type $\erasureof{\type}$ obtained from $\type$ by replacing all occurrences of $\ftypeOf{\ltype}$ with $\ltype$ and dropping the distinction between the different kinds of type variables (i.e., considering each of them as a standard type variable $t$).
	Similarly, the erasure of a type scheme $\forall\overline{\tvar}\overline{\ltvar}\overline{\rtvar}\overline{\stvar}.\ltype$ is the type scheme $\forall\overline{\tvar}\overline{\ltvar}\overline{\rtvar}\overline{\stvar}.\erasureof{\ltype}$ (dropping distinction between kinds of variables).
	Finally, the erasure of a type environment $\TtypEnv = \overline\xname : \overline\type$ is $\erasureof{\TtypEnv} = \overline\xname : \erasureof{\overline\type}$; and the erasure of a type-scheme environment is $\erasureof{\TStypEnv} = \overline\xname : \erasureof{\overline\ltypescheme}$.
\end{defi}

\begin{thm}[Typing Correspondence] \label{thm:correspondence}
	Assume that $\expTypJud{\TStypEnv}{\TtypEnv}{\e}{\type}$ in  \FSCAFIi{}. Then $\expTypJud{\erasureof{\TStypEnv}}{\erasureof{\TtypEnv}}{\e}{\erasureof{\type}}$ in \FSCAFI{}.
\end{thm}
\begin{proof}
	We proceed by induction on the syntax of $\e$. If $\e = \nbrK\{\e'\}$, then $\type = \ftypeOf{\stype}$ where $\expTypJud{\TStypEnv}{\TtypEnv}{\e'}{\stype}$ by rule \ruleNameSize{[T'-NBR]}. By inductive hypothesis, $\expTypJud{\erasureof{\TStypEnv}}{\erasureof{\TtypEnv}}{\e'}{\erasureof{\stype}}$ in \FSCAFI{}. It follows by rule \ruleNameSize{[T-NBR]} that $\expTypJud{\erasureof{\TStypEnv}}{\erasureof{\TtypEnv}}{\e}{\erasureof{\stype}}$, and the thesis follows since $\erasureof{\stype} = \erasureof{\ftypeOf{\stype}}$.

	In all other cases, the thesis follows since the \FSCAFI{} rules can be obtained by removing restrictions from their counterparts in \FSCAFIi{}: that types are restricted to peculiar kinds $\ltype$, $\rtype$ or $\stype$ (every rule except for \ruleNameSize{[T'-VAR]}); that captured variables have local type (in rules \ruleNameSize{[T'-A-FUN]} and \ruleNameSize{[T'-FOLD]}).
\end{proof}

\subsection{Equivalence between \FSCAFIi{} and \HFCprime{}}\label{ssec:NCi-equiv-HFCprime}

In order to prove the equivalence between \HFCprime{} and \FSCAFIi{}, we first need to define what it means for an expression to have an equivalent behaviour in the two languages.

\begin{defi}[Coherence]
	Assume that $\expTypJud{\LTStypEnv}{\TtypEnv}{\e}{\type}$ in \HFCprime{}/\FSCAFIi{} where $\TtypEnv = \overline\xname : \overline\type$. Let $\overline\anyvalue$ be \HFCprime{} values of type $\overline\type$, and let $\anyvalue(\deviceId)$ be defined by cases as $\lvalue(\deviceId) = \lvalue$, $\fvalue(\deviceId) = \lvalue$ where $\lvalue$ is such that $\envmap{\deviceId}{\lvalue} \in \fvalue$.

	We say that $\applySubstitution{\e}{\substitution{\overline\xname}{\overline\anyvalue}}$ has the \emph{same behaviour} in \HFCprime{} and \FSCAFIi{} whenever:
	\begin{itemize}
		\item if $\type$ is a local type, $\bsopsem{\deviceId}{\Trees}{\senstate}{\applySubstitution{\e}{\substitution{\overline\xname}{\overline\anyvalue}}}{\mkvtree{}{\lvalue}{\overline\vtree}}$ if and only if $\bsopsem{\deviceId,\deviceId}{\Trees}{\senstate}{\applySubstitution{\e}{\substitution{\overline\xname}{\overline\anyvalue(\deviceId)}}}{\mkvtree{}{\lvalue}{\overline\vtree'}}$ for some $\overline\vtree$ and $\overline\vtree'$;
		\item if $\type$ is a field type, $\bsopsem{\deviceId}{\Trees}{\senstate}{\applySubstitution{\e}{\substitution{\overline\xname}{\overline\anyvalue}}}{\mkvtree{}{\fvalue}{\overline\vtree}}$ if and only if $\bsopsem{\deviceId,\deviceId_i}{\Trees}{\senstate}{\applySubstitution{\e}{\substitution{\overline\xname}{\overline\anyvalue(\deviceId_i)}}}{\mkvtree{}{\lvalue_i}{\overline\vtree'}}$ for some $\overline\vtree$ and $\overline\vtree'$ where $\fvalue = \envmap{\deviceId_i}{\lvalue_i}$.
	\end{itemize}
\end{defi}

Assuming that the above coherence condition holds for built-in functions, then it holds for every expression, as shown in the following theorem.

\begin{thm}[Equivalence between \HFCprime{} and  \FSCAFIi{}] \label{thm:equivalence}
	Assume that for every built-in function $\bname$ with local arguments, $\bname(\overline\xname)$ has the same behaviour in \HFCprime{} and \FSCAFIi{} for every substitution of $\overline\xname$. Assume that $\expTypJud{\LTStypEnv}{\TtypEnv}{\e}{\type}$ in \HFCprime{}/\FSCAFIi{} where $\TtypEnv = \overline\xname : \overline\type$, and let $\overline\anyvalue$ be \HFCprime{} values of type $\overline\type$. Then $\applySubstitution{\e}{\substitution{\overline\xname}{\overline\anyvalue}}$ has the same behaviour in \HFCprime{} and \FSCAFIi{}.
\end{thm}
\begin{proof}
	\revised{The proof follows from standard induction on expressions (see Appendix \ref{apx:equivalence}).}
\end{proof}

\subsection{\FSCAFI{} Expressiveness} \label{ssec:expressiveness}

In this section, we argue that \FSCAFI{} is an expressive language for distributed computations. \Cref{sssec:universality} shows that \FSCAFIi{} contains many relevant aggregate programs, including those providing universality for distributed computations. In \Cref{sssec:patterns}, we argue that most of the \HFC{} programs which are not directly interpretable in \FSCAFI{} can be automatically refactored (while preserving their behaviour) in order to fit within \HFCprime{} (hence \FSCAFIi{}), enlarging the class of aggregate programs expressible in \FSCAFI{} to virtually all \HFC{} programs. This process is exemplified in \Cref{sssec:selfstab}, where the self-stabilising fragment of \HFC{} \cite{viroli:selfstabilisation} is ported to \FSCAFI{} through the refactoring just introduced. Finally, \Cref{sssec:scafi_nonaligned} argues that the \FSCAFI{} programs that are not in \FSCAFIi{} can fruitfully extend the expressive power of \HFC{}.

\subsubsection{\FSCAFIi{} examples: G and T blocks and universality} \label{sssec:universality}

\begin{figure}
\begin{lstlisting}[language={hfc}]
def G(source, initial, metric, accumulate) {
  rep ( pair(source, initial) ) { (x) =>
    foldhood( pair(source, initial), min,
              pair(nbr{fst(x)}+metric(), accumulate(nbr{snd(x)})) )
} }
def T(initial, zero, decay) {
  rep ( initial ) { (x) => min(max(decay(x), zero), initial)
} }
\end{lstlisting}
\caption{The G and T blocks (code reported from literature \cite{viroli:selfstabilisation}).} \label{fig:gct}
\end{figure}

The correspondence between \FSCAFIi{}  and \HFCprime{} given by Theorem \ref{thm:equivalence}, although restricted to a fragment of the two languages, is applicable to many relevant aggregate programs, allowing for the automatic transfer of algorithms and properties proven in \HFC{} to \FSCAFI{}.
As a first paradigmatic example, consider the G and T blocks, proposed as part of a combinator basis able to express most aggregate systems \cite{bv:building:blocks}. The \HFC{} formalisation of their code \cite{viroli:selfstabilisation} is reported in \Cref{fig:gct}, and has the same behaviour in \FSCAFI{} since it actually belongs to \HFCprime{}/\FSCAFIi{}.

\begin{figure}
\begin{lstlisting}[language={hfc}]
// previous round value of v
def older(v, null) {
  fst(rep (pair(null, null)) { (old) => pair(snd(old), v) })
}
// gathers values from causal past events into a labelled DAG
def gather(node, dag) {
  let old = older(dag, dag_empty()) in
  let next = dag_join(foldhood(old, dag_union, nbr{dag}), node) in
  if (next == node) { dag } { gather(node, dag_union(dag, next)) }
}
def f_field(e, v...) {
  f( gather(dag_node(e, v...), dag_node(e, v...)) )
}
\end{lstlisting}
\caption{The universal translation of distributed Turing Machines (code reported from literature \cite{abdv:universality}).} \label{fig:universality}
\end{figure}

As a further example, consider the code in \Cref{fig:universality}, which is also in \HFCprime{}/\FSCAFIi{} and encodes the behaviour of any distributed Turing Machine (expressed as a function $\funvalue$ taking as input in every firing the whole collection of causally available data) as an \HFC{} function. This code was used to prove \emph{Turing Universality} for \HFC{},\footnote{A programming model for distributed systems is Turing-universal if and only if it is able to replicate the behaviour of any distributed Turing machine.} but in fact it also proves that \HFCprime{}/\FSCAFIi{} hence \FSCAFI{} are Turing-universal as well (assuming a sufficient collection of built-ins).


\subsubsection{Refactoring of \HFC{} programs into \HFCprime{}} \label{sssec:patterns}

Despite many common aggregate functions being in \HFCprime{}/\FSCAFIi{}, not every relevant such function belongs to this fragment. In particular, the restrictions imposed by the type system in Figure \ref{fig:restrictedTypes} prohibit the common pattern of functions with field arguments, folding those arguments in their body. However, refactoring strategies exist that we can use in order to turn most \HFC{} programs into \HFCprime{} while preserving their behaviour, by converting a field argument of a function into a local argument, thus allowing its capture within $\foldK$ statements.
\begin{enumerate}
    \item[(1)] \emph{Abstracting}: a field argument of a function may be passed ``by name'' through
            \[
            ( (\xname) \toSymK{ \e_1 } )(\e_2) \quad\longrightarrow\quad ( (\xname) \toSymK{ \applySubstitution{\e_1}{\substitution{\xname}{\xname()}} } )(() \toSymK{\e_2})
            \]
\end{enumerate}
This refactoring preserves the behaviour of a program, provided that either \emph{(i)} the abstracted parameter does not depend on the current domain (e.g., relational sensors such as \texttt{nbrRange}), or \emph{(ii)} the parameter does not occur within branches in the body of the function (so that the evaluation is performed in the function with respect to the same domain as in the argument). If instead it does occur within a branch and depends on the current domain, its deferred evaluation within the function body will generally produce a different result. For example, consider function \texttt{counthood} counting the number of neighbours in the current domain:
\begin{lstlisting}[language={hfc}]
def counthood() { foldhood(0, +, 1) }
\end{lstlisting}
and suppose that the argument is $\e_2 = \nbrK\{\mathtt{counthood}()\}$, and the function $\e_1$ is
\[
\ifK (\mathtt{temperature}() < 30) \{ 0 \} \{ \foldK(\mathtt{counthood}(), \max, \xname) \}.
\]
This functions returns 0 on \emph{``cold''} devices; otherwise, it returns the maximum number of \emph{\underline{total}} neighbours that any of the \emph{hot} neighbours has (a metric that may be used, e.g., to trigger an alarm). After the \emph{abstracting} refactoring, this result changes slightly to the maximum number of \emph{\underline{hot}} neighbours that any of the \emph{hot} neighbours has. This reduces the computed value, possibly reducing the effectiveness of the metric and delaying the triggering of the alarm. However, most programs with this characteristic can still be refactored, by resorting to the following extended refactoring whenever the domain dependency is due to sub-expressions of local type.
\begin{enumerate}
    \item[(2)] \emph{Abstracting with parameters}: a field argument $\e_2[\overline\e']$ where $\overline\e'$ have local type can be passed as a function with parameters through
            \[
            ( (\xname) \toSymK{ \e_1 } )(\e_2[\overline\e']) \quad\longrightarrow\quad ( (\xname,\overline\yname) \toSymK{ \applySubstitution{\e_1}{\substitution{\xname}{\xname(\overline\yname)}} } )((\overline\yname) \toSymK{\e_2[\overline\yname]}, \overline\e')
            \]
\end{enumerate}
This refactoring can address the previous example, by leaving the $\mathtt{counthood}()$ sub-expression (which is the one depending on the current domain and has local type) as a parameter:
\begin{align*}
((\xname,\yname) \toSymK{\ifK (\mathtt{temperature}() < 30) \{ 0 \} \{ \foldK(\mathtt{counthood}(), \max, \xname(\yname)) \}}) \quad \\
((\yname) \toSymK{\nbrK\{\yname\}}, \mathtt{counthood}())
\end{align*}

\begin{thm}[Correctness of the abstracting with parameters refactoring]\label{lem:abstracting:rewrite}
	Assume an \HFC{} expression $\e$ consists in a function $(\xname) \toSymK{ \e_1 }$ called with a field argument $\e_2[\overline\e']$ where $\overline\e'$ all have local type, and consider the refactored expression $\e''$ as in (2). If the argument $\xname$ does not occur within branches in $\e_1$, then $\e''$ has the same behaviour as $\e$. If the result of $\e_2[\overline\yname]$, for any given values for $\overline\yname$ and for any domain of computation $\domof{\Trees}$, is a field value $\fvalue$ such that $\fvalue(\deviceId)$ does not depend on $\domof{\Trees}$ for any $\deviceId$, then $\e''$ has the same behaviour as $\e$.
\end{thm}
\begin{proof}
	First, notice that the semantics of \HFC{} is compositional except for the set of aligned neighbours: in other words, the evaluation result of an expression depends on the enclosing expression only in the determination of the set of aligned neighbours; once that set is fixed, the result of the expression is fixed too. The effect of refactoring (2) is shifting the evaluation of $\e_2[\overline\e']$ in a different part of the evaluation tree: function call itself (by Rule \ruleNameSize{[E-APP]}) to several points in the function body (the occurrences of $\xname$).

	If the evaluation of $\e_2[\overline\e']$ does not depend on the domain of evaluation, then this computation shift has no effect, concluding the proof for that case. If $\xname$ does not occur within branches, the set of aligned neighbours for each of its occurrences is the same as the set of aligned neighbours for the function call itself. It follows that the evaluation of $\e_2[\overline\e']$ in place of every occurrence of $\xname$ has to produce the same result as in the argument, concluding the proof.
\end{proof}

The following simpler refactoring can address cases when the argument is a simple $\nbrK$-expression, regardless of branches in $\e_1$.
\begin{enumerate}
    \item[(3)] \emph{Deferring}: an $\nbrK$ in the argument may be transferred into the body as
            \[
            ( (\xname)~\toSymK{ \e_1 } )(\nbrK\{\e_2\}) \quad\longrightarrow\quad ( (\xname)~\toSymK{ \applySubstitution{\e_1}{\substitution{\xname}{\nbrK\{\xname\}}} } )(\e_2)
            \]
\end{enumerate}
This refactoring is a simplification of refactoring (2), abstracting $\nbrK\{\e_2\}$ with respect to the parameter $\e_2$ of local type. In fact, refactoring (2) in this case would produce:
\[
( (\xname,\yname) \toSymK{ \applySubstitution{\e_1}{\substitution{\xname}{\xname(\yname)}} } )((\yname) \toSymK{\nbrK\{\yname\}}, \e_2)
\]
and inlining $(\yname) \toSymK{\nbrK\{\yname\}}$ into the occurrences of the first parameter $\xname$ yields exactly the result of refactoring (3).

Notice that the general refactoring (2) and its common special cases (1) and (3) cover practically any realistic situation. In order for neither of those to be applicable, we would need the argument $\e_2$ to be a field expression which depends on the current domain (and not only because of sub-expressions of local type), and which is passed into a function $\e_1$ folding it inside a branching statement: we are not aware of any meaningful or realistic situation meeting these requirements.
Thus, through these refactorings we can convert practically any \HFC{} program into \HFCprime{}, hence into \FSCAFI{}, although the resulting refactoring may be extensive: since the refactoring of a function depends on its call patterns, more than one version of a function may be needed to cover all of them. However, this is \emph{not} an issue for a ``native'' \FSCAFI{} programmer, which would just interpret the above refactorings as \emph{programming patterns}, following them from the start of application development. Interpreted as programming patterns, the above rules suggest deferring $\nbrK$ applications when possible, or using call-by-name instead of call-by-value on field arguments otherwise.

\subsubsection{An example: refactoring the self-stabilising fragment of \HFC{}} \label{sssec:selfstab}

A paradigmatic example of \HFC{} programs that can be converted into \HFCprime{} through the refactorings in \Cref{sssec:patterns} is the \emph{self-stabilising fragment} \cite{viroli:selfstabilisation}. We say that a time- and space-distributed data is \emph{stabilising} iff it remains constant in every point after a certain time $t_0$, and its \emph{limit} is the value assumed after $t_0$. We say that a distributed program is \emph{self-stabilising} iff given stabilising inputs and topology, it produces a stabilising output which depends only on the limits of the inputs and topology (and not on the concrete scheduling of events, nor on the input values before stabilisation).

A subset of \HFC{}, called \emph{self-stabilising fragment}, has been proved in the literature \cite{viroli:selfstabilisation} to consist of self-stabilising programs only. This fragment makes use of functions folding field arguments, and thus is \emph{not} part of \HFCprime{} and cannot be directly translated into \FSCAFI{} preserving its behaviour (or self-stabilisation). However, the refactorings in \Cref{sssec:patterns} can be applied to obtain the equivalent fragment in Figure \ref{fig:fragment} that does belong to \HFCprime{}/\FSCAFIi{}. The self-stabilising fragment combines syntactic requirements with mathematical requirements on functions, annotated in the figure through superscripts on function names.

\paragraph{$\mathsf{C}$ (Converging)}
A function $\funvalue(\anyvalue_1, \anyvalue_2, \overline\anyvalue)$ is said converging iff, in every firing, its return value is closer to $\anyvalue_2$ than the maximal distance that the two arguments $\anyvalue_1$ and $\anyvalue_2$ have in any neighbour firing (according to any metric measuring that distance).

\paragraph{$\mathsf{M}$ (Monotonic non-decreasing)}
A stateless\footnote{A function $\funvalue(\overline\xname)$ is \emph{stateless} iff given fixed inputs $\overline\anyvalue$ always produces the same output, independently of the environment or specific firing event. In other words, its behaviour corresponds to that of a mathematical function.} function $\funvalue(\xname, \overline\xname)$ with arguments of local type is monotonic non-decreasing in its first argument iff whenever $\lvalue_1 \leq \lvalue_2$, also $\funvalue(\lvalue_1, \overline\lvalue) \leq \funvalue(\lvalue_2, \overline\lvalue)$.

\paragraph{$\mathsf{P}$ (Progressive)}
A stateless function $\funvalue(\xname, \overline\xname)$ with local arguments is progressive in its first argument iff $\funvalue(\lvalue, \overline\lvalue) > \lvalue$ or $\funvalue(\lvalue, \overline\lvalue) = \top$ (where $\top$ is the maximal element of the relevant type).

\paragraph{$\mathsf{R}$ (Raising)}
A function $\funvalue(\lvalue_1, \lvalue_2, \overline\anyvalue)$ is raising with respect to partial orders $<$, $\vartriangleleft$ iff:
\textit{(i)} $\funvalue(\lvalue, \lvalue, \overline\anyvalue) = \lvalue$;
\textit{(ii)} $\funvalue(\lvalue_1, \lvalue_2, \overline\anyvalue) \geq \min(\lvalue_1, \lvalue_2)$;
\textit{(iii)} either $\funvalue(\lvalue_1, \lvalue_2, \overline\anyvalue) \vartriangleright \lvalue_2$ or $\funvalue(\lvalue_1, \lvalue_2, \overline\anyvalue) = \lvalue_1$.

\begin{figure}[!t]
\centering
\begin{small}
\centerline{\framebox[\linewidth]{$
	\begin{array}{lcl@{\hspace{-15pt}}r}
		\s & \BNFcce &  \xname \; \BNFmid \; \anyvalue  \; \BNFmid \; \letK \; \xname = \s \; \inK \; \s \; \BNFmid \; \funvalue(\overline\s) \; \BNFmid \; \ifK (\s) \{ \s \} \{ \s \} \; \BNFmid \; \nbrK\{\s\}
		&   {\footnotesize \mbox{self-stabilising expression}} \\[3pt]
		&&  \; \BNFmid \;  \repK(\e)\{ (\xname) \toSymK{\funvalue^\mathsf{C}(\xname, \s, \overline\e)} \}\\[3pt] 
		&&  \; \BNFmid \;  \repK(\e)\{ (\xname) \toSymK{\funvalue((\yname,\zname) \toSymK{\muxK(\nbrK\{\yname\} < \yname, \nbrK\{\xname\}, \zname)}, \overline\s)} \}\\[3pt] 
		&&  \; \BNFmid \;\repK(\e)\{ (\xname) \toSymK{\funvalue^\mathsf{R}(\foldK(\s, \min, \funvalue^\mathsf{MP}(\nbrK\{\xname\}, \overline\s)), \xname, \overline\e)} \} 
	\end{array}
	$}
}
\end{small}
\caption{Syntax of a self-stabilising fragment of field calculus expressions, where self-stabilising expressions $\s$ occurring inside a $\repK$ statement cannot contain free occurrences of the $\repK$-bound variable $\xname$.}
\label{fig:fragment}
\end{figure}

\begin{thm}[Self-stabilising fragment of  \FSCAFI{}]
	Every self-stabilising expression according to the fragment in \Cref{fig:fragment} is a self-stabilising \FSCAFI{} expression.
\end{thm}
\begin{proof}
	Notice that the fragment in \Cref{fig:fragment} can be obtained from that in \cite[Figure 2]{viroli:selfstabilisation} by means of the following two refactorings.
	\begin{itemize}
		\item In the \emph{converging} $\repK$ expression $\repK(\e)\{ (\xname) \toSymK{\funvalue^\mathsf{C}(\xname, \s, \overline\e)} \}$, deferring is applied to transfer $\nbrK$ in arguments to the body of the function $\funvalue^\mathsf{C}$, adapting the definition of \emph{converging} function accordingly. By Theorem \ref{lem:abstracting:rewrite}, these expressions have the same behaviour as those in \cite[Figure 2]{viroli:selfstabilisation}, which are self-stabilising by \cite[Theorem 1]{viroli:selfstabilisation}.
		\item In the \emph{acyclic} $\repK$ expression $\repK(\e)\{ (\xname) \toSymK{\funvalue((\yname,\zname) \toSymK{\muxK(\nbrK\{\yname\} < \yname, \nbrK\{\xname\}, \zname)}, \overline\s)} \}$, abstraction is applied to the first argument with respect to parameters $\yname, \zname$ (whose values are added into the additional parameters $\overline\s$ passed to $\funvalue$). By Theorem \ref{lem:abstracting:rewrite}, these expressions have the same behaviour as those in \cite[Figure 2]{viroli:selfstabilisation}, which are self-stabilising by \cite[Theorem 1]{viroli:selfstabilisation}.
	\end{itemize}
	The rest of the fragment is identical modulo inessential modifications (expansion of the $\nbrlt$ and $\texttt{minHoodLoc}$ functions into their definition), concluding the proof.
\end{proof}

\subsubsection{\FSCAFI{} programs that cannot be straightforwardly expressed as \HFC{} programs} \label{sssec:scafi_nonaligned}

As argued in \cite{eCAS:alignment}, programs such as \emph{updatable metrics} and \emph{combined Boolean restriction} are not conveniently expressed in \HFC{}. In the former case, we can use the following general scheme for updatable functions, first proposed in \cite{audrito2019tocl}:
\begin{lstlisting}[language={hfc}]
def up(injecter) {
  snd( rep(injecter()) {
    (x) => { foldhood(injecter(), max, nbr{x}) }
} ) }
\end{lstlisting}
where $\mathtt{injecter}$ is a function returning a pair $\langle \mathtt{version~number}, \mathtt{~function~code} \rangle$, and the built-in operator $\mathtt{max}$ selects the pair with the highest version number among its arguments. This procedure defines a perfectly reasonable ``upgradeable function'' by spreading functions with higher version number throughout devices. However, it is not allowed by the type system of \HFC{} for functions returning fields, such as metrics (which usually have type $() \rightarrow \ftypeOf{\ntype}$). 
This scheme can instead be used in \FSCAFI{} (as shown in Section \ref{sec-case}), and works properly provided that new versions are injected at a slow rate, and an occasionally empty domain of a field-like expression does not produce critical effects.

Another situation where the permissive behaviour of \FSCAFI{} is crucial is that of \emph{combined Boolean restriction}. In this setting, a field-like expression $\e$ needs to be restricted to those devices agreeing on the value of $n$ Boolean parameters $\bname_1, \ldots, \bname_n$, before being folded with $\foldK(\anyvalue, \funvalue, \e)$. This rather abstract example might be concretely instantiated, e.g., in case an aggregation needs to be executed separately on devices with different configurations. In \HFC{}, this effect can be achieved only by restricting on each of the $2^n$ possibilities for the parameters, as in the following.
\begin{lstlisting}[language={hfc}]
if (b1 && b2 && ... bn) {foldhood(v, f, e)} {
  if (!b1 && b2 && ... bn) {foldhood(v, f, e)} {
    if (b1 && !b2 && ... bn) {foldhood(v, f, e)} { ... }}}
\end{lstlisting}
However, such a program might be infeasibly large even for small values of $n$. On the other hand, in \FSCAFI{} the above program can be concisely rewritten as:
\begin{lstlisting}[language={hfc}]
foldhood(v, f, e + if (b1) {nbr{0}} {nbr{0}} + if (b2) {nbr{0}} {nbr{0}} + ...)
\end{lstlisting}
whose size is linear in $n$.\footnote{In this code we assumed that $\xname$ has numerical type, but similar code can be obtained for any type by defining a binary operator which is the identity on its first argument.} The domain of the $i$-th 0-valued field-like subexpression above is equal to the set of devices agreeing on $\bname_i$, hence by intersecting all of them the resulting domain corresponds to the set of devices agreeing on each of the $n$ given parameters.

\section{\scafi{}: an Implementation of the \FeatherweightSCAFI{}}\label{sec-NC-at-work}

In this section,
 we present \scafi{},
 an implementation of \FSCAFI{} embedded in the Scala programming language.

While \emph{external} (or \emph{standalone}) domain-specific languages (DSLs)
 have their own syntax and semantics,
 \emph{internal} DSLs are \emph{embedded} into some other language~\cite{voelter2013dsl}.
Developing an internal DSL
 brings major benefits in terms of \emph{reuse}
 of features, toolchain support, and familiarity with
 the host language---at the expense of the syntactical and semantical constraints imposed by the host.
Scala has been chosen as the host language of \scafi{}
 for its
 modernity,
 its feature set (it provides mechanisms supporting the creation of expressive APIs and  DSLs~\cite{DBLP:conf/icfem/ArthoHKY15})
 and
 its ability to target and reuse libraries from various execution platforms
 (such as the JVM, but also Javascript through the ScalaJS project~\cite{doeraene2018cross-platform-language}).
With respect to other field calculi,
 \FSCAFI{} lends itself to smooth implementation in Scala,
 as e.g. there are no neighbouring fields to be dealt with
 at the typing, syntactical, and semantical level.
A Scala expression of type \texttt{Int} is automatically interpreted in \scafi{} as an aggregate program producing a field of \texttt{Int}, completely fulfilling the ``everything is a field'' view.

\subsection{The Scala implementation of \FSCAFI{}}

\setLanguage{scafi}

The following interface, implemented as a Scala trait, represents the basic \FSCAFI{} constructs as methods:
\begin{lstlisting}[escapechar=\%]
trait Constructs {
  // Key constructs
  def rep[A](init: => A)(fun: (A) => A): A
  def foldhood[A](init: => A)(aggr: (A, A) => A)(expr: => A): A
  def nbr[A](expr: => A): A
  def aggregate[A](b: => A): A

  // Abstract types
  type ID;         // type of device identifiers
  type LSNS, NSNS; // type of local and neighbour sensor names

  // Contextual, but foundational
  def mid(): ID
  def sense[A](name: LSNS): A
  def nbrvar[A](name: NSNS): A
}
\end{lstlisting}
In Scala, methods are introduced with the \texttt{def} keyword,
can be generic (with type parameters specified in square brackets),
may accept multiple parameter lists,
and specify a return type at the end of the signature (when this is not given the compiler attempts inference).
Function types may take the form \texttt{($I_1$,$\ldots$,$I_N$)=>$O$}, which is actually syntactic sugar over \texttt{FunctionN[$I_1$,$\ldots$,$I_N$,$O$]};
 curried function types can be written as \texttt{$I_1$=>$\cdots$=>$I_N$=>$O$} (\texttt{=>} is right associative).
Tuple types may take the form \texttt{($T_1$,$\ldots$,$T_N$)}, which is actually syntactic sugar over \texttt{TupleN[$T_1$,$\ldots$,$T_N$]};
 similarly, a literal tuple value can be denoted as \texttt{($v_1$,$\ldots$,$v_N$)}.
By-name parameters, denoted with type \texttt{=>$T$},
capture expressions or blocks of code that are passed unevaluated to the method and are actually evaluated every time the parameter is used---they are basically syntactic sugar over 0-ary function types.
As a relevant note on syntax, especially useful in DSLs to render constructs with code blocks, unary parameter lists in a method can be called also with curly brackets instead of parentheses.
E.g., all the following are valid invocations for \texttt{rep} method above: \texttt{rep($\cdot$)($\cdot$)}, \texttt{rep($\cdot$)\{$\cdot$\}}, \texttt{rep\{$\cdot$\}($\cdot$)}, \texttt{rep\{$\cdot$\}\{$\cdot$\}}.
Finally, nullary methods can be invoked without parentheses;
e.g., \texttt{mid} is a valid method call just like \texttt{mid()}.

First of all, notice that the input and output types of the constructs (especially \lstinline|nbr|) are proper, i.e., no field-like datatype appears in the signatures.
Then, compare trait \lstinline|Construct| with the \FSCAFI{} syntax from \Cref{sec-calculus-syntax}.
Beside the curried form of \lstinline|rep| and \lstinline|foldhood|, the only significant difference is an additional construct \lstinline|aggregate| which is used to turn standard Scala functions/methods into field calculus functions, i.e., as units of alignment (refer to~\Cref{sec-calculus-device-semantics-overall} for details).
We report here the encoding of $\ifK$, which is called \lstinline|branch| in \scafi{} due to the fact that \lstinline|if| is a reserved keyword in Scala:
\begin{lstlisting}
def branch[A](cond: => Boolean)(th: => A)(el: => A): A =
  mux(cond)(() => aggregate{ th })(() => aggregate{ el })()
\end{lstlisting}
Namely, the \lstinline|th|en and \lstinline|el|se expressions are passed unevaluated (as by-named arguments) to \lstinline|branch|, and there used in the bodies of two corresponding lambdas, wrapped with \lstinline|aggregate| (which could be seen as a form of tagging).
The \lstinline|mux| on \lstinline|cond| is used to select one of those functions, which is then invoked straight on via operator \lstinline|()|.

Moreover, it is crucial to note that
 the programming patterns discussed in \Cref{sssec:patterns} (abstraction and deferring)
 are supported
 in a syntactically transparent way
 thanks to the Scala support of by-name parameters.
Indeed, the refactoring shown in \Cref{sssec:patterns}
 are automatically applied.

Finally, in practice, an aggregate program can be defined in \scafi{}
 by extending class \lstinline|AggregateProgram| and implementing a method called \lstinline|main| which defines the main program expression.
This class provides an implementation of the \lstinline|Constructs| trait: by subclassing it, the syntax and semantics of the \scafi{} DSL is made available within the subclass definition, such that its objects could be used both for simulation purposes or to encapsulate the aggregate logic in actual distributed implementations.
\begin{lstlisting}
class GradientProgram extends AggregateProgram
  with Gradients        // brings gradient() in
  with StandardSensors  // brings nbrRange() in
{
  // Custom definitions
  def isSource = sense[Boolean]("source")

  // Main expression
  def main(): (Int,Double) = {
    // Return a 2-element tuple (round count, gradient value)
    (rep(0)(r=>r+1), gradient(isSource, nbrRange))
  }
}
\end{lstlisting}
A full description of \scafi{} is beyond the scope of this paper: the interested reader can refer to~\cite{cv2017actorcas} for more details.

\subsection{On pragmatics: neighbouring fields vs computation against a neighbour}
\label{nbrfields-vs-nc-in-practice}

In this section, we provide a hint to the practical issues that \FSCAFI{}/\scafi{} solve, enabling easy implementations \emph{and} a clean support for  aggregate programming.

Consider the \scafi{} expression of the classic gradient algorithm (cf. \Cref{ex:time-nbr-gradient}):
\begin{lstlisting}
rep(Double.PositiveInfinity) { case d =>
  mux(source) { 0.0 } { minHood(nbr(d) + metric()) }
}
\end{lstlisting}
and focus on expression \lstinline|nbr(d)+metric()|.
While in \FSCAFI{}/\scafi{} the latter is a sum of two local values, computed against aligned neighbours (by the \lstinline|*hood| operator), in the field calculus it would be a sum of two neighbouring fields.
While in a standalone DSL, implementations may transparently
 lift local operators to work with fields (i.e., collection-like objects),
 this would not typically be easy or even possible in internal DSLs.
Suppose a neighbouring field of \lstinline|T|-typed objects is represented as an instance of type \lstinline|Field[T]|.
In a purely functional world, two fields could be combined with a ``classical'' \lstinline|map2| function:
\begin{lstlisting}
map2(nbr(d), metric(), _+_)
\end{lstlisting}
However, this would introduce some clutter in aggregate programs
and require a versatile set of functional combinators.
In a powerful object-oriented language with mechanisms like extension methods (such as Scala or C\#), it would be possibly to manually add specific operators to work with fields.
\begin{lstlisting}[escapechar=\%]
implicit class NumericField[T:Numeric](f: Field[T]){
  def +(f2: Field[T]): Field[T] = map2(f, f2, implicitly[Numeric[T]].plus(_,_))
  // ...
}
// Assuming nbr(d) has type Field[Double]
nbr(d) + metric // possible by Scala's implicits (cf. %\cite{oliveira2010typeclasses}%) + operator syntax
\end{lstlisting}
This example leverages Scala's \emph{pimp-my-library idiom}~\cite{oliveira2010typeclasses}, which uses implicits to enable automatic, compile-time conversion of a \lstinline|Field| of \lstinline|Numeric|s to a \lstinline|NumericField| that could accept a method \lstinline|+|.
However, the lifting of local operators to fields would be manual, and not automatic---introducing a overhead on library development.
Some language could provide advanced features to mitigate this issue (e.g., through powerful compile-time macro mechanisms), but in this case we would pose severe constraints on the host language for the internal DSL, hence limiting the development environment where aggregate programming could be supported.

\section{Case Study} \label{sec-case}

\setLanguage{scafi}



The goal of this section is to show how \FSCAFI{}/\scafi{}
 can be used to eloquently express collective behaviour.
%
%
More specifically, the goal is not to show that \FSCAFI{}/\scafi{}
 necessarily provide a better programming support with respect to e.g. \HFC{}/Protelis---though the Scala embedding does provide various practical benefits, through reuse of Scala's powerful type system and features.
Indeed, the benefit of \FSCAFI{} lies in a conceptual frugality
 (by avoiding the ``neighbouring field'' notion)
 that leads to a different computational model (cf., \Cref{sec-properties})
 with easier DSL embeddability in mainstream programming languages (cf. \Cref{sec-NC-at-work}).
Therefore, this section shows that such a model and implementation can be used to smoothly program increasingly complex field coordination-based systems.
In particular, we apply the programming model to the context of
edge computing~\cite{shi2016edgecomp}
and ad-hoc cloudlets~\cite{chen2015adhoccloudlet} (introduced in \Cref{case:edge-comp}).
Evaluation of correctness of the proposed solution is performed by simulation.
The source code, configuration files, and instructions for launching the experiments and reproducing the presented results is publicly available at an online and permanent repository\footnote{\url{https://github.com/metaphori/experiments-neighbours-calculus}.}~\cite{zenodo-experiment}.

\subsection{Ad-hoc Edge Computing Support}
\label{case:edge-comp}

Edge computing~\cite{shi2016edgecomp}
 is a recent paradigm that aims to bring
 cloud-like functionality (e.g., virtualisation, elastic provisioning of resources, and ``anything-as-a-service'') closer to end devices
 at the edge of the network.
Its goal is to ultimately improve the Quality of Service (QoS) of situated systems as in the IoT,
 e.g., through reduction of communication latencies, bandwidth and energy consumption,
 essentially achieved by shortening the round-trip path of data to/from elastic resource pools.

In this case study for \scafi{}, we focus on \emph{ad-hoc} edge computing---i.e.,
 a decentralised form of edge computing
 that does not rely on pre-existing infrastructure
 but is rather supported by a large-scale network of devices
 logically interacting in a peer-to-peer fashion.
This may also be a first step for merging volunteer computing~\cite{durrani2014volunteercomputing} and edge computing into
 a form of \emph{volunteer edge computing}
 where devices with spare resources make them available (possibly
 by incentives) to nearby users.

We assume to operate in a region of a smart city
 where hundreds of devices willing to offer resources
 run the Aggregate Computing middleware
 and may hence participate in the Aggregate Edge Computing (AEC) application.
In a vision of smart city operating systems,
 it is indeed sensible to assume that devices willing to fully benefit from smart city services
 are asked to give their ``social contribution'' by participating in the system as well, where, of course, we also assume proper security and privacy systems are in place.
The system is \emph{open}, in the sense that devices may
 enter or leave the system as they like.
For this paper we abstract from the way in which tasks are assigned to devices.

The goal of AEC is to support a \emph{self-organisation process}
 to monitor resource availability and usage (i.e., load) in the system,
 and spread information as well as directives
 which can be exploited by devices for decentralised control activities
 and by users for determining advantageous deployment options.
The idea is to provide support for edge computing by leveraging
 the space-time distribution of resources.
 Managing edge resources in a collective, space-time-oriented fashion
  is meant to provide the following benefits:
 \emph{(i)} it makes straightforward to exploit the locality principles that often sustain IoT, CPS, and smart city applications,
 \emph{(ii)} it allows us to assume that devices are only able to communicate with other nearby devices through short-range wireless technologies (e.g., $10$ to $100$ meters range);
 \emph{(iii)} it allows us to relate the spatio-temporal distribution of resources
 with human social activity;
 and
 \emph{(iv)} it provides a natural way to handle collective metrics
  and take into account contextual data and constraints at multiple scales.


\subsection{Model and Design}

We model devices of the AEC system as self-aware, situated entities
 that can sense
 the resources they own (e.g., number of CPUs, memory, and storage)
 and their current levels of utilisation.
\begin{lstlisting}[]
def getLocalProcessing() = sense[ProcessingStat]("processing")
def getLocalStorage() = sense[StorageStat]("storage")
// ...
def getLocalResources() = LocalStats(getLocalProcessing, getLocalStorage, ...)
\end{lstlisting}
Data can be easily modelled in Scala with case classes (i.e., algebraic data types with built-in equality, pattern matching support, etc.).
\begin{lstlisting}[]
// Just for simplcity of exposition, CPUs are not shared.
case class ProcessingStat(numCPUs: Int = 0, mHz: Int = 0, numCPUsInUse: Int = 0)

trait DiskType; case object HDD extends DiskType; case object SSD extends DiskType
case class StorageStat(MiBs: Int = 0, dtype: DiskType = HDD, MiBsInUse: Int = 0)

case class LocalStats(ps: ProcessingStat, ss: StorageStat, ...)
\end{lstlisting}
We also assume the devices can communicate with other devices
 in their vicinity---i.e., the neighbouring relationship is basically spatial.

To divide the complexity of the resource monitoring task,
 we split the whole environment
 into manageable \emph{areas} (which may also be called \emph{cloudlets})
 that can be monitored and controlled more easily.
In order to solve issues related to distributed consensus,
 we want the system to resiliently self-organise to elect a leader in each area,
 so as to provide consistent views of the state of areas
 as well as enacting global decisions for an area.
That is, the leaders effectively act as decentralised management points
 which are responsible for orchestrating a subset of devices in the system---so, in a sense, we apply a hybrid approach that combines
 orchestration and choreography as suggested in \cite{velasquez2018fogorchestration}.
In practice, leaders act as sinks of contributions from the devices of the
 corresponding area, and as propagators of area-wide context information
 back to its feeders;
 but, in order to implement such a process,
 we need mechanisms for resiliently streaming information
 across dynamic space-time regions.
The following section shows library functions that can be used
 to streamline the implementation of aforementioned mechanisms.

\subsection{Building Blocks}

Implementing the above ideas requires a careful crafting of the aggregate specification.
There is need of identifying some recurrent patterns
 (e.g., peer-to-peer propagation of information in- and out-ward an area)
 that may come handy in many situations.
They can then be used to raise the abstraction level
 through a reusable API that hides low-level detail (e.g. the interplay of \texttt{rep} and \texttt{nbr})
 only exposing the functional contract
 expressing the relationship between input fields and output fields.
%
%
%
In doing so, we progressively move from a device-centric view
 (i.e., local sensing and neighbour-oriented communication)
 to an aggregate view (i.e., collective behaviour and data)
 of the system, regaining declarativity and intent.

In this case study,
 we leverage the following building blocks,
 which are available in the \scafi{} standard library.
\begin{lstlisting}[]
case class Gradient(algorithm: (Boolean,()=>Double)=>Double,
                    source: Boolean, metric: () => Double)
def G[V](gradient: Gradient, field: V, acc: V => V): V
def C[P: Bounded, V](potential: P, acc: (V, V) => V, local: V, Null: V): V
def S(grain: Double, metric: () => Double): Boolean
\end{lstlisting}

Their implementation is already described in other papers~\cite{BPV-COMPUTER2015,viroli:selfstabilisation}
 so here we present them by a functional perspective
 and just provide some insight about how they could work internally.

Recall the notion of gradient from \Cref{ex:time-nbr-gradient}:
 a field computation that,
 from a Boolean field denoting a region,
 returns the field of minimum distances from that region.
The gradient field is key in spatial computations
 because it provides a direction from any point to a target:
 by descending (or ascending, resp.) the gradient field, i.e.,
 by following the minimum (or maximum, resp.) values locally observed,
 one moves close (or away, resp.) to a certain target.
With case class \texttt{Gradient}, we package together the algorithm
 and input fields---e.g., it may be useful to ensure that both a particular gradient
  and another function using it adopt the same metric.
The \emph{gradient-cast} operator, also known as \texttt{G}, is a generalised form of the gradient computation
 that accumulates values through a function \texttt{acc}
 along a gradient field, starting with the given \texttt{field}
 at sources and new devices.
So, functionally \texttt{G} maps a source field and a field of values
 to a field whose values are progressively transformed, distance-wise,
 from source values.
For instance, a gradient field can be constructed via \texttt{G} as follows.
\begin{lstlisting}[]
val src = sense[Boolean]("source")
val classicGradient = Gradient(gradient(_,_), src, nbrRange)
G[Double](classicGradient, mux(src){ 0 }{ inf }, _+classicGradient.metric())
\end{lstlisting}
With \texttt{G}, it is also trivial to define a \texttt{broadcast} function for propagating
 information outward from a source, by ascending the potential field.
\begin{lstlisting}[]
def broadcast[V](source: Boolean, field: V): V =
  G[V](classicGradient, field, (v: V) => v)
\end{lstlisting}
A change of \texttt{field} where \texttt{source} is true
 basically generates a new wave
 in this continuous streaming of information.
Notice that \texttt{field} has a value everywhere
 but only the value of the \texttt{source} is selected
 and then identically preserved in the propagation.
Also, notice that both \texttt{G} and \texttt{broadcast} may
 admit a different signature where the ``potential field''
 is given as input and not calculated internally.

Dual to the gradient-cast operator is the \emph{converge-cast} operator,
 also known as \texttt{C}\revised{~\cite{DBLP:journals/cee/AudritoCDPV21}}:
 indeed, the conceptually inverse notion of a propagation from a point outward
 is a propagation inward, from the edge towards a point.
Functionally, \texttt{C}
 \texttt{acc}umulates a \texttt{local} field
 along some \texttt{potential} field
 (yielding \texttt{Null} if no direction to descend can be found),
 where the final results of the accumulation
 end up collected at points of zero potential.
In the signature above, we show that we may abstract from the concrete
 type \texttt{P} of the potential field (which is usually a \texttt{Double}),
 provided that there exists an ordering over \texttt{P} values---as
 enforced by the context bound constraint \texttt{P:Bounded} (which is the mechanism
 for the typeclass idiom in Scala~\cite{oliveira2010typeclasses}, where the method may be called provided
 an implicit instance of type \texttt{Bounded[P]} can be statically resolved).
It is also worth pointing out that both \texttt{G} and \texttt{C}
 are self-stabilising operators,
 whose compositions are also self-stabilising~\cite{viroli:selfstabilisation} (see the discussion in \Cref{sssec:universality});
 hence, there are formal guarantees that programs built using only
 self-stabilising constructs like \texttt{G} and \texttt{C}
 will eventually reach a fixpoint once inputs stop changing.

The last fundamental building block that we cover is the \emph{sparse-choice operator}, also known as \texttt{S},
 which yields a Boolean field holding \texttt{true}
 in correspondence of elected nodes (also called \emph{leaders})
 which are selected in a manner such that adjacent leaders
 are at about distance \texttt{grain}.
In practice, it performs a leader election process
 which also results in a split of the space
 into areas of a diameter which is approximately \texttt{grain},
 since any connected device will be at a distance at most \texttt{grain}
 from a leader (think of a gradient from the leader field).
Internally, an implementation of \texttt{S}
 may work by performing a distance competition among devices:
 initially, every device will propose itself as a leader,
 but initially assigned random values (ultimately discriminated by device IDs in case of breakeven) will be used to break symmetry (e.g., by selecting the minimum);
 if the gradient from the currently selected leader is larger than \texttt{grain}
  at a device, then it will start another competition for leadership of another area.

\subsubsection{Upgradeable functions}
The case study also uses the notion of \emph{upgradeable function}
 explained in \Cref{sssec:scafi_nonaligned}.
The idea is to support hot upgrades of application algorithms
 by allowing injection of versioned functions at some device of the system
 and then spreading such novel versions through a gossiping process
 that eventually settles everywhere the function of the higher version.
Concretely, we model versioned functions and injectors as types of the form
\begin{lstlisting}[]
case class Fun[R](ver: Int, fun: ()=>R)
type Injecter[R] = () => Fun[R]
\end{lstlisting}
and define a function \texttt{up},
\begin{lstlisting}[]
def up[R](injecter: Injecter[R]): Fun[R] = rep(injecter()){ f =>
  foldhood(injecter())((f1,f2) => if(f1.ver>f2.ver) f1 else f2)(nbr{f})
}
\end{lstlisting}
 to handle the upgrade process
 through corresponding propagation and selection of functions.
%

\subsection{Implementation}
\label{sec:cs:impl}

Considering the modelling and library functions
 covered in previous sections,
 the core specification of the case study can be expressed
 as in the following \scafi{} program.
\begin{lstlisting}[]
class AdhocCloud extends AggregateProgram
  with StandardSensors with BlockG with BlockC with FieldUtils {

  val grain = 150 // parameter: mean distance between leaders

  def main = {
    // Get a metric for S, dynamically
    var Smetric = up[Double](metricInjecter(k)).fun
    // Resilient leader election
    val leaders = S(grain, Smetric)
    // Use the standard neighbouring range metric for the other operations
    val metric = nbrRange
    // Potential field to leaders
    val potential = gradient(leaders, metric)
    // Query local resources
    val localResources = getLocalResources()
    // Collect resources information towards leaders
    val resourcesPerArea =
      C[Double,Stats](potential, _++_, Map(mid->localResources), Map())
    // Broadcast from leaders outwards aggregated statistics
    val allResourcesInArea = broadcast(potential, metric,
      resourcesPerArea.values.foldLeft(AggregateStats())(_.aggregation(_)))
    // ...
  }
}
\end{lstlisting}
The flow is simple and reasoning is simplified by the
 collective stance and compositionality of Aggregate Computing:
 we use \texttt{S} to split the system into cloudlets
 monitored by corresponding leaders
 (where the metric used to create the areas can be dynamically updated);
 we build a gradient from leaders to set the pathways for collecting
 data into leaders and propagating the determinations of leaders to their workers;
 and finally, we use those pathways to stream
 the cloudlet-wise estimation of resource availability and usage to all the members.
By the way, notice that this resource monitoring example represents the application of a more general pattern~\cite{casadei19scr} that finds application in other large-scale, distributed coordination scenarios such as for situated problem solving~\cite{casadei2019scc} and client/server task allocation~\cite{casadei2019fmec}.

\subsection{Evaluation: Correctness} 

\begin{figure}[t!]
\centering
\subfloat[Initial metric.]{\includegraphics[width=0.48\textwidth]{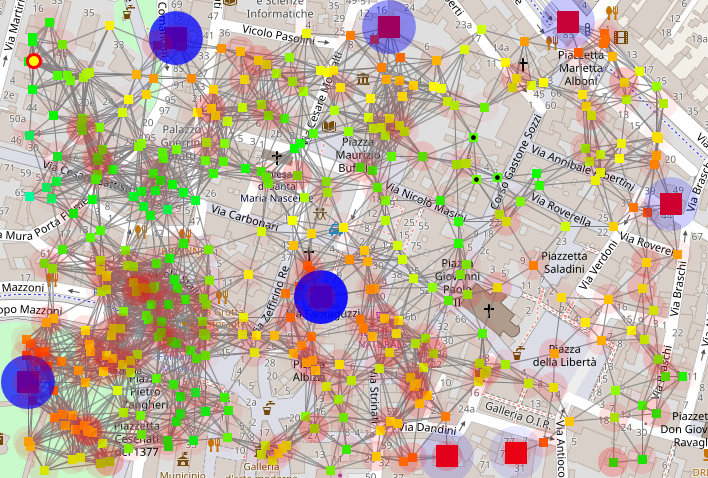}}
\hfill
\subfloat[Updated metric.]{\includegraphics[width=0.48\textwidth]{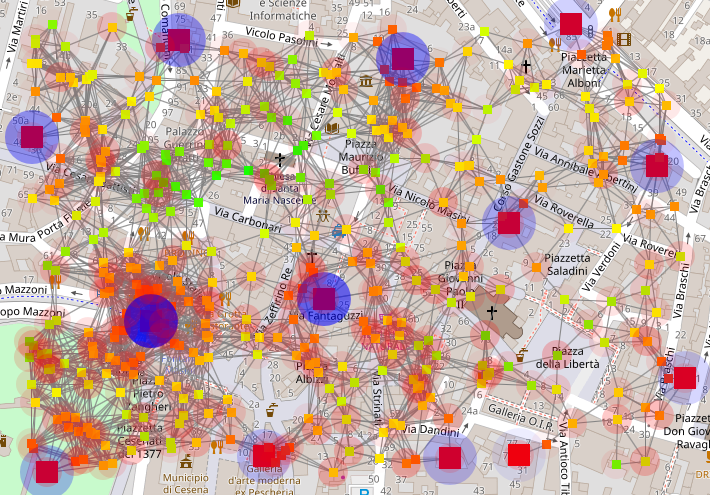}}
\caption{A visual representation of the scenario, with two snapshots taken before and after a hot update of the metric.
The big red squares represent leader nodes. Superimposed with any leader is a semi-transparent blue ball whose intensity is proportional to the amount of resources available in the corresponding area. Other nodes have a colour representing the gradient field towards leaders (warmer, red-like colours for nodes close to the leader, and colder, green-like colours for peripheral nodes).
The updated metric in (b) provides an overestimation of distances which results
 in smaller areas or, equivalently, in more cloudlets (and leaders).}
\label{fig:scenario}
\end{figure}

\begin{figure}[t!]
\centering
\subfloat{\includegraphics[width=0.48\textwidth]{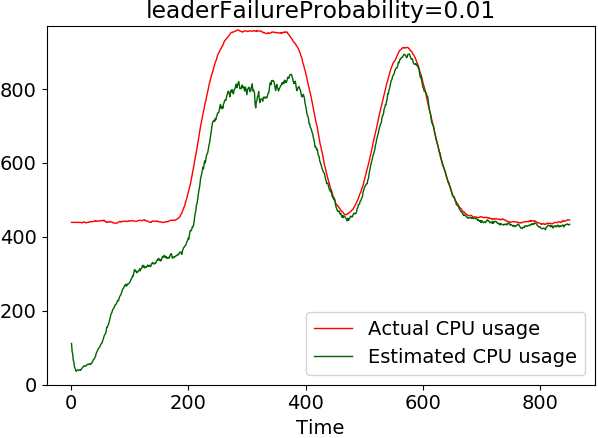}}
\hfill
\subfloat{\includegraphics[width=0.48\textwidth]{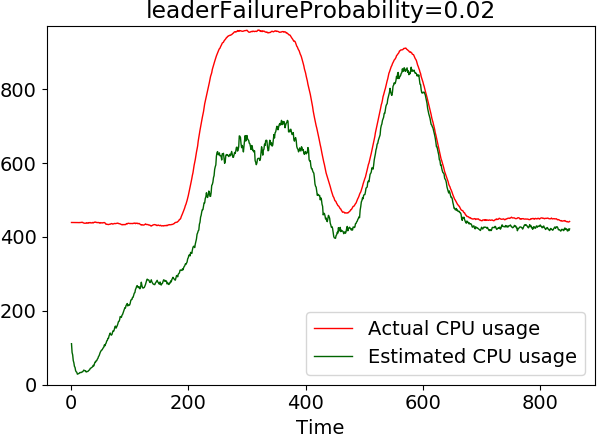}}
\caption{These graphs show the load estimation capability at leader nodes,
for different round-wise probabilities of leader failure (leading to a disproportionate  number of failures ($150$ to $250$) in the considered timespan). The plots are obtained by taking the mean of $30$ runs with different random seeds.
The scenario is modelled as follows: the spikes of high load start at time $t=200$ and $t=500$, the new metric is injected at time $t=450$, and we can observe that,
 after such upgrade, the system is able to perform a much more precise estimation of CPU usage.}
\label{fig:sim}
\end{figure}


%
The system has been simulated in the Alchemist simulator~\cite{alchemist-jos2013,casadei2016simulating}---a graphical snapshot is given by \Cref{fig:scenario}\footnote{Full-size colour pictures are available at the provided repository.},
whereas \Cref{fig:sim} provides empirical evidence of its correctness.
The simulated system consists of $400$ devices, each with a random amount of resources,
 dispersed unevenly around the downtown of Cesena, Italy. Mobility is not considered.
A certain level of load is simulated in the system,
 and we perturb it with spikes of high load.
The goal of the system is to reactively adjust the estimate of the load.
As for additional perturbations,
 we also introduce a probability for temporary failure of leaders.
Finally, in order to show the positive impact of updatable metrics,
 we inject a new metric to fine tune the dimension of areas at runtime---other aspects of configuration can be examined in the repository.

\section{Related Work} \label{sec-related}

Scenarios like the IoT, CPS, smart cities, and the like,
 foster a vision of rich computational ecosystems
 providing services by leveraging strict cooperation of large collectives of smart devices, which mostly operate in a contextual way.
Engineering complex behaviour in these settings calls for approaches (from formal languages to execution platforms) providing some abstraction of the notion of \emph{ensemble}, neglecting as much as possible the more traditional view of focussing on the single device and the messages it exchanges with peers.
Several works developed in different research communities share this attempt, often using different terminology, witnessed by various surveys, focussing on organisation of aggregates of devices
\revised{\cite{SpatialIGI2013,Viroli-et-al:JLAMP-2019,DBLP:journals/corr/abs-2201-03473}}, on developing frameworks for general-purpose self-organisation \cite{jsa}, addressing the issue of autonomic communication \cite{autonomicommunications}, and so on.
By factoring out common ideas from these works, and neglecting the diversity in lower-level concepts, one can identify families of approaches which have relations with the programming framework developed in this paper.
So, whereas the main related work is deeply covered in \Cref{sec-properties},
 in the following we comprehensively describe the research area in which our contribution can be positioned.

\subsubsection*{Device abstraction}
A first family is that of languages providing some form of abstraction over device behaviour and interaction.
TOTA~\cite{tota} and SAPERE~\cite{Zambonelli2015} define platforms for pervasive computing focussing on agent coordination, where agents indirectly interact by injecting/perceiving ``tuples'' equipped with diffusion and aggregation behaviour (Java-defined in TOTA, declaratively specified by rules in SAPERE), and resulting in ``fields of tuples'' spread over the network; such works provide archetypal approaches to create computational fields that were an inspiration for Aggregate Computing and \FSCAFI{}/\scafi{}---though, differently from them, we provide an expressive language to better control the shape and dynamics of fields over time.
Hood~\cite{hood} defines data types to model an agent's neighbourhood and attributes, with operations to read/modify such attributes across neighbours, and a platform optimising execution of such operations by proper caching techniques; \scafi{} provides a comparable neighbourhood abstraction, in that operator \texttt{nbr} is essentially used to declaratively access the set of neighbours as well as to combine observation of neighbours' attributes (indicated by the expression passed as argument) with modification of the same attribute locally.
Finally, works such as
\revised{ActorSpace~\cite{callsen1994actorspace},
 SCEL~\cite{SCEL,DBLP:journals/scp/AlrahmanNL20},
 and AbC~\cite{DBLP:journals/scp/AlrahmanNL20}} rely on so-called \emph{attribute-based communication}, where each actor/agent exposes a list of attributes, and communication can be directed to the \emph{group} of actors whose attribute match a given pattern; \scafi{} can achieve a similar expressiveness with construct \texttt{branch}, by which one can define subcomputations carried on by a subset  of nodes, which are those that execute the same branch and hence remain actually ``observable'' by operator \texttt{nbr}.
Generally speaking, it is worth noting that Aggregate Computing and \FSCAFI{}/\scafi{} address the key feature of fitting useful device abstractions (such as neighbourhood, message exchange, attribute-based filtering) into a purely functional approach, which can then smoothly interoperate with more traditional programming frameworks and languages.

\subsubsection*{Geometric/topological abstractions} Another class of related approaches falls under the umbrella of languages to express geometric constructs and topological patterns. In fact, in several application contexts concerning environment sensing and controlling, what is key is the physical (geometric, topological) shape that coalitions of mobile agents take, or that certain data items create while diffusing in the environment. In the Growing Point Language~\cite{coore1999botanical}, an \emph{amorphous medium}~\cite{bealamorphous} (essentially defined by an ad-hoc network) can be programmed by a nature-inspired approach of ``botanical computing'', where computational processes are seen as ``growing points'' increasingly expanding across neighbours until reaching a fixpoint shape defined by declarative constraints; \FSCAFI{}/\scafi{} and Aggregate Programming work on similar hypothesis on structure and behaviour of the underlying network, though adopting a different functional paradigm that is more expressive as it can address dynamical aspects as well, and which could be used as a lower-level language to reproduce the expressiveness of the growing point abstraction. The Origami Shape Language defined in \cite{nagpal2008programmable} is used to achieve similar goals of the Growing Point Language though focussing on programming a ``computational surface'', intended as a set of small devices working independently of their density in the surface: this language defines geometrical constructs to create basic regions and compose them, which could be turned into an API of \FSCAFI{}/\scafi{} blocks to be functionally composed to achieve similar complex geometrical structures.
\revised{
A more recent work is PLEIADES~\cite{DBLP:conf/dsn/BougetBLT18},
 which aim at expressing self-stabilising overlay structures in large-scale distributed systems; however, these goals are not achieved through programs as in \FSCAFI{}/\scafi{} but rather through shape formation protocols and equation-based specifications.
}
In general, due to the universal character of field computations \cite{abdv:universality}, one could consider \FSCAFI{}/\scafi{} as a viable implementation framework for a number of approaches to organise the shape of computational entities in a physical environment, with the additional byproduct of leveraging the theory of field computations to assess formal validity of certain properties, such as density independence as developed in \cite{BVPD-TAAS2017} or self-stabilisation in \cite{viroli:selfstabilisation}.

\subsubsection*{WSN-based discovery and streaming} A number of works originating in the context of information systems for sensor networks, such as TinyDB~\cite{tinydb}, Regiment~\cite{regiment}, and Cougar~\cite{Yao02thecougar}, address the problem of gathering information extracted from sensors in a given region of space, aggregating them somehow, and redirecting results over the devices in another region. They either focus on spatial query languages (Cougar), diffusion/aggregation policies (TinyDB) or functional models to express and manipulate streams of events (Regiment).
\revised{
A more recent approach is makeSense~\cite{DBLP:journals/tse/MottolaPOEFFGKM19},
 with leverages graphical models of business processes
 compiling to an imperative object-oriented macro-programming language (mPL)
 which supports
 many-to-one report actions
 and many-to-many collective actions, i.e.,
 actions possibly spanning multiple nodes.
}
Similar approaches for mobile ad-hoc networks (MANET) have been considered that do not use data-oriented techniques but rather focus on services: SpatialViews~\cite{ni2005manetspatialview} works by abstracting a MANET into \emph{spatial views} that can be iterated on to visit nodes and request services; AmbientTalk~\cite{van2014ambienttalk} is another language for MANETs that provides resilience against transient network partitions by automatically buffering sent messages. A clear advantage of functional-based field computations as supported in \FSCAFI{}/\scafi{} is that the various bricks of information collection, aggregation, diffusion, can be defined with the same language, wrapped in homogeneous components represented by \FSCAFI{}/\scafi{} functions, and composed to create more complex applications as exemplified in Section~\ref{sec-case}, while leveraging the discovery of nearby services through an actor-based runtime.

\subsubsection*{Distributed/parallel computing} Aggregate Computing and its incarnation in \scafi{} can evidently be considered as a declarative model for distributed programming, relying on abstractions and assumptions to make it easy to express certain kinds of programs by delegating important features (e.g., synthesising the concrete execution plan) to an underlying platform. As a notable example, we can draw a bridge with big data processing frameworks like MapReduce~\cite{dean2008mapreduce} and its derivation Apache Spark~\cite{zaharia2016apachespark}: they essentially provide a highly declarative language of stream processing (based on simple map, reduce, filter and fold operations) and delegate to the underlying platform the duty of breaking tasks into smaller chunks to be allocated by the available computational/storage resources. An approach which is execution-wise technically more similar to the one developed in this paper is given by Bulk-Synchronous Parallel (BSP)-inspired frameworks, such as the large-scale graph processing framework Apache Giraph~\cite{sakr2016giraph}, which defines computations in terms of transformation of large graphs of data, typically stored in a distributed database.
 In this respect, \FSCAFI{}/\scafi{} can be seen as relying on a similar approach to achieve a more complex goal, namely, that of declaratively specifying a dynamic collective behaviour (based on distributed field computation), ultimately digesting information coming from distributed sensors and producing instructions for actuators, in such a way so as to make it breakable into smaller pieces (single rounds of computations) allocated to each device in the network.

\subsubsection*{Service choreographies} Choreographies~\cite{peltz2003webserviceorchestrationchoreography} are an approach to service composition
 where the interaction protocol of collaborative workflows is specified by a global viewpoint.
So, they define the cooperative contract of multiple parties
 playing certain roles and collaborating to achieve a global goal.
There are strong similarities with Aggregate Computing; indeed, aggregate programs
 \emph{(i)} globally define the interactions supporting a collective computation, cf. the \texttt{nbr} construct;
 \emph{(ii)} define services which are fundamentally collaborative in nature;
 and \emph{(iii)} define roles for devices implicitly by the set of (sub-)computations executed by them or, equivalently, by the set of domain branches that they select in the collective workflow.
By contrast, however, while classical choreographies usually express goal-oriented workflows where the roles are few and statically assigned to parties,
 aggregate computations typically carry on continuous processes
 where devices repeatedly participate in the collective service
 and can also play different roles across time depending on the context.
Additionally, choreographies abstract from the particular services carried on by the involved parties,
 focussing instead only on when and what messages are exchanged,
 whereas aggregate programs specify \emph{computations},
 which are collective and yield global results represented as fields.
Moreover, interaction in field computations is only possible (by alignment)
 between entities playing the same (sub-)computation,
 and is neighbour-driven (rather than peer-to-peer and role-driven).
Finally, choreographies and choreographic programming~\revised{\cite{DBLP:journals/tcs/Cruz-FilipeM20}}
 mainly focus on checking conformance or building correct-by-construction concurrent programs (e.g., deadlock-free), but do little for functional composability of adaptive behaviours, which is instead the core of Aggregate Computing.

\subsubsection*{Spatial \revised{and field-based} computing languages}
Aggregate Computing directly descends from the class of so-called general-purpose spatial computing languages, all addressing the problem of engineering distributed (or parallel) computing by providing mechanisms to manipulate data structures diffused in space and evolving in time. Notable examples include the StarLisp approach for parallel computing \cite{starlisp}, the SDEF programming system inspired by systolic computing \cite{SDEF}, and topological computing with MGS~\cite{GiavittoMGS05}. They typically provide specific abstractions that significantly differ from that of computational fields: for instance, MGS defines computations over manifolds, the goal of which is to alter the manifold itself as a way to represent input-output transformation.

\revised{
A particular class of spatial computing languages is given by
 \emph{field-based coordination languages}~\cite{Viroli-et-al:JLAMP-2019}.
Beside the field calculus, covered in this paper,
 other works like TOTA~\cite{tota} and
 calculi like SMuC~\cite{DBLP:journals/corr/Lluch-LafuenteL16}
 also come under this umbrella.
SMuC, however,
 does not express fields as functional compositions,
 but as formulas on fixpoint steps
 whose result may serve as continuations of the program.
}
Specific programming languages to work with computational fields have been introduced as well, with the Proto~\cite{proto06a} programming language as common ancestor, \Protelis{}~\cite{Protelis15} as its Java-oriented DSL version,
 \revised{and the more recent FCPP~\cite{FCPP-ACSOS-2020}}. 
The field calculus, deeply described and compared to in Section~\ref{sec-properties}, has been designed as common core calculus for such languages, and for studying behavioural properties and semantic aspects \cite{BVPD-TAAS2017,viroli:selfstabilisation,audrito2019tocl}. Though rather similar to those languages, \FSCAFI{} 
evolved in a different way: its design is profoundly influenced by the need of smoothly integrating field computations in the syntactic, semantic, and typing structures of modern, conventional languages (like Scala---see \Cref{sec-NC-at-work}), and this required key semantic changes that motivated a more general and expressive calculus, as presented in this paper.

\section{Conclusions} \label{sec-conclusion}

%
Aggregate Computing is a recent paradigm
 for ``holistically'' engineering CASs and
 smart situated ecosystems,
 that aims to exploit, both functionally
 and operationally, the increasing computational capabilities of
 our environments---as fostered by driver scenarios like IoT, CPS, and smart cities.
It formally builds on computational fields and corresponding calculi
 to functionally compose macro behavioural specifications
 that capture, in a declarative way, the adaptive logic for
 turning local activity into global, resilient behaviour.
In order to promote conceptual frugality and foster smooth embedding of this programming model in mainstream languages,
 we propose a field calculus variant, called \FSCAFI{},
 which substitutes the notion of a ``neighbouring field''
 with a novel notion of a ``computation against a neighbour''.
We formalise \FSCAFI{} and thoroughly compare to the higher-order field calculus (HFC),
 stressing differences in expressiveness
 and identifying a common fragment
 allowing straightforward transfer of interesting properties
 such as self-stabilisation.
To witness the benefits of the novel calculus, we cover the \scafi{} aggregate programming language, which implements \FSCAFI{} as a DSL embedded in Scala.
Finally, we use a simulated case study in edge computing
 to show that \scafi{}/\FSCAFI{} are effectively expressive in practice.

The availability of a Scala-based implementation of field computations naturally suggests a number of future works.
From the linguistic viewpoint, it is interesting the study the interplay of field programs with advanced functional programming techniques, like the use of monads to structure the specification of increasingly complex field computations and field processes~\cite{casadei19processes,casadei2021eaai-procs}, or the use of implicit parameters to define common contexts for library functions, there included the ability of dynamically select the most proper implementation of building blocks for the application at hand \cite{viroli:selfstabilisation}.
At the \scafi{} platform level, instead, we plan to deeply investigate and implement techniques for infrastructure- and QoS-aware
adaptation of deployment and execution strategies for aggregate system execution, along the lines of~\cite{casadei2020pulverization}.
This is to determine suitable application partitioning schemas and build a monitoring and control plan to effectively carry out configuration transitions.
The case study we discussed in the paper also suggests the importance of collective, self-adaptive/self-organising techniques
for (decentralised) edge computing and volunteer computing in ad-hoc cloudlets.
In particular, a question relates to how much such edge computing systems
 need orchestration, what benefits and challenges
 can emerge from a hybrid approach that brings self-organising processes in,
 and how aggregate techniques allows us to tackle such integration of paradigms.

\appendix

\section{\HFC{} Operational Device Semantics} \label{apx:hfc:opsem}

We now present the \HFC{} operational device semantics, modelling computation of a device within one round, as developed in literature \cite{audrito2019tocl} and reported in \Cref{fig:hfc:deviceSemantics}. The network semantics of \HFC{} is not reported, since it is the same as that of \FSCAFI{} (cf.~\Cref{sec-calculus-network-semantics}). The \HFC{} device semantics is based on value-trees and value-tree environments, as that of \FSCAFI{}. The syntax of them, together with the auxiliary functions and syntactic shorthands, are not reported here since they are the same as in \FSCAFI{} (cf.~\Cref{fig:deviceSemantics}). The derived judgements are slightly different, and follow the form $\bsopsem{\deviceId}{\Trees}{\senstate}{\e}{\vtree}$, to be read ``expression $\e$ evaluates to  value-tree $\vtree$ on device $\deviceId$  with respect to the value-tree environment $\Trees$ and sensor state $\senstate$''.

\begin{figure}[!t]{
 \framebox[1\textwidth]{
 $\begin{array}{l}
\begin{array}{c}
\nullsurfaceTyping{E-LOC}{
\bsopsem{\deviceId}{\Trees}{{\senstate}}{\lvalue}{\mkvtree{E-LOC}{\lvalue}{}}
}
\qquad\qquad
\surfaceTyping{E-FLD}{\qquad \fvalue' = \proj{\fvalue}{\domof{\Trees}\cup\{\deviceId\}}}{
\bsopsem{\deviceId}{\Trees}{{\senstate}}{\fvalue}{\mkvtree{E-FLD}{\fvalue'}{}}
}
\skiptransition\\[-3pt]
\surfaceTyping{E-B-APP}{\quad
\begin{array}{ll}
	\bsopsem{\deviceId}{\piIof{1}{\Trees}}{\senstate}{\e}{\vtree}
	&
	\bsopsem{\deviceId}{\piIof{i+1}{\Trees}}{\senstate}{\e_i}{\vtree_i}  \quad \text{for all}\; i \in 1, \ldots, n
	\\
	\bname = \vrootOf{\vtree}
	&
	\anyvalue=\builtinop{{\bname}}{\deviceId}{\Trees{,\senstate}}(\vrootOf{\overline{\vtree}})
\end{array}
}{
	\bsopsem{\deviceId}{\Trees}{{\senstate}}{\e(\overline{\e})}{\mkvtree{E-B-APP}{\anyvalue}{\vtree,\overline{\vtree},\anyvalue}}
}
\skiptransition\\
\surfaceTyping{E-D-APP}{ \quad
\begin{array}{ll}
	\bsopsem{\deviceId}{\piIof{1}{\Trees}}{\senstate}{\e}{\vtree}
	&
	\bsopsem{\deviceId}{\piIof{i+1}{\Trees}}{\senstate}{\e_i}{\vtree_i}  \quad \text{for all}\; i \in 1, \ldots, n
	\\
	\funvalue = \vrootOf{\vtree} {\text{ is not a built-in}}
	&
	\bsopsem{\deviceId}{\piBof{\funvalue}{\Trees}}{{\senstate}}{\applySubstitution{\body{\funvalue}}{\substitution{\args{\funvalue}}{\vrootOf{\overline{\vtree}}}}}{\vtree'}
\end{array}
}{
	\bsopsem{\deviceId}{\Trees}{{\senstate}}{\e(\overline{\e})}{\mkvtree{E-D-APP}{\vrootOf{\vtree'}}{\vtree,\overline{\vtree},\vtree'}}
}
\skiptransition\\
\surfaceTyping{E-NBR}{
         \qquad
\Trees_1=\piIof{1}{\Trees}
\qquad
     \bsopsem{\deviceId}{\Trees_1}{{\senstate}}{\e}{\vtree_1}
\qquad
 \fvalue=\mapupdate{\vrootOf{\Trees_{1}}}{\envmap{\deviceId}{\vrootOf{\vtree_1}}}
 }{
\bsopsem{\deviceId}{\Trees}{{\senstate}}{\nbrK\{\e\}}{\mkvtree{E-NBR}{\fvalue}{\vtree_1}}
}
\skiptransition\\
\surfaceTyping{E-REP}{
	\begin{array}{ll}
     \bsopsem{\deviceId}{\piIof{1}{\Trees}}{{\senstate}}{\e_1}{\vtree_1} & \lvalue_1=\vrootOf{\vtree_{1}}\\
     \bsopsem{\deviceId}{\piIof{2}{\Trees}}{{\senstate}}{\applySubstitution{\e_2}{\substitution{\xname}{\lvalue_0}}}{\vtree_2}~~& \lvalue_2=\vrootOf{\vtree_{2}}
	\end{array}
	\quad
	\lvalue_0 = \left\{\begin{array}{lc}
                             \vrootOf{\piIof{2}{\Trees}}(\deviceId) & \mbox{if} \;  \deviceId \in \domof{\Trees} \\
                             \lvalue_1 & \mbox{otherwise}
                           \end{array}\right.
	\!\!\!\!
 }{
\bsopsem{\deviceId}{\Trees}{{\senstate}}{\repK(\e_1)\{(\xname) \; \toSymK{\e_2}\}}{\mkvtree{E-REP}{\lvalue_2}{\vtree_1,\vtree_2}}
}
\end{array}
\end{array}$}
}
\vspace{-0.1cm}
\caption{Big-step rules for \HFC{} expression evaluation.} \label{fig:hfc:deviceSemantics}
\end{figure}

Rules \ruleNameSize{[E-LOC]} and \ruleNameSize{[E-FLD]} model the evaluation of expressions that are either a local value or a neighbouring field value, respectively. For instance, evaluating the expression $\mathtt{1}$ produces (by rule  \ruleNameSize{[E-LOC]}) the value-tree $\mkvtree{E-LOC}{1}{}$, while evaluating the expression $\mathtt{+}$ produces the value-tree $\mkvtree{E-LOC}{+}{}$. Note that, in order to ensure that domain alignment is obeyed, rule \ruleNameSize{[E-FLD]} restricts the domain of the neighbouring field value $\fvalue$ to the domain of $\Trees$ augmented by $\deviceId$.

Rule \ruleNameSize{[E-B-APP]} models the application of built-in functions. It is used to evaluate expressions of the form $\e(\e_1 \cdots \e_n)$ such that the evaluation of $\e$ produces a value-tree $\vtree$ whose root $\rho(\vtree)$ is a built-in function $\bname$. It produces the value-tree $\mkvtree{E-B-APP}{\anyvalue}{\vtree,\vtree_{1},\ldots,\vtree_{n},\anyvalue}$, where  $\vtree_{1},\ldots,\vtree_{n}$ are the value-trees produced by the evaluation of the actual parameters $\e_{1},\ldots,\e_{n}$ ($n\ge 0$) and $\anyvalue$ is the value returned by the function.
Rule \ruleNameSize{[E-B-APP]}  exploits the special auxiliary function $\builtinop{\bname}{\deviceId}{\Trees, \senstate}$, whose actual definition is abstracted away, in order to allow for customised sets of built-in functions. This auxiliary function is such that $\builtinop{\bname}{\deviceId}{\Trees, \senstate}(\overline\anyvalue)$ computes the result of applying built-in function $\bname$ to values $\overline\anyvalue$ in the current environment and sensor state of device $\deviceId$. We require that $\builtinop{\bname}{\deviceId}{\Trees, \senstate}(\overline\anyvalue)$ always yields values of the expected type $\type$ where $\bname$ has a suitable type $(\overline{\type}) \to \type$.

Rule \ruleNameSize{[E-D-APP]} models the application of user-defined or anonymous functions, i.e., it is used to evaluate expressions of the form $\e(\e_1 \cdots \e_n)$ such that the evaluation of $\e$ produces a value-tree $\vtree$ whose root  $\funvalue=\vrootOf{\vtree}$ is a user-defined function name or an anonymous function value.
It is similar to rule \ruleNameSize{[E-B-APP]}, except for the last subtree $\vtree'$ of the result, which is produced by evaluating the body of the function $\funvalue$ with respect to the value-tree environment $\piBof{{\funvalue}}{\Trees}$ containing only the value-trees associated to the evaluation of functions with the same name as $\funvalue$.

The evaluation of rule \ruleNameSize{[E-REP]} depends on whether it is performed against a tree environment with or without $\deviceId$ in its domain. If it is present, $\lvalue_0$ is obtained from it (being the previously computed value for the $\repK$ construct), otherwise it is set to the result of $\e_1$. The evaluation concludes substituting $\lvalue_0$ for $\xname$ in the body of $\e_2$. Notice that this substitution corresponds to the result of applying $(\xname) \toSymK{\e_2}(\lvalue_0)$ according to rule \ruleNameSize{[E-D-APP]} (skipping some branches of the resulting value-tree).

Value-trees also support modelling information exchange through the $\nbrK$ construct, as of rule \ruleNameSize{[E-NBR]}.
In this rule, the neighbours' values for $\e$ are extracted into a neighbouring field value as $\fvalue = \vrootOf{\Trees_1}$. Then $\fvalue(\deviceId)$ is updated to the more recent value $\lvalue = \vrootOf{\vtree_1}$, as represented by the notation $\mapupdate{\fvalue}{\envmap{\deviceId}{\lvalue}}$.

\section{Proofs} \label{apx:proofs}

\subsection{Computation Determinism}\hfill\\ \label{apx:preservation}

\begin{replem}[Computation Determinism]{lem:completeness}
	Let $\e$ be a well-typed closed expression, $\Trees \in \coherentEnv{\e}{\type}{\LTStypEnv;\TtypEnv}$. Then for all device identifiers $\deviceId$, $\deviceId'$ and sensor state $\senstate$:
	\begin{enumerate}
		\item
		$\bsopsemFAIL{\deviceId,\deviceId}{\Trees}{\senstate}{\e}$ cannot hold.
		\item
		There is at most one derivation of the kind $\bsopsem{\deviceId,\deviceId'}{\Trees}{\senstate}{\e}{\vtree}$ or $\bsopsemFAIL{\deviceId,\deviceId'}{\Trees}{\senstate}{\e}$.
	\end{enumerate}
\end{replem}
\begin{proof}
	\begin{enumerate}
		\item \label{subfact:noneighbourok}
		Notice that a failure can occur only if in a certain subexpression (which is a relational sensor application or $\nbrK$-expression) the neighbour $\deviceId'$ is not in $\domof{\Trees} \cup \{\deviceId\}$, thus in particular $\deviceId' \neq \deviceId$. Any construct other than a $\foldK$ either propagates the neighbour $\deviceId'$ as is or resets it to $\deviceId$, so for a failure to occur in a sub-evaluation of $\bsopsemFAIL{\deviceId,\deviceId}{\Trees}{\senstate}{\e}$ the involved subexpression must be inside a $\foldK$-expression. However, $\foldK$-expressions never fail, as they ``absorb'' failures by skipping failing neighbours, concluding the proof.
		\item
		We assume that $\e$ has at least one possible derivation of the big-step operational semantics, and prove by induction on its length that this derivation is unique.

		If $\e$ is a value, a $\repK$- or $\foldK$-expression, then it cannot fail since the conclusions of every rule in Figure \ref{fig:deviceSemanticsFAIL} is either a function application or a $\nbrK$-expression. Furthermore, its evaluation produces a value by rules \ruleNameSize{[E-VAL]}, \ruleNameSize{[E-REP]} or \ruleNameSize{[E-FOLD]} (respectively), using fact (\ref{subfact:noneighbourok}) and observing that the only sub-evaluation which may fail occur in the third argument of a $\foldK$, and does not influence the evaluation of the $\foldK$ itself. Such value must be unique, since by inductive hypothesis so are the results of its subexpressions and rules are deterministic.

		If $\e$ is an $\nbrK$-expression, then exactly one of the rules \ruleNameSize{[E-NBR]}, \ruleNameSize{[E-NBR-LOC]} or \ruleNameSize{[E-NBR-FAIL]} must be applicable, depending on whether resp.~$\deviceId' = \deviceId$, $\deviceId' \in \domof{\Trees} \setminus \deviceId$, or neither holds. In the first two cases, the result is then unique by inductive hypothesis and rule determinism.

		Finally, if $\e = \e'(\overline\e)$ is a function application, then the evaluation of $\e'$ does not fail by fact (\ref{subfact:noneighbourok}), and thus produces a uniquely determined function $\funvalue$ (with a possibly infinite derivation) by inductive hypothesis. If the evaluation of some of the arguments fails, then rule \ruleNameSize{[E-APP-ARG-FAIL]} applies and $\e$ fails, and none of the other function application rules (\ruleNameSize{[E-B-APP]}, \ruleNameSize{[E-D-APP]}, \ruleNameSize{[E-R-APP-FAIL]}, \ruleNameSize{[E-D-APP-FAIL]}) applies since they all require all arguments to evaluate to a value.

		Otherwise, $\overline\e$ evaluates (uniquely) to $\overline\anyvalue$. If $\funvalue = \snvalue$ is a relational sensor and $\deviceId'$ is not in $\domof{\piBof{\snvalue}{\Trees}} \cup \{\deviceId\}$, then $\e$ fails by Rule \ruleNameSize{[E-R-APP-FAIL]} and does not produce a value since $\builtinop{\snvalue}{\deviceId,\deviceId'}{\piBof{\snvalue}{\Trees},\senstate}$ is undefined making rule \ruleNameSize{[E-B-APP]} inapplicable. If instead $\funvalue$ is a built-in function which is not a relational sensor or $\deviceId'$ is in $\domof{\piBof{\funvalue}{\Trees}} \cup \{\deviceId\}$, then $\builtinop{\funvalue}{\deviceId,\deviceId'}{\piBof{\funvalue}{\Trees},\senstate}$ is defined and rule \ruleNameSize{[E-B-APP]} is applicable while rule \ruleNameSize{[E-R-APP-FAIL]} is not. The result is then unique by inductive hypothesis and rule determinism.

		The only remaining case is when $\funvalue$ is a user-defined or anonymous function. By inductive hypothesis, $\applySubstitution{\body{\funvalue}}{\substitution{\args{\funvalue}}{\vrootOf{\overline{\vtree}}}}$ either fails or it produces a unique value. In the first case, rule \ruleNameSize{[E-D-APP]} applies (producing an unique value) while rule \ruleNameSize{[E-D-APP-FAIL]} does not; in the second case the converse happens, concluding that derivations are unique (when present).
%
%
		  \qedhere
	\end{enumerate}
\end{proof}

\subsection{Calculus Equivalence}\hfill\\ \label{apx:equivalence}

\begin{repthm}[Equivalence between \HFCprime{} and  \FSCAFIi{}]{thm:equivalence}
	Assume that for every built-in function $\bname$ with local arguments, $\bname(\overline\xname)$ has the same behaviour in \HFCprime{} and \FSCAFIi{} for every substitution of $\overline\xname$. Assume that $\expTypJud{\LTStypEnv}{\TtypEnv}{\e}{\type}$ in \HFCprime{}/\FSCAFIi{} where $\TtypEnv = \overline\xname : \overline\type$, let $\overline\anyvalue$ be \HFCprime{} values of type $\overline\type$. Then $\applySubstitution{\e}{\substitution{\overline\xname}{\overline\anyvalue}}$ has the same behaviour in \HFCprime{} and \FSCAFIi{}.
\end{repthm}
\begin{proof}
	We proceed by induction on the syntax of closed expressions $\e$, simultaneously for all possible network environments.
	\begin{itemize}
		\item
        $\e = \anyvalue$: since \FSCAFI{} rule \ruleNameSize{[E-VAL]} is identical to \HFC{} rule \ruleNameSize{[E-LOC]}, the thesis holds whenever $\e = \anyvalue$ is a local value. Since field values are not allowed to appear in source \HFC{} programs, and do not exist in \FSCAFI{} programs, the thesis follows.

        \item
        $\e = \xname$: if $\type$ is a local type, the same reasoning for $\e = \anyvalue$ applies. If $\type$ is a field type and $\xname$ is substituted with $\fvalue$, the evaluation result is $\proj{\fvalue}{\domof{\Trees}}$ in \HFC{}. In \FSCAFI{}, for every $\deviceId \in \domof{\Trees}$, $\applySubstitution{\e}{\substitution{\overline\xname}{\overline\anyvalue(\deviceId)}}$ is $\fvalue(\deviceId)$, which is a local value hence it evaluates to itself by rule \ruleNameSize{[E-VAL]}, concluding this part of the proof.

        \item
        $\e = \e_{n+1}(\overline\e)$: notice that rules \ruleNameSize{[E-B-APP]} and \ruleNameSize{[E-D-APP]} are identical in \HFC{} and \FSCAFI{} (ignoring the neighbour device $\deviceId'$). This concludes the proof in case all $\overline\e$ and the resulting output of $\e_{n+1}$ are local. Otherwise, there are three possibilities:
        \begin{itemize}
            \item $\e_{n+1}$ evaluates to a built-in function, so that the thesis follows by inductive hypothesis and the coherence hypothesis on built-in functions.
            \item $\e_{n+1}$ evaluates to a user-defined or anonymous function $\fname$. In this case, the domain of computation is restricted both in \HFC{} and in \FSCAFI{} to the aligned neighbours $\piBof{\fname}{\Trees}$; which in fact reduces all fields involved in the computation to the new domain, both in \HFC{} by Rule \ruleNameSize{[E-FLD]} and in \FSCAFI{} by Rule \ruleNameSize{[E-FOLD]} (which aggregates only values from devices in $\domof{\Trees}$).

            After alignment, in both cases the argument's values are substituted into the body of $\fname$. If all arguments have local type, the substitution is performed in the same way both in \HFC{} and \FSCAFI{}, hence the result of $\e$ corresponds by inductive hypothesis.
            If some argument has field type, then those arguments cannot occur within folding expressions, and are either ignored or manipulated point-wise to form a field result. In this case, the results of $\e$ for every neighbour $\deviceId'$ correspond to the point-wise results of field in \HFC{}, concluding the proof in this case.
        \end{itemize}

        \item
        $\e = \repK(\e_1)\{(\xname) \toSymK{\e_2}\}$: the thesis follows (with a further induction on firing events) by noticing that rule \ruleNameSize{[E-REP]} in \HFC{} corresponds to that of \FSCAFI{} together with the expansion of the anonymous function application.

        \item
        $\e = \nbrK\{\e_1\}$: by inductive hypothesis, $\e_1$ is an expression of local type which always evaluates to the same trees in \HFC{} and \FSCAFI{}. It follows by \HFC{} rule \ruleNameSize{[E-NBR]} that $\e$ in $\deviceId$ evaluates to the field mapping $\overline\deviceId$ to the corresponding values $\overline\anyvalue$ of $\e_1$ in those devices. By \FSCAFI{} rules \ruleNameSize{[E-NBR]} (for $\deviceId_i \neq \deviceId$) and \ruleNameSize{[E-NBR-LOC]} (for $\deviceId_i = \deviceId$) it follows that $\bsopsem{\deviceId,\deviceId_i}{\Trees}{\senstate}{\e}{\mkvtree{}{\anyvalue_i}{\overline\vtree}}$ for all $i = 1 \ldots n$.

        \item
        $\e = \foldK(\e_1, \e_2, \e_3)$: by inductive hypothesis, $\e_1$ and $\e_2$ evaluate to the same local values $\anyvalue'$, $\funvalue$ while $\e_3$ does not contain free variables of field type and it satisfies that
        \[
        \bsopsem{\deviceId}{\Trees}{\senstate}{\e_3}{\mkvtree{}{(\envmap{\overline\deviceId}{\overline{\anyvalue}})}{\overline\vtree}} \iff \bsopsem{\deviceId,\overline\deviceId}{\Trees}{\senstate}{\e_3}{\mkvtree{}{\overline\anyvalue}{\overline\vtree}}
        \]
        By \FSCAFI{} rule \ruleNameSize{[E-FOLD]}, the overall result of the expression is then $\funvalue(\anyvalue', \overline\anyvalue)$ (computed two elements at a time), which corresponds to the result of applying $\foldK$ in \HFC{}. \qedhere
	\end{itemize}
\end{proof}

\section*{Acknowledgment}

This work has been partially supported by the EU/MUR FSE REACT-EU PON R\&I 2014-2020.

\bibliographystyle{alphaurl}
\bibliography{long}

\end{document}